\documentclass{eptcs}
\newfont{\fsc}{eusm10}
\usepackage{theorem}\theorembodyfont{\rmfamily}
\usepackage{latexsym}                   % load symbols like \Box                
\usepackage{amsmath}                    % \boldsymbol
\usepackage{dsfont}
\usepackage{epsfig}
\usepackage{amssymb}
\usepackage{color}

\newcommand{\plat}[1]{\raisebox{0pt}[0pt][0pt]{#1}} % no vertical space
\newcommand{\nietplat}[1]{#1} % normal vertical space
\newcommand{\diam}[1]{\langle#1\rangle}
\newcommand{\eps}{\langle\epsilon\rangle}
\newcommand{\trans}[1]{\stackrel{#1}{\longrightarrow}}
\newcommand{\ntrans}[1]{\mathrel{\hspace{.4em}\not\hspace{-.4em}\trans{#1\;}}}
\newcommand{\epsarrow}{\stackrel{\epsilon}{\Longrightarrow}}
\newcommand{\bis}[1]{\,\, \raisebox{.3ex}{$\underline{\makebox[.7em]{$\leftrightarrow$}}$}\,_{#1}\,}
\newcommand{\notbis}[1]{\not\mathrel{\raisebox{.3ex}{$\underline{\makebox[.7em]{\,$\leftrightarrow$}}$}\,_{#1}\!}}
\newcommand{\var}{{\it var}}
\newcommand{\ar}{{\it ar}}
\newcommand{\IO}[1]{\mathbb{O}_{#1}}
\renewcommand{\phi}{\varphi}
\newcommand{\Brac}[1]{[#1]}
\newcounter{saveenumi}
\newcommand{\valuation}{closed substitution}
\newcommand{\aL}{\aleph\,\mathord\cap\Lambda}
\newcommand{\qed}{\hfill$\Box$}
\newtheorem{defi}{Definition}
\newtheorem{theo}{Theorem}
\newtheorem{prop}{Proposition}
\newtheorem{lemm}{Lemma}
\newtheorem{coro}{Corollary}
\newtheorem{exam}{Example}
\newtheorem{rema}{Remark}

\newenvironment{definition}{\begin{defi} \rm }{\end{defi}}
\newenvironment{theorem}{\begin{theo}}{\end{theo}}
\newenvironment{proposition}{\begin{prop}}{\end{prop}}
\newenvironment{lemma}{\begin{lemm}}{\end{lemm}}
\newenvironment{corollary}{\begin{coro}}{\end{coro}}
\newenvironment{example}{\begin{exam} \rm }{\end{exam}}
\newenvironment{remark}{\begin{rema}}{\end{rema}}
\newenvironment{proof}{\begin{trivlist} \item[\hspace{\labelsep}\bf Proof:]}{\end{trivlist}}
\begin{document}

\title{\texorpdfstring{Divide and Congruence II:\\
From Decomposition of Modal Formulas\\ to Preservation of
Delay and Weak Bisimilarity}
{Divide and Congruence II:
From Decomposition of Modal Formulas
to Preservation of
Delay and Weak Bisimilarity}}
\def\titlerunning{Divide and Congruence II}
\def\authorrunning{W.J. Fokkink \& R.J. van Glabbeek}

\author{Wan Fokkink
\institute{Department of Computer Science\\
Vrije Universiteit Amsterdam, The Netherlands}
\email{w.j.fokkink@vu.nl}
\and
Rob van Glabbeek
\institute{NICTA\thanks{NICTA is funded by the Australian Government
    through the Department of Communications and the Australian
    Research Council through the ICT Centre of Excellence
    Program.}\;\,/Data61, CSIRO, Sydney, Australia}
\institute{School of Computer Science and Engineering\\
University of New South Wales, Sydney, Australia}
\email{rvg@cs.stanford.edu}
}

\providecommand{\publicationstatus}{An extended abstract of this paper appeared as \cite{FvG16}.}
\maketitle

\begin{abstract}
Earlier we presented a method to decompose modal formulas for processes with the internal action $\tau$, and congruence formats for branching and $\eta$-bisimilarity were derived on the basis of this decomposition method. The idea is that a congruence format for a semantics must ensure that the formulas in the modal characterisation of this semantics are always decomposed into formulas that are again in this modal characterisation. In this follow-up paper the decomposition method is enhanced to deal with modal characterisations that contain a modality $\eps\diam{a}\phi$, to derive congruence formats for delay and weak bisimilarity.
\end{abstract}

\section{Introduction}\label{sec:introduction}

In \cite{BFvG04} a method was developed to generate congruence formats for (concrete) process semantics from their modal characterisation. It crosses the borders between process algebra, structural operational semantics, process semantics, and modal logic. Cornerstone is the work in \cite{LL91} to decompose formulas from Hennessy-Milner logic \cite{HM85} with respect to a structural operational semantics in the De Simone format \cite{Sim85}. It was extended to the ntyft format \cite{Gro93} without lookahead in \cite{BFvG04}, and to the tyft format \cite{GV92} in \cite{FvGdW03}.

An equivalence is a congruence for a given process algebra or programming language if the equivalence class of a term $f(p_1,\ldots,p_n)$ is determined by the function $f$ and the equivalence classes of its arguments $p_1,\ldots,p_n$. Being a congruence is an important property, for instance to fit a process semantics into an axiomatic framework. A wide range of syntactic formats for structural operational semantics have been developed for several process semantics, to ensure that such a semantics is a congruence; notably for unrooted and rooted weak bisimilarity in \cite{Blo95} and for unrooted and rooted delay bisimilarity in \cite{vGl11}. These formats are contained in the positive GSOS format \cite{BIM95}. They include so-called patience rules for arguments $i$ of function symbols $f$, which imply that any term $f(p_1,\ldots,p_n)$ inherits the $\tau$-transitions of its argument $p_i$.
\advance\textheight 1pt

Key idea in \cite{BFvG04} is that a congruence format for a process semantics must ensure that the formulas in a modal characterisation of this semantics are always decomposed into formulas that are again in this modal characterisation. This yielded congruence formats for all known concrete (i.e., $\tau$-free) process semantics in a convenient way. Moreover, the resulting congruence formats are more elegant and expressive than existing congruence formats for individual process semantics. In \cite{FvGdW12} this method was extended to weak process semantics, which take into account the internal action $\tau$. As a result, congruence formats for rooted branching and $\eta$-bisimilarity were derived. These formats use two predicates $\aleph$ and $\Lambda$ on arguments of function symbols: $\aleph$ marks processes that can execute immediately, and $\Lambda$ marks processes that have started executing (but may currently be unable to execute). Formats for unrooted branching and $\eta$-bisimilarity were obtained by imposing one extra restriction on top of the format for the corresponding rooted semantics: $\Lambda$ holds universally.

\advance\textheight -1pt
The framework from \cite{FvGdW12} covers only a small part of the spectrum of weak semantics from \cite{vGl93}. In particular, it does not readily extend to delay and weak bisimilarity \cite{Mil81,Mil89}. The reason is that in the definition of these semantics, in contrast to branching and $\eta$-bisimilarity, a process $q$ that mimics an $a$-transition from a process $p$, does not need to be related to $p$ at the moment that $q$ performs the $a$-transition. This implies that in the modal characterisation of delay and weak bisimilarity, a modality $\diam{a}\phi$ stating that an $a$-transition to a process where $\phi$ holds, is always preceded by a modality $\eps$ allowing any number of $\tau$-transitions. As a consequence, devising congruence formats for delay and weak bisimilarity has been notoriously difficult, see e.g.\ \cite{Blo95,vGl11}. Here we show how this technical obstacle can be overcome by means of the semantic notion of \emph{delay resistance}, which generalises earlier notions from \cite{Blo95,vGl11}. This notion ensures that modalities $\eps\diam{a}\phi$ are decomposed into formulas that again have this form. Thus congruence formats can be derived for semantics with a modal characterisation containing such modalities. We derive congruence formats for rooted delay and weak bisimilarity. The congruence formats for the unrooted counterparts of these semantics are again obtained by the extra requirement that $\Lambda$ must be universal. We moreover provide syntactic restrictions which imply delay resistance, leading to the first entirely syntactic congruence formats for rooted delay and weak bisimilarity.

In \cite{vGl11} a general method is presented to generalise any congruence format $F$, contained in
the GSOS format, into a \emph{two-tiered} version of $F$. Two-tiered formats distinguish so-called
``principal'' function symbols and ``abbreviations''.  The latter can be regarded as syntactic
sugar, adding nothing that could not be expressed with principal function symbols.  The original
format $F$ is essentially the restriction of its two-tiered version that allows principal function
symbols only.  As shown in \cite{vGl11}, the general formats of \cite{Blo95,vGl11} can be obtained
as the two-tiered versions of the simplified formats from \cite{Blo95,vGl11}.  In \cite{FvGdW12}
this two-tiered approach was generalised from GSOS to decent ntyft format.  Consequently, the
two-tiered versions of the congruence formats presented in the current paper are again congruence
formats for rooted/unrooted delay/weak bisimilarity. These two-tiered versions of our formats (or
more precisely, of the intersection of our formats with the decent ntyft format) generalise the full
formats of \cite{Blo95,vGl11}. Ulidowski \cite{Uli92,UP02,UY00} proposed congruence formats for weak
semantics with a different treatment of divergence, which interestingly allow the inclusion of the
priority operator; (divergence-insensitive) rooted weak bisimilarity is not a congruence for this operator.
The TSSs of BPA$_{\epsilon\delta\tau}$, binary Kleene star and deadlock testing in Sect.~\ref{sec:applications}
are however outside those formats.

This research line shows that it is worthwhile to study the interplay of structural operational semantics and modal logic. The modal characterisation of a process semantics turns out to be fundamental for its congruence properties. Although some rather heavy technical machinery is needed to set the scene, especially the derivation of so-called ruloids and the decomposition method for modal formulas, the bulk of this work can be reused for the development of congruence formats for other weak process semantics. This is witnessed by the fact that the congruence results for rooted delay and weak bisimilarity are obtained in an almost identical fashion, and build upon the congruence proofs for rooted branching bisimilarity in \cite{FvGdW12}. Furthermore, the congruence formats that we obtain here are more liberal and more elegant than existing congruence formats for these semantics. In particular, in \cite{Blo95} it is stated that the RWB format put forward in that paper has a ``horrible definition''. In \cite{Blo95} it is moreover stated that ``negative rules seem incompatible with weak process equivalences.'' Here we show how negative premises can be included in congruence formats for rooted delay and weak bisimilarity.

The paper is structured as follows. Sect.~\ref{sec:preliminaries} contains technical preliminaries. Sect.~\ref{sec:delay-resistance} introduces the notion of delay resistance and explains how the decomposition method of modal formulas from \cite{FvGdW12} needs to be adapted. Sect.~\ref{sec:congruence} presents the congruence formats for rooted delay and weak bisimilarity and the proofs of these congruence results. Sect.~\ref{sec:checking-delay-resistance} shows that it is sufficient to check delay resistance for a limited set of variables (to be precise, the $\Lambda$-frozen arguments of a source), and how the semantic notion of delay resistance can be captured by means of syntactic criteria. Sect.~\ref{sec:applications} provides applications of our congruence formats. Sect.~\ref{sec:conclusions} concludes the paper.

\section{Preliminaries}
\label{sec:preliminaries}

This section recalls the basic notions of labelled transition systems
and weak semantics (Sect.~\ref{sec:equivalences_terms}), and presents
modal characterisations of the semantic equivalences that are studied
in this paper (Sect.~\ref{sec:modal}). Then follows a brief
introduction to structural operational semantics and the notion of a
well-supported proof (Sect.~\ref{sec:sos}).
Next we recall some syntactic restrictions on transition rules (Sect.~\ref{sec:ntytt}).
Then we present the notion of patience rules (Sect.~\ref{sec:predicates}),
and a basic result from \cite{BFvG04}, Prop.~\ref{prop:ruloid}, regarding ruloids (Sect.~\ref{sec:ruloids}).
Sect.~\ref{sec:linearity} shows that in Prop.~\ref{prop:ruloid} we may restrict attention to
  ruloids with so-called linear proofs.
Finally, we recall from \cite{FvGdW12} a method for decomposition of modal formulas (Sect.~\ref{sec:decomposition}).

\subsection{Equivalences on labelled transition systems}\label{sec:equivalences_terms}

A \emph{labelled transition system (LTS)} is a pair $(\mathbb{P},\rightarrow)$, with $\mathbb{P}$ a set of \emph{processes} and $\rightarrow\;\subseteq\mathbb{P}\times(A\cup\{\tau\})\times\mathbb{P}$, where $\tau$ is an \emph{internal action} and $A$ a set of \emph{concrete actions} not containing $\tau$.
We use $p,q$ to denote processes, $\alpha,\beta,\gamma$ for elements of $A\cup\{\tau\}$, and $a,b$ for elements of $A$. We write \plat{$p\trans\alpha q$} for $(p,\alpha,q)\in\;\rightarrow$ and \plat{$p\ntrans\alpha$} for \plat{$\neg(\exists q\in\mathbb{P}:p\trans\alpha q)$}. Furthermore, \plat{$\epsarrow$} denotes the transitive-reflexive closure of \nietplat{$\trans{\tau}$}.

Processes can be distinguished from each other by a wide range of semantics, based on e.g.\ branching structure or decorated versions of execution sequences. Van Glabbeek \cite{vGl93} classified so-called weak semantics, which take into account the internal action $\tau$. Here we focus on two such equivalences which, to different degrees, abstract away from internal actions: delay bisimilarity \cite{Mil81} and weak bisimilarity \cite{Mil89}. They are the two weak semantics that are employed most widely in the literature.

\begin{definition}\label{def:bb}
{\rm
Let $B\subseteq\mathbb{P}\times\mathbb{P}$ be a symmetric relation.
\begin{itemize}
\item
$B$ is a \emph{delay bisimulation} if $pBq$ and \plat{$p\trans{\alpha}p'$} implies that either $\alpha = \tau$ and $p'\,B\,q$, or \plat{$q\epsarrow \trans{\alpha} q''$} for some $q''$ with $p'Bq''$.

Processes $p,q$ are {\em delay bisimilar}, denoted $p\bis{d}q$, if there exists a delay bisimulation $B$ with $pBq$.

\item
$B$ is a \emph{weak bisimulation} if $pBq$ and \plat{$p\mathbin{\trans{\alpha}}p'$} implies that either $\alpha \mathbin= \tau$ and $p'\,B\,q$, or \plat{$q\mathbin{\epsarrow \trans{\alpha} \epsarrow} q''$} for some $q''$ with $p'Bq''$.

Processes $p,q$ are {\em weakly bisimilar}, denoted $p\bis{w}q$, if there exists a weak bisimulation $B$ with $pBq$.
\end{itemize}
}
\end{definition}
The notions of delay and weak bisimilarity were originally both introduced by Milner under the name ``observation equivalence''. Clearly, delay bisimilarity is included in weak bisimilarity.

It is well-known that delay and weak bisimilarity constitute equivalence relations \cite{Mil81,Mil89}. However, these two semantics are not {\em congruences} for most process algebras from the literature, meaning that the equivalence class of a process $f(p_1,\ldots,p_n)$, with $f$ an $n$-ary function symbol, is not always determined by the equivalence classes of its arguments, i.e.\ the processes $p_1,\ldots,p_n$. Rooted counterparts of these equivalences were introduced, which require for the pair of initial states that a $\tau$-transition needs to be matched by at least one $\tau$-transition. Unlike the unrooted versions they are congruences for basic process algebras, notably for the alternative composition operator.

\begin{definition}\label{def:rbb}
{\rm
Let $R\subseteq\mathbb{P}\times\mathbb{P}$ be a symmetric relation.
\begin{itemize}
\item
$R$ is a \emph{rooted delay bisimulation} if $pRq$ and \plat{$p \trans{\alpha} p'$} implies that \plat{$q \epsarrow\trans{\alpha}q '$} for some $q'$ with $p'\bis{d}q'$.

Processes $p,q$ are \emph{rooted delay bisimilar}, denoted $p\bis{rd}q$, if there exists a rooted delay bisimulation $R$ with $pRq$.

\item
$R$ is a \emph{rooted weak bisimulation} if $pRq$ and \plat{$p \!\trans{\alpha}\! p'$} implies that \plat{$q \!\epsarrow\trans{\alpha}\epsarrow\! q '$} for some $q'$ with $p'\bis{w}q'$.

Processes $p,q$ are \emph{rooted weakly bisimilar}, denoted $p\bis{rw}q$, if there exists a rooted weak bisimulation $R$ with $pRq$.
\end{itemize}
}
\end{definition}

\begin{example}\label{ex:weak bisimulations}
The processes $p_0$ and $p_1$ in the following LTS are rooted delay bisimilar but not $\eta$-bisimilar.
The idea is that in an $\eta$-bisimulation the transition $p_0\trans{b}0$ cannot be mimicked by $p_1$; the only candidate $p_1\trans{\tau}q\trans{b}0$ fails because
$q$ cannot be related to $p_0$, while this would be required for an $\eta$-bisimulation.

\vspace{4mm}

\centerline{\begin{picture}(0,0)%
\includegraphics{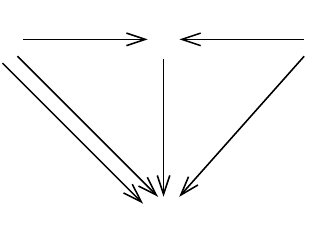}%
\end{picture}%
\setlength{\unitlength}{3947sp}%
\begingroup\makeatletter\ifx\SetFigFont\undefined%
\gdef\SetFigFont#1#2#3#4#5{%
  \reset@font\fontsize{#1}{#2pt}%
  \fontfamily{#3}\fontseries{#4}\fontshape{#5}%
  \selectfont}%
\fi\endgroup%
\begin{picture}(1508,1144)(4616,-2666)
\put(5461,-2086){\makebox(0,0)[lb]{\smash{{\SetFigFont{10}{12.0}{\rmdefault}{\mddefault}{\updefault}{\color[rgb]{0,0,0}$b$}%
}}}}
\put(5821,-2236){\makebox(0,0)[lb]{\smash{{\SetFigFont{10}{12.0}{\rmdefault}{\mddefault}{\updefault}{\color[rgb]{0,0,0}$a$}%
}}}}
\put(5089,-2086){\makebox(0,0)[lb]{\smash{{\SetFigFont{10}{12.0}{\rmdefault}{\mddefault}{\updefault}{\color[rgb]{0,0,0}$a$}%
}}}}
\put(4876,-2236){\makebox(0,0)[rb]{\smash{{\SetFigFont{10}{12.0}{\rmdefault}{\mddefault}{\updefault}{\color[rgb]{0,0,0}$b$}%
}}}}
\put(5401,-1741){\makebox(0,0)[b]{\smash{{\SetFigFont{10}{12.0}{\rmdefault}{\mddefault}{\updefault}{\color[rgb]{0,0,0}$q$}%
}}}}
\put(4951,-1645){\makebox(0,0)[b]{\smash{{\SetFigFont{10}{12.0}{\rmdefault}{\mddefault}{\updefault}{\color[rgb]{0,0,0}$\tau$}%
}}}}
\put(4671,-1741){\makebox(0,0)[rb]{\smash{{\SetFigFont{10}{12.0}{\rmdefault}{\mddefault}{\updefault}{\color[rgb]{0,0,0}$p_0$}%
}}}}
\put(6109,-1741){\makebox(0,0)[lb]{\smash{{\SetFigFont{10}{12.0}{\rmdefault}{\mddefault}{\updefault}{\color[rgb]{0,0,0}$p_1$}%
}}}}
\put(5401,-2611){\makebox(0,0)[b]{\smash{{\SetFigFont{10}{12.0}{\rmdefault}{\mddefault}{\updefault}{\color[rgb]{0,0,0}$0$}%
}}}}
\put(5776,-1645){\makebox(0,0)[b]{\smash{{\SetFigFont{10}{12.0}{\rmdefault}{\mddefault}{\updefault}{\color[rgb]{0,0,0}$\tau$}%
}}}}
\end{picture}%
}

\vspace{4mm}

\noindent
The processes $p_0$ and $p_1$ in the following LTS are rooted weakly bisimilar but not delay bisimilar.
The idea is that in a delay bisimulation the transition $p_0\trans{a}0$ cannot be mimicked by $p_1$; the only candidate $p_1\trans{a}q\trans{\tau}0$ fails because
$q$ cannot be related to $0$, while this would be required for a delay bisimulation.

\vspace{4mm}

\centerline{\begin{picture}(0,0)%
\includegraphics{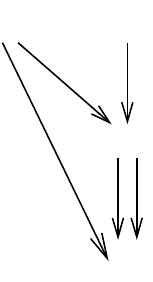}%
\end{picture}%
\setlength{\unitlength}{3947sp}%
\begingroup\makeatletter\ifx\SetFigFont\undefined%
\gdef\SetFigFont#1#2#3#4#5{%
  \reset@font\fontsize{#1}{#2pt}%
  \fontfamily{#3}\fontseries{#4}\fontshape{#5}%
  \selectfont}%
\fi\endgroup%
\begin{picture}(732,1342)(2764,-1431)
\put(2821,-236){\makebox(0,0)[b]{\smash{{\SetFigFont{10}{12.0}{\rmdefault}{\mddefault}{\updefault}{\color[rgb]{0,0,0}$p_0$}%
}}}}
\put(3376,-801){\makebox(0,0)[b]{\smash{{\SetFigFont{10}{12.0}{\rmdefault}{\mddefault}{\updefault}{\color[rgb]{0,0,0}$q$}%
}}}}
\put(3076,-436){\makebox(0,0)[lb]{\smash{{\SetFigFont{10}{12.0}{\rmdefault}{\mddefault}{\updefault}{\color[rgb]{0,0,0}$a$}%
}}}}
\put(3376,-1376){\makebox(0,0)[b]{\smash{{\SetFigFont{10}{12.0}{\rmdefault}{\mddefault}{\updefault}{\color[rgb]{0,0,0}$0$}%
}}}}
\put(3376,-236){\makebox(0,0)[b]{\smash{{\SetFigFont{10}{12.0}{\rmdefault}{\mddefault}{\updefault}{\color[rgb]{0,0,0}$p_1$}%
}}}}
\put(3439,-436){\makebox(0,0)[lb]{\smash{{\SetFigFont{10}{12.0}{\rmdefault}{\mddefault}{\updefault}{\color[rgb]{0,0,0}$a$}%
}}}}
\put(3481,-1036){\makebox(0,0)[lb]{\smash{{\SetFigFont{10}{12.0}{\rmdefault}{\mddefault}{\updefault}{\color[rgb]{0,0,0}$b$}%
}}}}
\put(3301,-1036){\makebox(0,0)[rb]{\smash{{\SetFigFont{10}{12.0}{\rmdefault}{\mddefault}{\updefault}{\color[rgb]{0,0,0}$\tau$}%
}}}}
\put(2926,-811){\makebox(0,0)[rb]{\smash{{\SetFigFont{10}{12.0}{\rmdefault}{\mddefault}{\updefault}{\color[rgb]{0,0,0}$a$}%
}}}}
\end{picture}%
}

\end{example}

\noindent
Our main aim is to develop congruence formats for both the rooted and the unrooted
versions of the two weak semantics defined in this section. These congruence formats will
impose syntactic restrictions on the transition rules (see Sect.~\ref{sec:sos}) that are
used to generate the underlying LTS. The congruence formats will be determined using the
characterising modal logics for these two weak semantics, which are presented in the next section.

\subsection{Modal logic}\label{sec:modal}

Behavioural equivalences can be characterised in terms of the observations that an experimenter could make during a session with a process. Modal logic captures such observations, with the aim to formulate properties of processes in an LTS. Following \cite{vGl93}, we extend Hennessy-Milner logic \cite{HM85} with the modal connective $\eps\phi$, expressing that a process can perform zero or more $\tau$-transitions to a process where $\phi$ holds.

\begin{definition}\label{def:formulas}
{\rm
The class $\mathbb{O}$ of \emph{modal formulas} is defined as follows, where $I$ ranges over all index sets:
\[
\mathbb{O} \qquad \phi ::= \mbox{$\bigwedge$}_{i\in I}\phi_i ~|~ \neg\phi ~|~ \diam{\alpha}\phi ~|~ \eps\phi\enspace .
\]
}
\end{definition}

\noindent
$p\models\phi$ denotes that $p$ satisfies $\phi$. By definition, $p\models\diam{\alpha}\phi$ if
\plat{$p\trans{\alpha}p'$} for some $p'$ with $p'\models\phi$, and $p\models\eps\phi$ if
\plat{$p\epsarrow p'$} for some $p'$ with $p'\models\phi$. We use abbreviations $\top$ for the empty
conjunction and $\phi_1\land\phi_2$ for $\bigwedge_{i\in\{1,2\}}\phi_i$.
We write $\phi\equiv\phi'$ if $p\models\phi\Leftrightarrow
p\models\phi'$ for any process $p$ in any LTS.

A modal characterisation of an equivalence on processes consists of a class $C$ of modal formulas such that two processes are equivalent if and only if they satisfy the same formulas in $C$. Hennessy-Milner logic is a modal characterisation of bisimilarity. We now introduce modal characterisations for (unrooted and rooted) delay and weak bisimilarity.

\begin{definition}\label{def:modal-characterisations}
{\rm
The subclasses $\IO{e}$ and $\IO{re}$ of $\mathbb{O}$, for $e\in\{d,w\}$, are defined as follows:
\begin{align*}
&\IO{d}       && \phi ::= \mbox{$\bigwedge$}_{i\in I}\phi_i ~|~ \neg\phi ~|~ \eps\phi ~|~ \eps\diam{a}\phi \\
&\IO{rd}      && \phi ::= \mbox{$\bigwedge$}_{i\in I}\phi_i ~|~ \neg\phi ~|~ \eps\diam{\alpha}\hat\phi ~|~ \hat\phi ~(\hat\phi\in\IO{d}) \\
\\
&\IO{w}       && \phi ::= \mbox{$\bigwedge$}_{i\in I}\phi_i ~|~ \neg\phi ~|~ \eps\phi ~|~ \eps\diam{a}\eps\phi \\
&\IO{rw}      && \phi ::= \mbox{$\bigwedge$}_{i\in I}\phi_i ~|~ \neg\phi ~|~ \eps\diam{\alpha}\eps\hat\phi ~|~ \hat\phi ~(\hat\phi\in\IO{w}).
\end{align*}
In these definitions, $a$ ranges over $A$ and $\alpha$ over $A\cup\{\tau\}$. The classes $\IO{e}^\equiv$ and $\IO{re}^\equiv$ denote the closures of $\IO{e}$, respectively $\IO{re}$, under $\equiv$.
}
\end{definition}
\noindent
The last clause in the definition of $\IO{re}$ guarantees that $\IO{e}
\subset \IO{re}$, which will be needed in the proof of
Prop.~\ref{prop:rooted_delay_preservation}. If this clause were omitted,
it would still follow that $\IO{e}^\equiv\subset\IO{re}^\equiv$,
using structural induction together with $\eps\phi\equiv\phi\vee\eps\diam{\tau}\phi$ (for $e=d$)
or $\eps\phi\equiv\phi\vee\eps\diam{\tau}\eps\phi$ (for $e=w$).

For $L\subseteq\mathbb{O}$, we write $p\sim_L q$ if $p$ and $q$ satisfy the same formulas in $L$. Note that, trivially, $p\sim_{\IO{e}}q \Leftrightarrow p\sim_{\IO{e}^\equiv}q$ and $p\sim_{\IO{re}}q \Leftrightarrow p\sim_{\IO{re}^\equiv}q$.

\begin{theorem}\label{thm:characterisation}
$p\bis{e}q\Leftrightarrow p\sim_{\IO{e}}q$ and $p\bis{re}q\Leftrightarrow p\sim_{\IO{re}}q$, for all $p,q\in\mathbb{P}$, where $e\in\{d,w\}$.
\end{theorem}
The first statement is a restatement of standard results from \cite{HM85,Mil81,Mil89,vGl93}; the
second statement is an easy corollary, obtained in the same vein.
Moreover, a proof for the case $e=w$ is presented in Appendix~\ref{app:modal}.
The proof for the case $e=d$ is similar.

\subsection{Structural operational semantics}\label{sec:sos}

A \emph{signature} is a set $\Sigma$ of function symbols $f$ with arity $\ar(f)$.
Let $V$ be an infinite set of variables, with typical elements
$x,y,z$; we always take $|\Sigma|, |A| \leq |V|$.
A syntactic object is {\em closed} if it does not contain any variables.
The set $\mathbb{T}(\Sigma)$ of terms over $\Sigma$ and $V$ is defined as
usual; $t,u,v,w$ denote terms and $\var(t)$ is the set of variables
that occur in term $t$.
A term is {\em univariate} if it is without multiple occurrences of the same variable.
 A substitution $\sigma$ is a partial function
from $V$ to $\mathbb{T}(\Sigma)$. A closed substitution is a total
function from $V$ to closed terms. The domain of substitutions is
extended to $\mathbb{T}(\Sigma)$ as usual.

Structural operational semantics \cite{Plo04} provides process algebras and specification languages with an interpretation. It generates an LTS, in which processes are the closed terms over a (single-sorted, first-order) signature, and transitions between processes may be supplied with labels. The transitions between processes are obtained from a transition system specification, which consists of a set of proof rules called transition rules.

\begin{definition}\label{def:TSS}\rm
A (\emph{positive} or \emph{negative}) \emph{literal} is an expression
\plat{$t \trans\alpha u$} or \plat{$t\ntrans\alpha$}. A
\emph{(transition) rule} is of the form $\frac{H}{\lambda}$ with $H$ a
set of literals called the \emph{premises}, and $\lambda$ a literal
called the \emph{conclusion}; the term at the left-hand
side of $\lambda$ is called the \emph{source} of
the rule.
Given a transition rule $\frac{H}{\lambda}$, write $H^+$ for the set of
positive premises in $H$, and $H^{s-}$ for the set of \emph{stable}
negative premises in $H$: those premises \plat{$t \ntrans\alpha$} for which
also the premise $t \ntrans\tau$ is in $H$.
With ${\it rhs}(H)$ we denote the set of
right-hand sides of the premises in $H^{+}$.
A rule $\frac{\emptyset}{\lambda}$ is also written $\lambda$.
A rule is {\em standard} if it has a positive conclusion,
and {\em positive} if moreover it has only positive premises.
A {\em transition system specification (TSS)}, written $(\Sigma,R)$,
consists of a signature $\Sigma$ and a collection $R$ of transition rules over $\Sigma$.
A TSS is {\em standard} or {\em positive} if all its rules are.
\end{definition}

\noindent
The following definition tells when a literal is provable from a TSS. It generalises the standard definition (see e.g.\ \cite{GV92}) by allowing the derivation of transition rules. The derivation of a literal $\lambda$ corresponds to the derivation of the transition rule $\frac{H}{\lambda}$ with $H=\emptyset$. The case $H \neq \emptyset$ corresponds to the derivation of $\lambda$ under the assumptions $H$.

\begin{definition}\label{def:proof}
{\rm
Let $P=(\Sigma,R)$ be a TSS. An {\em irredundant proof} from $P$ of a transition rule $\frac{H}{\lambda}$ is a well-founded tree with the nodes labelled by literals and some of the leaves marked ``hypothesis'', such that the root has label $\lambda$, $H$ is the set of labels of the hypotheses, and if $\mu$ is the label of a node that is not a hypothesis and $K$ is the set of labels of the children of this node then $\frac{K}{\mu}$ is a substitution instance of a transition rule in $R$.
}
\end{definition}
\noindent
The proof of $\frac{H}{\lambda}$ is called irredundant \cite{BFvG04} because $H$ must equal (instead of include)
the set of labels of the hypotheses. Irredundancy will be crucial for the preservation under provability of our congruence formats;
see Sect.~\ref{sec:preservation}. Namely, in a `redundant' proof one can freely add premises to the derived rule,
so also a premise that violates a syntactic restriction of the congruence format under consideration.

A TSS is meant to specify an LTS in which the transitions are closed positive literals. A standard TSS with only positive premises
specifies an LTS in a straightforward way, but it is not so easy to associate an LTS to a TSS with negative premises.
From \cite{vGl04} we adopt the notion of a well-supported proof of a closed literal.
Literals \plat{$t \trans{\alpha} u$} and \plat{$t\ntrans{\alpha}$} are said to \emph{deny} each other.

\begin{definition}\label{def:wsp}
{\rm
Let $P=(\Sigma,R)$ be standard TSS. A \emph{well-supported proof} from $P$ of a closed literal ${\lambda}$ is a well-founded tree with the nodes labelled by closed literals, such that the root is labelled by $\lambda$, and if $\mu$ is the label of a node and $K$ is the set of labels of the children of this node, then: %\vspace{-3pt}
\begin{enumerate}
\item either $\mu$ is positive and $\frac{K}{\mu}$ is a closed substitution instance of a transition rule in $R$; %\vspace{-1pt}
\item or $\mu$ is negative and for each set $N$ of closed negative literals with $\frac{N}{\nu}$ irredundantly provable
from $P$ and $\nu$ a closed positive literal that denies $\mu$, a literal in $K$ denies one in $N$. %\vspace{-3pt}
\end{enumerate}
$P\vdash_{\it ws}\lambda$ denotes that a well-supported proof from $P$ of $\lambda$ exists.
A standard TSS $P$ is \emph{complete} if for each $p$ and $\alpha$, either \plat{$P\vdash_{\it ws}p\ntrans\alpha$} or
\plat{$P\vdash_{\it ws}p\trans\alpha p'$} for some $p'$.
}
\end{definition}
\noindent
In \cite{vGl04} it was shown that $\vdash_{\it ws}$ is consistent, in the sense that no standard TSS
admits well-supported proofs of two literals that deny each other.
A complete TSS specifies an LTS, consisting of the {\it ws}-provable closed positive literals.

\subsection{Syntactic restrictions on transition rules}\label{sec:ntytt}

In this section we present terminology for syntactic restrictions on rules,
originating from \cite{BFvG04,Gro93,GV92}, where congruence formats are presented for
a range of concrete semantics (which do not take into account the internal action $\tau$).

\begin{definition}\label{def:ntytt}
{\rm
An \emph{ntytt rule} is a transition rule in which the right-hand sides of positive premises are variables that
are all distinct, and that do not occur in the source. An ntytt rule is an \emph{ntyxt rule} if its
source is a variable, an \emph{ntyft rule} if its source contains exactly one function symbol and no
multiple occurrences of variables, an \emph{nxytt rule} if the left-hand sides of its premises
are variables, and an \emph{xyntt rule} if the left-hand sides of its positive premises
are variables. An \emph{xynft} rule is both ntyft and xyntt.
}
\end{definition} 

\begin{definition}\label{def:decent}
{\rm
A variable in a transition rule is \emph{free} if it occurs neither in the source nor in right-hand sides of premises. A transition rule has \emph{lookahead} if some variable occurs in the right-hand side of a premise and in the left-hand side of a premise. A transition rule is \emph{decent} if it has no lookahead and does not contain free variables.
}
\end{definition}

\noindent
Each combination of syntactic restrictions on transition rules
induces a corresponding syntactic format for TSSs of the same name.
For instance, a TSS is in decent ntyft format if it contains
decent ntyft rules only.

The following lemma, on the preservation of decency under irredundant provability, was proved in \cite{BFvG04}.

\begin{lemma}
\label{lem:preservation_decency}
Let $P$ be a TSS in decent ntytt format. Then any ntytt rule irredundantly provable from $P$ is decent.
\end{lemma}

\noindent
We define two more syntactic formats for TSSs. The ntyft/ntyxt and ready simulation formats \cite{Gro93,BFvG04}
were originally introduced to guarantee congruence for bisimilarity and ready simulation.

\begin{definition}\label{def:ready_sim}
{\rm
A TSS is in \emph{ntyft/ntyxt format} if it consists of ntyft and ntyxt rules,
and in \emph{ready simulation format} if moreover its transition rules have no lookahead.
}
\end{definition}

\subsection{Patience rules}\label{sec:predicates}

We introduce some terminology for predicates on arguments of function symbols from \cite{Blo95,BFvG04}.

\begin{definition}\label{def:liquid/frozen}
{\rm
Let $\Gamma$ be a unary predicate on $\{(f,i)\mid 1 \leq i \leq ar(f),~f \in \Sigma \}$. If $\Gamma(f,i)$, then argument $i$ of $f$ is {\em $\Gamma$-liquid}; otherwise it is {\em $\Gamma$-frozen}. An occurrence of $x$ in $t$ is {\em $\Gamma$-liquid} if either $t=x$, or $t=f(t_1,\ldots,t_{\ar(f)})$ and the occurrence is $\Gamma$-liquid in $t_i$ for a liquid argument $i$ of $f$; otherwise the occurrence is {\em $\Gamma$-frozen}.
}
\end{definition}

\noindent
Note that an occurrence of a variable $x$ in a term $t \in \mathbb{T}(\Sigma)$ is $\Gamma$-frozen if and only if $t$ contains a subterm $f(t_1,\ldots,t_{\ar(f)})$ such that the occurrence of $x$ is in $t_i$ for a $\Gamma$-frozen argument $i$ of $f$.

In Sect.~\ref{sec:decomposition} we will present a method for
decomposing modal formulas that gives a special treatment to arguments
of function symbols that are deemed \emph{patient}; we will use a
predicate $\Gamma$ to mark the arguments that get this special treatment.

\begin{definition}\label{def:patience_rule}
{\rm \cite{Blo95,Fok00}
A standard ntyft rule is a {\em patience rule} for argument $i$ of $f$
if it is of the form %\vspace{-1em}
\[
\frac{x_i\trans{\tau}y}{f(x_1,\ldots,x_{\ar(f)})\trans{\tau}f(x_1,\ldots,x_{i-1},y,x_{i+1},\ldots,x_{\ar(f)})}
\]
Given a predicate $\Gamma$, the rule above is called a {\em $\Gamma$-patience rule},
if $\Gamma(f,i)$.
A TSS is \emph{$\Gamma$-patient} if it contains all $\Gamma$-patience rules.
A standard ntytt rule is {\em $\Gamma$-patient} if it is
irredundantly provable from the $\Gamma$-patience rules; else it is called {\em $\Gamma$-impatient}.
}
\end{definition}

\noindent
A patience rule for an argument $i$ of a function symbol $f$ expresses that terms
$f(p_1,\ldots,p_n)$ can mimic the $\tau$-transitions of argument $p_i$.
Typically, in process algebra, there are patience rules for both arguments of the merge operator
and for the first argument of sequential composition, as they can contain running processes,
but not for the arguments of alternative composition or for the second argument of sequential composition.

\subsection{Ruloids}\label{sec:ruloids}

To decompose modal formulas, we use a result from \cite{BFvG04},
where for any standard TSS $P$ in ready simulation format
a collection of decent nxytt rules, called {\it $P$-ruloids}, is constructed.
We explain this construction at a rather superficial level;
the precise transformation can be found in \cite{BFvG04}.

First $P$ is converted to a standard TSS $P^\dagger$ in decent ntyft format.
In this conversion from \cite{GV92}, free variables in a rule are
replaced by all closed terms (generating a different rule for each substitution),
and if the source is of the form $x$, then
this variable is replaced by a term $f(x_1,\ldots,x_{\ar(f)})$
for each function symbol $f$ in the signature of $P$,
where the variables $x_1,\ldots,x_{\ar(f)}$ are fresh.

Next, using a construction from \cite{FvG96}, left-hand sides of
positive premises are reduced to variables. Roughly
the idea is, given a premise \plat{$f(t_1,\ldots,t_n)\trans{\alpha}y$}
in a rule $r$, and another rule
$\frac{H}{f(x_1,\ldots,x_n)\trans{\alpha}t}$, to transform $r$ by
replacing the aforementioned premise by $H$, $y$ by $t$, and the $x_i$
by the $t_i$; this is repeated (transfinitely) until all positive
premises with a non-variable term as left-hand side have disappeared.
This yields an intermediate standard TSS $P^\ddagger$ in xynft format, of which all the
rules are irredundantly provable from $P$. In fact, the rules of $P^\ddagger$ are exactly the
xynft rules irredundantly provable from $P^\dagger$.
The motivation for this transformation step is that for TSSs in xynft format the semantic phrase
``for each set $N$ of closed negative literals with $\frac{N}{\nu}$ irredundantly provable
from $P$'' in the second clause of Def.~\ref{def:wsp} of a well-supported proof
can be replaced by a syntactic phrase: ``for each closed substitution instance
$\frac{N}{\nu}$ of a rule in $R$''.

In the final transformation step, non-standard rules with a negative
conclusion \plat{$t\ntrans\alpha$} are introduced.  The motivation is
that instead of the notion of well-founded provability of
Def.~\ref{def:wsp}, we want a more constructive notion like
Def.~\ref{def:proof}, by making it possible that a negative premise is
matched with a negative conclusion. A non-standard rule
$\frac{H}{f(x_1,\ldots,x_n)\ntrans{\alpha}}$ is obtained by picking
one premise from each xynft rule in $P^\ddagger$ with a conclusion of the form
\plat{$f(x_1,\ldots,x_n)\trans{\alpha}t$}, and including the denial of
each of the selected premises as a premise in $H$.

The resulting TSS, which is in decent ntyft format, is denoted by $P^+$.
The above construction implies that if $P$ is $\Gamma$-patient, then so is $P^+$.
In \cite{BFvG04} it was established, for all closed literals $\mu$, that $P\vdash_{\it ws}\mu$
if and only if $\mu$ is irredundantly provable from $P^+$. By definition,
the \emph{$P$-ruloids} are the (decent) nxytt rules irredundantly provable from $P^+$.

The following correspondence result from
\cite{BFvG04} between a TSS and its ruloids is crucial for the soundness of the
decomposition method presented in Sect.~\ref{sec:decomposition}. It says that there is a well-supported proof
from $P$ of a transition \plat{$\rho(t) \trans{a}q$}, with $\rho$ a closed substitution,
if and only if there is a derivation of this transition that uses at the root a $P$-ruloid with source $t$.

\begin{proposition}\label{prop:ruloid}
Let $P=(\Sigma,R)$ be a standard TSS in ready simulation format,
$t\in\mathbb{T}(\Sigma)$ and $\rho: V \rightarrow \mathbb{T}(\Sigma)$ a closed substitution.
Then $P\vdash_{\it ws}\rho(t)\trans\alpha q$ if and only if there
are a $P$-ruloid $\frac{H}{t\trans\alpha u}$ and a {\valuation}
$\rho'$ such that $P\vdash_{\it ws}\rho'(\mu)$ for all $\mu\in H$,
$\rho'(t)=\rho(t)$ and $\rho'(u)=q$.
\end{proposition}

\subsection{Linear proofs}\label{sec:linearity}

\begin{definition}\rm
An irredundant proof of a transition rule $\frac{H}{\lambda}$ is called \emph{linear}
if no two hypotheses in the proof tree of Def~\ref{def:proof} are labelled with the same positive premise.
\end{definition}

\begin{example}
From the TSS with rules
\[\frac{x\trans a y}{f(x)\trans b x} \qquad \frac{x\trans a y}{f(x)\trans c x} \qquad \frac{x\trans b y \quad x\trans c z}{g(x)\trans d x}\]
the ntytt rule $\frac{x\trans a y}{g(f(x))\trans d x}$ is irredundantly provable, but not with a linear proof.
However, the ntytt rule $\frac{x\trans a y \quad x\trans a z}{g(f(x))\trans d x}$ has a linear proof.
\end{example}
Clearly, each ntytt rule provable from a TSS is a substitution instance of an ntytt rule that has a
linear proof.
In Def.~\ref{def:wsp} it does not make any difference whether in clause 2 we quantify over rules
$\frac{N}{\nu}$ that are provable (as in \cite{vGl04,BFvG04,FvGdW03}), irredundantly provable (as in
\cite{FvGdW12}), or linearly provable.

\begin{lemma}\label{linear proof composition}
If a rule \plat{$\frac{H}{\lambda}$} is linearly provable from a TSS $P$, then so
is \nietplat{$\frac{\sigma(H)}{\sigma(\lambda)}$} for any substitution $\sigma$.

If a rule \nietplat{$\frac{ H}{\lambda}$} as well as rules $\frac{ H_\mu}{\mu}$ for each
literal $\mu\mathbin\in  H$ are linearly provable, with the $H_\mu$ pairwise disjoint,
then so is the rule \nietplat{$\frac{\bigcup_{\mu\in  H} H_\mu}{\lambda}$}.
\end{lemma}
\begin{proof}
Directly from the definition, and by composition of linear proofs.\qed
\end{proof}

In Sect.~\ref{sec:ruloids} a non-standard TSS $P^+$ is constructed out of a given TSS $P$ in ready
simulation format, via the intermediate stages $P^\dagger$ and $P^\ddagger$.
Here $P^\ddagger$ consists of all xynft rules irredundantly provable from $P^\dagger$.
With $\hat P^\ddagger$ we denote the TSS consisting of all xynft rules linearly provable from $P^\dagger$.
The TSS $P^+$ is obtained by augmenting $P^\ddagger$ with non-standard rules;
with $\hat P^+$ we denote the corresponding augmentation of $\hat P^\ddagger$.
For the construction of the non-standard rules in the augmentation, it makes no
difference whether we start from $P^\ddagger$ or $\hat P^\ddagger$, since the difference disappears
when abstracting from the right-hand sides of positive literals.
Thus, $\hat P^+ \subseteq P^+$ and each rule in $P^+$ is a substitution instance of a rule in $\hat P^+$.
Hence, an ntytt rule is irredundantly provably from $\hat P^+$ iff it is irredundantly provable from $P^+$.

\begin{definition}\rm
A $P$-ruloid is called \emph{linear} if it has a linear proof from $\hat P^+$.
\end{definition}
Clearly, each $P$-ruloid is a substitution instance of a linear $P$-ruloid.
Consequently, Prop.~\ref{prop:ruloid} still holds if we only consider linear ruloids.

An essential part of our forthcoming congruence formats (cf.~Def.~\ref{def:rooted_delay_bisimulation_format})
is a semantic requirement---``delay resistance'', Def.~\ref{def:delay-resistant}---on linear $P$-ruloids.
To make the formats more easily applicable, we will show (in Thm.~\ref{thm:manifest}) that delay resistance is
implied by a requirement on the rules of $P$. That result would fail if
delay resistance were required for all $P$-ruloids. In the appendix it is indicated where linearity is used;
see the remark after the proof of Lem.~\ref{lem:preservation positive} as well as Ex.~\ref{ex:linearity}.

\subsection{Decomposition of modal formulas}\label{sec:decomposition}

In \cite{FvGdW12} it was shown how one can decompose formulas from $\mathbb{O}$.
To each term $t$ and formula $\phi\in\mathbb{O}$ a set $t^{-1}(\phi)$ of decomposition mappings
$\psi:V\rightarrow\mathbb{O}$ is assigned. Each of these mappings $\psi\in t^{-1}(\phi)$ guarantees that
for any {\valuation} $\rho$,
$\rho(t)\models\phi$ if $\rho(x)\models\psi(x)$ for all $x\in\var(t)$. Vice versa, whenever
$\rho(t)\models\phi$, there is a decomposition mapping $\psi \in t^{-1}(\phi)$ with
$\rho(x)\models\psi(x)$ for all $x \in \var(t)$. This is formalised in Thm.~\ref{thm:decomposition}.

\begin{definition}\label{def:decomposition}
{\rm \cite{FvGdW12}
Let $P=(\Sigma,R)$ be a $\Gamma$-patient standard TSS in ready
simulation format. We define
$\cdot^{-1}:\mathbb{T}(\Sigma)\times\mathbb{O}\rightarrow\mbox{\fsc
  P}(V\rightarrow\mathbb{O})$ as the function that for each
$t\in\mathbb{T}(\Sigma)$ and $\varphi\in\mathbb{O}$ returns the
set $t^{-1}(\varphi)\in\mbox{\fsc P}(V\mathbin{\rightarrow}\mathbb{O})$
of decomposition mappings $\psi:V\mathbin{\rightarrow}\mathbb{O}$
generated by following five conditions.
In the remainder of this definition, $t$ denotes a univariate term, i.e.\ without multiple occurrences of the same variable.

\begin{enumerate}
\item\label{dec1}
$\psi\in t^{-1}(\bigwedge_{i\in I}\phi_i)$ iff there are $\psi_i\in t^{-1}(\phi_i)$ for each $i\in I$ such that
\[
\psi(x)= \displaystyle{\bigwedge_{i\in I}\psi_i(x)}\qquad\text{ for all }x\in V
\]

\item\label{dec2}
$\psi\in t^{-1}(\neg\phi)$ iff there is a function $h:t^{-1}(\phi)\rightarrow\var(t)$ such that
\[
\psi(x)=\left\{
  \begin{array}{ll}
  \displaystyle{\bigwedge_{\chi\in h^{-1}(x)}\!\!\!\neg\chi(x)} & \mbox{if $x\in\var(t)$}\vspace{2mm}\\
  \top & \mbox{if $x\notin\var(t)$}
  \end{array}
\right.
\]

\item\label{dec3}
$\psi\in t^{-1}(\diam\alpha\phi)$ iff there is a $P$-ruloid $\frac{H}{t\trans{\alpha}u}$
and a $\chi\in u^{-1}(\phi)$ such that
\[
\psi(x)=\left\{
  \begin{array}{ll}
  \displaystyle{\chi(x)\ \land\ \bigwedge_{x\trans{\beta}y\in H}\!\!\!\!\!\diam{\beta}\chi(y)
                \ \land\  \bigwedge_{x\ntrans{\gamma}\in H}\!\!\!\!\!\neg\diam{\gamma}\top\;} & \mbox{if $x\in\var(t)$}\\
  \top & \mbox{if $x\notin\var(t)$}
  \end{array}
\right.
\]

\item\label{dec4}
$\psi\in t^{-1}(\eps\phi)$ iff one of the following holds:    \vspace{2mm}

\begin{enumerate}
\item  
\label{4a}
either there is a $\chi\in t^{-1}(\phi)$ such that
\[
\psi(x)=\left\{
\begin{array}{ll}
\eps\chi(x) & \textrm{if $x$ occurs $\Gamma$-liquid in $t$}\\
\chi(x) & \textrm{otherwise}\\
\end{array}
\right.
\]

\item
\label{4b}
or there is a $\Gamma$-impatient $P$-ruloid $\frac{H}{t\trans{\tau}u}$ and a
$\chi\in u^{-1}(\eps\phi)$ such that
\[
\psi(x)=\left\{
  \begin{array}{@{}ll@{}}
  \displaystyle{\eps\left(\chi(x)\ \land\ \!\!\!\!\!\!\bigwedge_{x\trans{\beta}y\in H}\!\!\!\!\!\diam{\beta}\chi(y)
                \ \land\  \!\!\!\!\!\!\bigwedge_{x\ntrans{\gamma}\in H}\!\!\!\!\!\neg\diam{\gamma}\top\right)} &
         \mbox{\begin{tabular}{@{}l@{}}if $x$ occurs\\ $\Gamma$-liquid in $t$\end{tabular}}\\
  \displaystyle{\chi(x) \land \!\!\!\!\!\!\bigwedge_{x\trans{\beta}y\in H}\!\!\!\!\!\diam{\beta}\chi(y)
                \ \land \!\!\!\!\!\bigwedge_{x\ntrans{\gamma}\in
  H}\!\!\!\!\!\neg\diam{\gamma}\top\;} &
         \mbox{\begin{tabular}{@{}l@{}}if $x$ occurs\\ $\Gamma$-frozen in $t$\end{tabular}}\\
  \top & \mbox{if $x\notin\var(t)$}
  \end{array}
\right.
\]
\end{enumerate}

\item\label{dec6}
$\psi\in\sigma(t)^{-1}(\phi)$ for a non-injective substitution $\sigma:\var(t)\rightarrow V$
iff there is a $\chi\in t^{-1}(\phi)$ such that
\[
\psi(x)=\bigwedge_{z\in\sigma^{-1}(x)}\chi(z)\qquad\text{ for all }x\in V.
\]
\end{enumerate}
}
\end{definition}

\begin{theorem}\label{thm:decomposition} {\rm \cite{FvGdW12}}
Let $P=(\Sigma,R)$ be a $\Gamma$-patient complete standard TSS in ready
simulation format. For any term $t\in\mathbb{T}(\Sigma)$, {\valuation}
$\rho$, and $\phi\in\mathbb{O}$:
\[
\rho(t)\models\phi\ \Leftrightarrow
\ \exists\psi\in t^{-1}(\phi)\ \forall x\in \var(t):\rho(x)\models\psi(x)\enspace .
\]
\end{theorem}

\section{Delay resistant TSSs}\label{sec:delay-resistance}

\newcommand{\ndr}{negative delay resistant}
\newcommand{\pdr}{positive delay resistant}
\newcommand{\dra}{delay resistant}
\newcommand{\dr}{delay resistant}

In the next section we will apply the decomposition method from
Def.~\ref{def:decomposition} and Thm.~\ref{thm:decomposition} to obtain
congruence formats for (rooted) delay and weak bisimilarity.
However, compared to branching and $\eta$-bisimilarity, which was the focus of \cite{FvGdW12},
Def.~\ref{def:decomposition} needs to be refined in the case $t^{-1}(\eps\phi)$.
This is because in the modal logics for delay and weak bisimilarity, occurrences of subformulas
$\diam{\beta}\phi'$ are always preceded by $\eps$, while in Def.~\ref{def:decomposition}
this is not always the case. The refinement of Def.~\ref{def:decomposition}, which is presented in
Def.~\ref{def:decomposition-dr}, is only valid for so-called {\dra} TSSs.

Def.~\ref{def:delay-resistant} of delay resistance is inspired by a requirement in the
RDB and RWB cool formats, see \cite[Def.~15(3)]{vGl11}.
It is crafted in such a way that Prop.~\ref{prop:delay-resistance} holds: %\vspace{1pt}
if for a premise \nietplat{$x \trans\beta y$} in a ruloid \nietplat{$r=\frac{H}{t\trans{\alpha}u}$} %\vspace{1pt}
the execution of $\beta$ is delayed by a $\tau$-step, i.e.\  for some substitution $\sigma$ we
merely have \nietplat{$\sigma(x) \trans\tau \trans\beta \sigma(y)$},
we want that the conclusion $\sigma(t) \trans{\alpha} \sigma(u)$ of the $\sigma$-instance of
$r$ is unaffected, or merely delayed by a $\tau$-step as well.

\begin{example}
Consider the rooted delay bisimilar processes $p_0$ and $p_1$ from the first LTS in Ex.~\ref{ex:weak bisimulations}.
The ruloid $r$ may apply when substituting $p_0$ for $x$ and $b$ for $\beta$, given that \plat{$p_0 \trans b 0$}.
If instead of $p_0$ we substitute $p_1$ for $x$, to safeguard the congruence property it is necessary
that the (possibly delayed) conclusion of $r$ can still be derived, even though we only have
\plat{$p_1 \trans \tau q \trans b 0$}.
\end{example}

In Def.~\ref{def:positive} we allow two possible implementations of this idea.
Each premise $x {\trans\beta} y$ of $r$ is either \emph{delayable} (Def.~\ref{def:delayable}),
in which case $\sigma(x) \trans\tau \trans\beta \sigma(y)$ induces $\sigma(t) \trans\tau\trans{\alpha} \sigma(u)$
by two ruloids that can be used in place of $r$; or
\emph{$\tau$-pollable}, meaning that $r$ remains valid if this premise
is replaced by \plat{$x \trans\tau z$} for a fresh variable $z$, so
that the premise \plat{$\sigma(x) \trans\tau \sigma(z)$}
takes over the role of $\sigma(x) \trans\beta \sigma(y)$.

Def.~\ref{def:positive} and Prop.~\ref{prop:delay-resistance} allow only finitely many delayable positive premises, but infinitely
many $\tau$-pollable ones.

\begin{definition}\rm\label{def:delayable}
A premise $w \trans\beta y$ of an ntytt rule $r=\frac{H}{t\trans\alpha u}$ is
\emph{delayable} in a TSS $P$ if
  there are ntytt rules $\frac{ H_1}{t \trans{\tau} v}$ and $\frac{ H_2}{v \trans{\alpha} u}$, linearly
  provable from $P$, with {$H_1\subseteq (H\setminus\{w \trans{\beta} y\})\cup\{w \trans{\tau} z\}$}
  and {$H_2\subseteq\linebreak (H\setminus\{w \trans{\beta} y\})\cup\{z \trans{\beta} y\}$}
  for some term $v$ and fresh variable $z$.
\end{definition}

\noindent
Suppose that $\frac{H}{t\trans\alpha u}$ in Def.~\ref{def:delayable} is a ruloid, so that $w$ is a variable $x$.
The intuition behind this definition is that the argument $x$ of $t$
may not be able to perform a $\beta$-transition to a term $y$ immediately, but only after
a $\tau$-transition to $z$. The ruloid $\frac{H_1}{t \trans{\tau} v}$ then
allows $t$ to postpone its $\alpha$-transition to $u$, by first performing a $\tau$-transition to $v$.
The ruloid $\frac{H_2}{v \trans{\alpha} u}$ guarantees that the postponed $\beta$-transition
from $z$ to $y$ still gives rise to an $\alpha$-transition from $v$ to $u$.

Linearity is needed to make sure that in the construction of ruloids,
distinct delayable positive premises are never collapsed to a single non-delayable premise.

Def.~\ref{def:positive} requires that each positive premise is either delayable (i.e., in the finite set $H^d$),
or $\tau$-pollable, or redundant. Recall from Def.~\ref{def:TSS} the notation $H^+$ for the positive premises in $H$.

\begin{definition}\rm\label{def:positive}
An ntytt rule \nietplat{$\frac{ H}{t\trans{\alpha}u}$}
is \emph{\pdr} w.r.t.\ a TSS $P$ if there exists a finite set $H^d \subseteq H^+$ of delayable positive premises such that
  for each set $M \subseteq H^+\setminus H^d$ there is a rule
  \nietplat{$r_M=\frac{H_M}{t\trans{\alpha}u}$}, linearly provable from $P$, where
  $H_M \subseteq (H{\setminus}M) \cup M_\tau$ with
  $M_\tau\mathbin=\{w \trans{\tau} z_y \mid (w \trans{\beta} y) \mathbin\in M,~ z_y~\mbox{fresh}\}$.
\end{definition}
The intuition behind Def.~\ref{def:negative} is closely related to Def.~\ref{def:delayable}.
If a ruloid $\frac{H}{t\trans\alpha u}$ has a premise $x\ntrans\gamma$, we want it not to apply
even if $\sigma(x)$ merely has a delayed $\gamma$-transition $\sigma(x) \trans\tau \trans\gamma$.
We therefore require that for each premise \nietplat{$x\ntrans\gamma$} there must also be a premise $x\ntrans\tau$.
However, we make an exception for redundant premises \nietplat{$x\ntrans\gamma$}.
Recall from Def.~\ref{def:TSS} the notation $H^{s-}$ for the stable negative premises in $H$.

\begin{definition}\rm\label{def:negative}
A rule \nietplat{$\frac{H}{t\trans{\alpha}u}$} %\vspace{1pt}
is \emph{\ndr} w.r.t.\ a TSS $P$ if there is a rule \nietplat{$\frac{H'}{t\trans{\alpha}u}$}, linearly provable from $P$,
with $H' \subseteq H^+ \cup H^{s-}$.
\end{definition}

\begin{definition}\rm\label{def:delay-resistant}
An ntytt rule \nietplat{$\frac{ H}{t\trans{\alpha}u}$} %\vspace{1pt}
is \emph{delay resistant} w.r.t.\ a TSS $P$ if it is {\pdr} as well as \ndr.
A standard TSS $P$ in ready simulation format is \emph{delay resistant} if all its linear ruloids
with a positive conclusion are delay resistant w.r.t.\ $\hat P^+$.
\end{definition}

\noindent
The following proposition is key to the notion of delay resistance. It will allow us to adapt the definition
of modal decomposition for {\dra} TSSs, so that it becomes applicable for generating congruence formats
for weak semantics, like delay and weak bisimilarity, with a modal characterisation in which a modality
$\diam{\beta}\phi$ is always preceded by $\eps$.

\begin{proposition}\label{prop:delay-resistance}
Let $P$ be a {\dra} standard TSS in ready simulation format.
Let \nietplat{$\frac{H}{t\trans{\alpha}u}$} be a $P$-ruloid and $\rho$ a {\valuation} such that
\nietplat{$P \vdash_{\it ws} \rho(x) \epsarrow\trans\beta \rho(y)$} for each premise \nietplat{$x\trans\beta y$} in $H^+$ and
$P \vdash_{\it ws} \rho(x)\!\ntrans\gamma$ for each premise \nietplat{$x\!\ntrans\gamma$} in $H^{s-}$.
Then \nietplat{$P \vdash_{\it ws} \rho(t) \epsarrow\trans\alpha \rho(u)$}.
\end{proposition}

\begin{proof}
We prove the lemma for linear $P$-ruloids \nietplat{$\frac{H}{t\trans{\alpha}u}$}; %\vspace{1pt}
as each ruloid is a substitution instance of a linear ruloid, the result for general ruloids then follows.

We apply induction on the sum, over all premises \nietplat{$x\trans\beta y$} in
$H^{+}$, of the (smallest possible) number of $\tau$-transitions in the sequence 
\nietplat{$\rho(x) \epsarrow$}. We note that the cases where this sum is infinite are in the proof
immediately reduced to cases where this sum is finite.

\textit{Induction base:} If the sum is zero, $P \vdash_{\it ws} \rho(x) \mathbin{\trans{\beta}} \rho(y)$
for each $x\mathbin{\trans\beta} y$ in $H^{+}$.\linebreak[2] By Def.~\ref{def:negative} there exists
a $P$-ruloid $\frac{H'}{t\trans{\alpha}u}$ with $H' \subseteq H^+ \cup H^{s-}$. By assumption,
$P \vdash_{\it ws} \rho(x) \!\ntrans{\gamma}$ for each $\mathord{x\!\ntrans{\gamma}}$ in $H'$.
So Prop.~\ref{prop:ruloid} yields $P \vdash_{\it ws} \rho(t) \trans{\alpha} \rho(u)$. %\vspace{3pt}

\textit{Induction step:\/} Suppose the sum is positive.
Take a finite set $H^d\subseteq H^+$ of delayable positive premises with the property ensured by Def.~\ref{def:positive}.

First we deal with the case that $P \not\vdash_{\it ws} \rho(x) \mathbin{\trans{\beta}} \rho(y)$
for some $x\trans\beta y$ in $H^{+}\setminus H^d$.
Let $M$ consist of those \nietplat{$x\trans\beta y$} in $H^+\setminus H^d$ for
which \nietplat{$P \vdash_{\it ws}\rho(x)\trans\tau q_y$} for some term $q_y$.
Let $\rho'$ be the closed substitution with $\rho'(z_y)=q_y$ for all right-hand sides $y$ of such premises
(where $z_y$ is the right-hand side of the corresponding premise in $M_\tau$),
and $\rho'$ coincides with $\rho$ on all other variables. Then $\rho'(t)=\rho(t)$ and $\rho'(u)=\rho(u)$.
Let \plat{$\frac{H_M}{t\trans{\alpha}u}$} be the linear $P$-ruloid that exists by Def.~\ref{def:positive},
with $H_M \subseteq (H{\setminus}M) \cup M_\tau$.
Then $P \vdash_{\it ws} \rho'(x)\!\ntrans\gamma$ for each $x\!\ntrans\gamma$ in $H^{s-}_M \subseteq H^{s-}$.
Moreover, we argue that $P \vdash_{\it ws} \rho'(x) \epsarrow\trans{\,\beta\!} \rho'(y)$ for each premise $x\trans\beta y$ in $H^+_M$.
For each $x\trans\beta y$ in $H^+\setminus M$ this is clear, because then $\rho'(x)=\rho(x)$ and $\rho'(y)=\rho(y)$. Furthermore,
the definition of $\rho'$ induces {$P \vdash_{\it ws} \rho'(x)=\rho(x) \trans\tau \rho'(y)$} for each \nietplat{$x\trans\tau y$} in $M_\tau$.
We apply induction with regard to the $P$-ruloid \plat{$\frac{H_M}{t\trans{\alpha}u}$} %\vspace{1pt}
and the closed substitution $\rho'$. Note that
$P \vdash_{\it ws} \rho'(x)=\rho(x) \trans\beta \rho(y)=\rho'(y)$ for each $x\trans\beta y$ in $H^+\setminus (H^d \cup M)$,
because the definition of $M$ induces $P \vdash_{\it ws} \rho(x)\!\ntrans\tau$.
Also the premises $x\trans\tau y$ in $M_\tau$ do not contribute
to the number of $\tau$-transitions $\rho'(x) \epsarrow$ on which we apply induction. Hence only premises $x\trans\beta y$ in $H^d$
contribute to this number. As $H^d$ is finite, this number is finite too.
Since $M$ is supposed to be non-empty, this number (for $\frac{H_M}{t\trans{\alpha}u}$ and $\rho'$)
is strictly smaller than for \plat{$\frac{H}{t\trans{\alpha}u}$} \vspace{1pt}
and $\rho$. An application of the induction hypothesis yields \plat{$P \vdash_{\it ws} \rho(t)=\rho'(t) \trans{\alpha} \rho'(u)=\rho(u)$}.

What remains is the case that $P \vdash_{\it ws} \rho(x) \mathbin{\trans{\beta}} \rho(y)$ %\vspace{-2pt}
for all \nietplat{$x\trans\beta y$} in $H^{+}\setminus H^d$. %\vspace{-2pt}
Then $P \not\vdash_{\it ws} \rho(x_0) \mathbin{\trans{\beta}} \rho(y_0)$ for some $x_0\trans\beta y_0$ in $H^d$.
Let \nietplat{$P \vdash_{\it ws}$} \nietplat{$\rho(x_0) \trans\tau p \epsarrow\trans{\beta_0} \rho(y_0)$},
where the closed term $p$ is chosen so that $\epsarrow$ is as short as possible.
Let $H_0=H{\setminus}\{x_0\trans{\beta_0} y_0\}$.
Since $x_0\trans{\beta_0} y_0$ is a delayable premise of \nietplat{$\frac{H}{t\trans{\alpha}u}$}, %\vspace{1pt}
by Def.~\ref{def:delayable} there are linear $P$-ruloids \nietplat{$\frac{H_1}{t \trans{\tau} v}$} and \nietplat{$\frac{H_2}{v \trans{\alpha} u}$}
%\vspace{1pt}
with $H_1\subseteq H_0\cup\{x_0 \trans{\tau} z_0\}$ and \nietplat{$H_2\subseteq H_0\cup\{z_0\trans{\beta_0}y_0\}$},
for some term $v$ and fresh variable $z_0$.
Let $\rho'(z_0)=p$ and $\rho'$ coincides with $\rho$ on all other variables. %\vspace{1pt}
Since $z_0$ does not occur in $H_0$, clearly \nietplat{$P \vdash_{\it ws} \rho'(x) \epsarrow\trans\beta \rho'(y)$}
for each \nietplat{$x\trans\beta y$} in $H_0^+$ and $P \vdash_{\it ws} \rho'(x)\!\ntrans\gamma$ for each
\nietplat{$x\!\ntrans\gamma$} in $H_0^{s-}$. Moreover, $P \vdash_{\it ws} \rho'(x_0) \trans\tau \rho'(z_0)$.
Compared with $\frac{H}{t\trans{\alpha}u}$ and $\rho$, in the case of $\frac{H_1}{t \trans{\tau} v}$ and $\rho'$
the number of $\tau$-transitions involved in the sequences $\epsarrow$ has decreased.
(As $H^d$ is finite, these numbers are finite too.)
So by induction, \nietplat{$P \vdash_{\it ws} \rho'(t) \epsarrow\trans\tau \rho'(v)$}. %\vspace{-3pt}
Furthermore, $P \vdash_{\it ws} \rho'(z_0)\epsarrow\trans{\beta_0}\rho'(y_0)$.
Again, compared with $\frac{H}{t\trans{\alpha}u}$ and $\rho$, in the case of $\frac{H_2}{v \trans{\alpha} u}$ and $\rho'$
the number of $\tau$-transitions involved in the sequences $\epsarrow$ has decreased. So by induction,
{$P \vdash_{\it ws} \rho'(v) \epsarrow\trans\alpha \rho'(u)$}.
Since $z_0\notin\var(t)\cup\var(u)$, we conclude that
\nietplat{$P \vdash_{\it ws} \rho(t){=}\rho'(t) \epsarrow \rho'(v) \epsarrow\trans\alpha \rho'(u){=}\rho(u)$}.%
\qed
\end{proof}

\noindent
As said before, for {\dra} TSSs case \ref{4b} of Def.~\ref{def:decomposition},
$t^{-1}(\eps\phi)$, needs to be adapted, to ensure that in the modal logics for delay and weak bisimilarity,
occurrences of subformulas $\diam{\beta}\phi''$ are always preceded by $\eps$. Moreover, case \ref{4a}
is provided with the restriction that $\phi$ is not of the form $\diam{\alpha}\phi'$.
Else decompositions of formulas $\eps\diam{\alpha}\phi'$ in $\IO{d}$ would
be defined in terms of formulas $\chi\in t^{-1}(\diam{\alpha}\phi')$, while $\diam{\alpha}\phi'$ is not in $\IO{d}$.
Instead, if $\phi$ is of the form $\diam{\alpha}\phi'$, cases \ref{dec3} and \ref{dec4}
of Def.~\ref{def:decomposition} are combined, as can be seen in case 4b(iii) below.

\begin{definition}\rm\label{def:decomposition-dr}
Let $P$ be a {\dra} $\Gamma$-patient standard TSS
in ready simulation format. We define 
$\cdot^{-1}_{\rm dr}:\mathbb{T}(\Sigma)\times\mathbb{O}\rightarrow\mbox{\fsc
  P}(V\rightarrow\mathbb{O})$ exactly as $\cdot^{-1}$ in
Def.~\ref{def:decomposition}, except for case \ref{dec4}:
$t^{-1}_{\rm dr}(\eps\phi)$ (with $t$ univariate).

\begin{enumerate}
\item[4.]
\label{dec-delay}
$\psi\in t^{-1}_{\rm dr}(\eps\phi)$ iff one of the following holds: \vspace{2mm}

\begin{enumerate}
\item\label{4a-delay}
either $\phi$ is not of the form $\diam{\alpha}\phi'$, and there is a $\chi\in t^{-1}_{\rm dr}(\phi)$ such that
\[
\psi(x)=\left\{
\begin{array}{ll}
\eps\chi(x) & \textrm{if $x$ occurs $\Gamma$-liquid in $t$}\\
\chi(x) & \textrm{otherwise}\\
\end{array}
\right.
\]

\item
or
\begin{enumerate}
\item[(i)]
$\phi$ is not of the form $\diam{\tau}\phi'$, and
there is a $\Gamma$-impatient $P$-ruloid $\frac{H}{t\trans{\tau}u}$ and a $\chi\in u^{-1}_{\rm dr}(\eps\phi)$,
\item[(ii)]
$\phi$ is of the form $\diam{\tau}\phi'$, and
there is a $\Gamma$-impatient $P$-ruloid $\frac{H}{t\trans{\tau}u}$ and a $\chi\in u^{-1}_{\rm dr}(\eps\phi')$,
\item[(iii)]
or $\phi$ is of the form $\diam{\alpha}\phi'$, and there is a $P$-ruloid $\frac{H}{t\trans{\alpha}u}$ and a
$\chi\in u^{-1}_{\rm dr}(\phi')$, such that\vspace{2mm}
\end{enumerate}
\[
\psi(x)=\left\{
  \begin{array}{ll}
  \displaystyle{\eps\left(\chi(x)\ \land\
          \!\!\!\!\!\bigwedge_{x\trans{\beta}y\in H^+}\!\!\!\!\!\eps\diam{\beta}\chi(y)
          \ \land\ \!\!\!\!\! \bigwedge_{x\ntrans{\gamma}\in
                H^{s-}}\!\!\!\!\!\!\neg\diam{\gamma}\top\right)} & \mbox{\begin{tabular}{l}if
                $x$ occurs\\ $\Gamma$-liquid in $t$\end{tabular}}\\
  \displaystyle{\chi(x)\ \land\ \!\!\!\!\!\bigwedge_{x\trans{\beta}y\in H^+}\!\!\!\!\!\eps\diam{\beta}\chi(y)
                \ \land\ \!\!\!\!\!\! \bigwedge_{x\ntrans{\gamma}\in H^{s-}}\!\!\!\!\!\neg\diam{\gamma}\top} &
          \mbox{\begin{tabular}{l}if $x$ occurs\\ $\Gamma$-frozen in $t$\end{tabular}}\\
  \top & \,\mbox{if $x\notin\var(t)$.}
  \end{array}
\right.
\]
\end{enumerate}
\end{enumerate}
\end{definition}

\noindent
The following three lemmas, from \cite{FvGdW12}, state basic properties of formulas $\psi(x)$.
We repeat them here to confirm that they also apply to Def.~\ref{def:decomposition-dr}.

\begin{lemma}
\label{lem:psi1-dr}
Let $\psi\in t^{-1}_{\rm dr}(\phi)$, for some term $t$ and formula $\phi$.
If $x\notin\var(t)$, then $\psi(x)\equiv\top$.
\end{lemma}

\begin{proof}
This can be derived in a straightforward fashion, by induction on the construction of $\psi$.
\qed
\end{proof}
The following lemma states that $\cdot^{-1}_{\rm dr}$ is invariant under
$\alpha$-conversion up to $\equiv$.

\begin{lemma}
\label{lem:alpha}
Let $\psi \in \sigma(t)^{-1}_{\rm dr}(\phi)$ for $\sigma:V \rightarrow V$ a
bijective renaming of variables. Then there is a $\psi' \in t^{-1}_{\rm dr}(\phi)$
satisfying $\psi'(x) \equiv \psi(\sigma(x))$ for all $x \in V$.
\end{lemma}

\begin{proof}
Again by induction on the construction of $\psi$.
\qed
\end{proof}

\begin{lemma}
\label{lem:psi2}
Let $\psi\in t^{-1}_{\rm dr}(\eps\phi)$ for some term $t$ and formula $\phi$.
If $x$ occurs only $\Gamma$-liquid in $t$, then $\psi(x)\equiv\eps\psi(x)$.
\end{lemma}

\begin{proof}
Let $x$ occur only $\Gamma$-liquid in $t$.
In case $t$ is univariate, it follows immediately from Def.~\ref{def:decomposition-dr} that
$\psi(x)\equiv\eps\phi'$ for some formula $\phi'$. So for general terms $t$,
$\psi(x)\equiv\bigwedge_{i\in I}\eps\phi_i$. This implies $\psi(x)\equiv\eps\psi(x)$.
\qed
\end{proof}

\noindent
The following theorem is the counterpart of Thm.~\ref{thm:decomposition} for {\dra} TSSs.

\begin{theorem}\label{thm:dr-decomposition}
Let $P=(\Sigma,R)$ be a {\dra} $\Gamma$-patient complete standard TSS in ready simulation format.
For any $t\mathbin\in\mathbb{T}(\Sigma)$, {\valuation} $\rho$, and $\phi\mathbin\in\mathbb{O}$: %\vspace{-1ex}
\[
\rho(t)\models\phi\ \Leftrightarrow
\ \exists\psi\in t^{-1}_{\rm dr}(\phi)\ \forall x\in \var(t):\rho(x)\models\psi(x)\enspace .
\]
\end{theorem}

\begin{proof}
By simultaneous induction on the structure of $\phi$---where a formula
$\eps\phi'$ counts as constructed before $\eps\diam\tau\phi'$---and the construction of $\psi$.
We only treat the case where $t$ is univariate; the case where $t$ is not univariate is identical
to the proof of Thm.~\ref{thm:decomposition} in \cite{FvGdW12}.
The proof is by a case distinction on the syntactic structure of $\phi$.
We only treat the case $\phi = \eps\phi'$ here, because all other cases
are identical to the proof of Thm.~\ref{thm:decomposition}.

($\Rightarrow$)
First we address the case that $\phi'$ is not of the form $\diam\tau\phi''$.
We prove by induction on $n$:
\begin{center}\begin{tabular}{l}
  if \plat{$P\vdash_{\it ws}p_i\trans{\tau}p_{i+1}$} for all
  $i\mathbin\in\{0,\ldots,n\mathord-1\}$, with $\rho(t)=p_0$ and
  $p_n\models\phi'$,\\ then there is a $\psi\in t^{-1}_{\rm dr}(\eps\phi')$ with
  $\rho(x)\models\psi(x)$ for all $x\in \var(t)$.
\end{tabular}\end{center}

\begin{description}
\item[$\boldsymbol{n=0}$] 
Since $\rho(t)=p_0\models\phi'$, by induction on formula size, there
is a $\chi\in t^{-1}_{\rm dr}(\phi')$ with $\rho(x)\models\chi(x)$ for
all $x\in \var(t)$. We distinguish two cases.

\item[{\sc Case 1:}] $\phi'$ is not of the form $\diam{\alpha}\phi''$.
Define $\psi\in t^{-1}_{\rm dr}(\eps\phi')$ as in
Def.~\ref{def:decomposition-dr}.4a, using $\chi$. Then clearly
$\rho(x)\models\psi(x)$ for all $x\in \var(t)$.

\item[{\sc Case 2:}] $\phi'$ is of the form $\diam{\alpha}\phi''$.
By Def.~\ref{def:decomposition}.\ref{dec3}
there is a $P$-ruloid $\frac{H}{t\trans{\alpha}u}$
and a $\xi\in u^{-1}_{\rm dr}(\phi'')$ such that
\[
\chi(x)=\left\{
  \begin{array}{ll}
  \displaystyle{\xi(x)\ \land\ \bigwedge_{x\trans{\beta}y\in H}\!\!\!\!\!\diam{\beta}\xi(y)
                \ \land\  \bigwedge_{x\ntrans{\gamma}\in H}\!\!\!\!\!\neg\diam{\gamma}\top\;} & \mbox{if $x\in\var(t)$}\\
  \top & \mbox{if $x\notin\var(t)$}
  \end{array}
\right.
\]
Such a formula remains valid for $\rho(x)$ if the conjuncts $\diam{\beta}\chi(y)$ are
weakened to $\eps\diam{\beta}\chi(y)$, if some conjuncts $\neg\diam{\gamma}\top$
are dropped (namely the ones for which \plat{$x\!\ntrans\gamma \;\notin H^{s-}$}),
and if the entire formula is prefixed by $\eps$.
Hence there is a $\psi\in t^{-1}_{\rm dr}(\eps\phi')$, defined
according to Def.~\ref{def:decomposition-dr}.4b(iii), such that
$\rho(x)\models\psi(x)$ for all $x\in \var(t)$.

\item[$\boldsymbol{n>0}$] 
Since \plat{$P\vdash_{\it ws}\rho(t)\trans{\tau}p_1$}, by Prop.\
\ref{prop:ruloid} there is a $P$-ruloid $\frac{H}{t\trans{\tau}u}$ and
a {\valuation} $\rho'$ with $P\vdash_{\it ws}\rho'(\mu)$ for all $\mu\in H$,
$\rho'(t)=\rho(t)$, i.e.\ $\rho'(x)=\rho(x)$ for all $x\mathbin\in \var(t)$,
and $\rho'(u)\mathbin=p_1$. Since $P \vdash_{\it ws}
\rho'(u)\mathbin=p_1\trans{\tau}\cdots\trans{\tau}p_n\models\phi'$,
by induction on $n$ there is a $\chi\in u^{-1}_{\rm dr}(\eps\phi')$ with
$\rho'(z)\models\chi(z)$ for each $z\in \var(u)$. Moreover, by
Lem.~\ref{lem:psi1-dr}, $\rho'(z)\models\chi(z)\equiv\top$ for each
$z\notin\var(u)$. We distinguish two cases.

\item[{\sc Case 1:}]
$\frac{H}{t\trans{\tau}u}$ is $\Gamma$-impatient.\vspace{1pt}
Define $\psi \mathbin\in t^{-1}_{\rm dr}(\eps\phi')$ as in
Def.~\ref{def:decomposition-dr}.4b(i), using $\frac{H}{t \trans{\tau} u}$ and $\chi$. 
Let $x\in\var(t)$.\vspace{1pt}
For each \plat{$x\trans{\beta}y\in H$}, \plat{$P \vdash_{\it ws} \rho'(x)
  \trans{\beta}\rho'(y)$} and $\rho'(y) \models \chi(y)$, so $\rho'(x)
\models\diam{\beta}\chi(y)$, and thus certainly
$\rho'(x)\models\eps\diam{\beta}\chi(y)$. Moreover, for each
\plat{$x\!\ntrans{\gamma}\;\in H^{s-}\subseteq H$}, \plat{$P \vdash_{\it ws} \rho'(x)
  \ntrans{\gamma}$}, so the consistency of $\vdash_{\it ws}$ yields
$P\not\vdash_{\it ws}\rho'(x)\trans{\gamma}q$ for all $q$, and thus
$\rho'(x) \models \neg\diam{\gamma}\top$. Hence $\rho(x) \mathbin= \rho'(x)
\models \psi(x)$. (In case the occurrence of $x$ in $t$ is
$\Gamma$-liquid, note that $p\models\xi$ implies $p\models\eps\xi$.)

\item[{\sc Case 2:}]
$\frac{H}{t\trans{\tau}u}$ is $\Gamma$-patient.
  Using that $t$ is univariate, $H$ must be of the form
{$\{x_0\mathbin{\trans{\tau}}y_0\}$}, with $u\mathbin=t[y_0/x_0]$, and the
unique occurrence of $y_0$ in $u$ being $\Gamma$-liquid.
Let $\sigma: V\mathbin\rightarrow V$ be the bijection that swaps $x_0$
  and $y_0$, so that $u=\sigma(t)$.  According to
  Lem.~\ref{lem:alpha}, there is a $\chi'\in t^{-1}_{\rm dr}(\eps\phi')$
  satisfying $\chi'(x) \equiv \chi(\sigma(x))$ for all $x \mathbin\in V$.
  For each $x\in \var(t)\backslash\{x_0\}$,
  $\rho(x)=\rho'(x)\models\chi(x)\equiv\chi'(x)$, so
  $\rho(x)\models\chi'(x)$.
  Furthermore, \plat{$P\vdash_{\it ws}\rho'(x_0)\trans{\tau}\rho'(y_0)\models\chi(y_0)$}.
  By Lem.~\ref{lem:psi2},
  $\chi(y_0)\mathbin\equiv\eps\chi(y_0)$. Hence
  $\rho(x_0)=\rho'(x_0)\models\chi(y_0)\equiv\chi'(x_0)$, so
  $\rho(x_0)\models\chi'(x_0)$.
\end{description}

\noindent
Next we address the case $\phi' = \diam\tau\phi''$.
We prove by induction on $n$:
\begin{center}\begin{tabular}{l}
if \plat{$P\vdash_{\it ws}p_i\trans{\tau}p_{i+1}$}
for all $i\mathbin\in\{0,\ldots,n\}$, with $\rho(t)=p_0$ and $p_{n+1}\models\phi''$,\\
then there is a $\psi\in t^{-1}_{\rm dr}(\eps\diam\tau\phi'')$ with $\rho(x)\models\psi(x)$ for all $x\in \var(t)$.
\end{tabular}\end{center}

\begin{description}
\item[$\boldsymbol{n=0}$] 
Since \plat{$P \vdash_{\it ws} \rho(t) \trans{\tau} p_1$} with $p_1 \models \phi''$,
by Prop.\ \ref{prop:ruloid} there is a $P$-ruloid $\frac{H}{t \trans{\tau} u}$
and a {\valuation} $\rho'$ with $P \vdash_{\it ws} \rho'(\mu)$ for
$\mu\mathbin\in H$, \mbox{$\rho'(t)\mathbin=\rho(t)$}, i.e.\ $\rho'(x)
\mathbin= \rho(x)$ for all $x\mathbin\in \var(t)$, and $\rho'(u)\mathbin=p_1$. Since
\mbox{$\rho'(u) \models \phi''$}, by induction on formula construction
there is a $\chi \in u^{-1}_{\rm dr} (\phi'')$ with
\mbox{$\rho'(z) \models \chi(z)$} for each $z \mathbin\in \var(u)$. Moreover, by
Lem.~\ref{lem:psi1-dr}, $\rho'(z) \models \chi(z)\mathbin\equiv\top$ for each $z
\mathbin{\notin} \var(u)$.

Define $\psi \in t^{-1}_{\rm dr}(\eps\diam\tau\phi'')$ as in
Def.~\ref{def:decomposition-dr}.4b(iii), using $\frac{H}{t \trans{\tau} u}$ and $\chi$.
 Let $x \in \var(t)$. For each\vspace{1pt}
\plat{$x\trans{\beta}y\in H$}, \plat{$P \vdash_{\it ws} \rho'(x)
  \trans{\beta} \rho'(y) \models \chi(y)$}, so $\rho'(x)\models\diam{\beta}\chi(y)$,
\vspace{1pt}and thus $\rho'(x)\models\eps\diam{\beta}\chi(y)$. Moreover, for each
\plat{$x\ntrans{\gamma}\;\in H^{s-}\subseteq H$}, \plat{$P \vdash_{\it ws} \rho'(x)
  \ntrans{\gamma}$}, so the consistency of $\vdash_{\it ws}$ yields
$P\not\vdash_{\it ws}\rho'(x)\trans{\gamma}q$ for all $q$, and thus
$\rho'(x) \models \neg\diam{\gamma}\top$. Hence $\rho(x) = \rho'(x)\models \psi(x)$.

\item[$\boldsymbol{n>0}$] 
Since \plat{$P\vdash_{\it ws}\rho(t)\trans{\tau}p_1$}, by Prop.\
\ref{prop:ruloid} there is a $P$-ruloid $\frac{H}{t\trans{\tau}u}$ and
a {\valuation} $\rho'$ with $P\vdash_{\it ws}\rho'(\mu)$ for all $\mu\in H$,
$\rho'(t)=\rho(t)$, i.e.\ $\rho'(x)=\rho(x)$ for all $x\mathbin\in \var(t)$,
and $\rho'(u)=p_1$. We distinguish two cases.

\item[{\sc Case 1:}]
$\frac{H}{t\trans{\tau}u}$ is $\Gamma$-impatient.
Since \mbox{$\rho'(u) \models \eps\phi''$}, by induction on formula construction
there is a $\chi \in u^{-1}_{\rm dr} (\eps\phi'')$ with
\mbox{$\rho'(z) \models \chi(z)$} for each $z \mathbin\in \var(u)$. Moreover, by
Lem.~\ref{lem:psi1-dr}, $\rho'(z) \models \chi(z)\mathbin\equiv\top$ for each
$z \mathbin{\notin} \var(u)$.
Define $\psi \in t^{-1}_{\rm dr}(\eps\diam\tau\phi'')$ as in
Def.~\ref{def:decomposition-dr}.4b(ii), using \plat{$\frac{H}{t \trans{\tau} u}$} and $\chi$. Let $x \in \var(t)$.
% For\vspace{1pt} \plat{$x\trans{\beta}y\in H$}, \plat{$P \vdash_{\it ws} \rho'(x)
%   \trans{\beta} \rho'(y) \models \chi(y)$}, so $\rho'(x)\models\diam{\beta}\chi(y)$,
% and thus certainly $\rho'(x)\models\eps\diam{\beta}\chi(y)$. Moreover, for each
% \plat{$x\ntrans{\gamma}\;\in H^{s-}\subseteq H$}, \plat{$P \vdash_{\it ws} \rho'(x)
%   \ntrans{\gamma}$}, so the consistency of $\vdash_{\it ws}$ yields
% $P\not\vdash_{\it ws}\rho'(x)\trans{\gamma}q$ for all $q$, and thus
% $\rho'(x) \models \neg\diam{\gamma}\top$. Hence $\rho(x) = \rho'(x)\models \psi(x)$.
That $\rho(x) = \rho'(x)\models \psi(x)$ follows exactly as in the case $n=0$ above.

\item[{\sc Case 2:}]
$\frac{H}{t\trans{\tau}u}$ is $\Gamma$-patient. Since $P \vdash_{\it ws}
\rho'(u)\mathbin=p_1\trans{\tau}\cdots\trans{\tau}p_{n+1}\models\phi''$,
by induction on $n$ there is a $\chi\in u^{-1}_{\rm dr}(\eps\diam\tau\phi'')$ with
$\rho'(z)\models\chi(z)$ for each $z\in \var(u)$.
Using that $t$ is univariate, $H$ must be of the form
\plat{$\{x_0\trans{\tau}y_0\}$}, with $u\mathbin=t[y_0/x_0]$, and the
unique occurrence of $y_0$ in $u$ being $\Gamma$-liquid.
Let $\sigma: V\mathbin\rightarrow V$ be the bijection that swaps $x_0$
  and $y_0$, so that $u=\sigma(t)$.  According to
  Lem.~\ref{lem:alpha}, there is a $\chi'\in t^{-1}_{\rm dr}(\eps\diam\tau\phi'')$
  satisfying $\chi'(x) \equiv \chi(\sigma(x))$ for all $x \mathbin\in V$.
  For each $x\in \var(t)\backslash\{x_0\}$,
  $\rho(x)=\rho'(x)\models\chi(x)\equiv\chi'(x)$, so
  $\rho(x)\models\chi'(x)$.  Furthermore, \plat{$P\vdash_{\it
  ws}\rho'(x_0)\trans{\tau}\rho'(y_0)\models\chi(y_0)$}.
  By Lem.~\ref{lem:psi2},
  $\chi(y_0)\mathbin\equiv\eps\chi(y_0)$. Hence
  $\rho(x_0)=\rho'(x_0)\models\chi(y_0)\equiv\chi'(x_0)$, so
  $\rho(x_0)\models\chi'(x_0)$.
\end{description}
\vspace{2mm}

($\Leftarrow$) Let $\psi\in t^{-1}_{\rm dr}(\eps\phi')$ with
$\rho(x)\models\psi(x)$ for all $x\in\var(t)$. The case where $\psi$ is defined according
to Def.~\ref{def:decomposition-dr}.4a is identical to the treatment of this case in the
proof of Thm.~\ref{thm:decomposition}. We focus on the case where $\psi$ is defined according 
to Def.~\ref{def:decomposition-dr}.4b. Then there are either (case~4b(i)) a $\Gamma$-impatient $P$-ruloid
$\frac{H}{t\trans{\tau}u}$ and a $\chi\in u^{-1}_{\rm dr}(\eps\phi')$,
or (case~4b(ii) with $\phi'=\eps\phi''$) a $\Gamma$-impatient $P$-ruloid
$\frac{H}{t\trans{\tau}u}$ and a $\chi\in u^{-1}_{\rm dr}(\eps\phi'')$,
or (case 4b(iii) with $\phi'=\diam{\alpha}\phi''$) a $P$-ruloid $\frac{H}{t\trans{\alpha}u}$ and a $\chi\in u^{-1}_{\rm dr}(\phi'')$,
such that $\psi$ is defined by
\[
\psi(x)=\left\{
  \begin{array}{ll}
  \eps\psi'(x) & \mbox{if $x$ occurs $\Gamma$-liquid in $t$}\vspace{1mm}\\
  \psi'(x) & \mbox{if $x$ occurs $\Gamma$-frozen in $t$}\vspace{1mm}\\
  \top & \mbox{if $x\notin\var(t)$}
  \end{array}
\right.
\]
where
\[
\psi'(x)~=~\chi(x)\ \land\
          \!\!\!\!\!\bigwedge_{x\trans{\beta}y\in H^+}\!\!\!\!\!\eps\diam{\beta}\chi(y)
          \ \land\ \!\!\!\!\! \bigwedge_{x\ntrans{\gamma}\in
                H^{s-}}\!\!\!\!\!\neg\diam{\gamma}\top\enspace .
\]
So for each $x$ that occurs $\Gamma$-liquid in $t$, $\rho(x)\models\eps\psi'(x)$,
i.e.\ for some $p_x$ we have {$P\vdash_{\it ws}\rho(x)\epsarrow p_x\models\psi'(x)$}.
Moreover, $\rho(x)\models\psi'(x)$ for each $x$ that occurs $\Gamma$-frozen in $t$.
Define $\rho'(x)=p_x$ if $x$ occurs $\Gamma$-liquid in $t$, and
$\rho'(x)=\rho(x)$ otherwise. Since $P$ is $\Gamma$-patient,
$P\vdash_{\it ws}\rho(t)\epsarrow \rho'(t)$. Furthermore, $\rho'(x)\models\psi'(x)$
for all $x\in \var(t)$ implies: $\rho'(x)\models\chi(x)$ for all $x\in \var(t)$;
for each $x \trans{\beta}y \in H^+$ we have $\rho'(x)\models\eps\diam{\beta}\chi(y)$,
i.e.\ $P \vdash_{\it ws} \rho'(x) \epsarrow\trans{\beta} p_y \models \chi(y)$
for some $p_y$; and for each $\mathord{x\ntrans{\gamma}}\in H^{s-}$
we have $P \not\vdash_{\it ws} \rho'(x) \trans{\gamma} q$ for all $q$,
so $P \vdash_{\it ws} \rho'(x) \ntrans{\gamma}$ by the completeness of $P$.
Define $\rho''(x) = \rho'(x)$ and $\rho''(y) = p_y$ for all $x \mathbin\in \var(t)$
and $x\trans{\beta}y \in H$.
Then $\rho''(z) \models \chi(z)$ for all $z\mathbin\in\var(u)$.
Moreover, $P \vdash_{\it ws} \rho''(x) \epsarrow\trans\beta \rho''(y)$
for each premise $x\trans\beta y$ in $H^+$, whereas
$P \vdash_{\it ws} \rho''(x)\!\ntrans\gamma$ for each premise $x\!\ntrans\gamma$ in $H^{s-}$.
\begin{description}
\item[{\sc Case} {\rm 4b(i):}]
Prop.~\ref{prop:delay-resistance} yields
$P \vdash_{\it ws} \rho''(t) \epsarrow \trans{\tau} \rho''(u)$.
By induction on the construction of $\psi$, $\rho''(u) \models \eps\phi'$.
Since \nietplat{$P\vdash_{\it ws}\rho(t)\epsarrow \rho'(t)$} and $\rho'(t)=\rho''(t)$,
it follows that $\rho(t) \models \eps\phi'$.
\item[{\sc Case} {\rm 4b(ii):}]
Prop.~\ref{prop:delay-resistance} yields
$P \vdash_{\it ws} \rho''(t) \epsarrow \trans{\tau} \rho''(u)$.
By induction on formula structure, $\rho''(u) \models \eps\phi''$.
Since \nietplat{$P\vdash_{\it ws}\rho(t)\epsarrow \rho'(t)$} and $\rho'(t)=\rho''(t)$,
it follows that $\rho(t) \models \eps\diam\tau\phi''$.
\item[{\sc Case} {\rm 4b(iii):}]
Prop.~\ref{prop:delay-resistance} yields
$P \vdash_{\it ws} \rho''(t) \epsarrow \trans{\alpha} \rho''(u)$.
By induction on formula structure, $\rho''(u) \models \phi''$.
Since \nietplat{$P\vdash_{\it ws}\rho(t)\epsarrow \rho'(t)$} and $\rho'(t)=\rho''(t)$, it follows that $\rho(t) \models \eps\diam{\alpha}\phi''$.
\qed
\end{description}
\end{proof}

\section{Rooted delay and weak bisimilarity as a congruence}\label{sec:congruence}

A behavioural equivalence $\sim$ is a {\em congruence} for a function symbol
$f$ defined on an LTS if $p_i\sim q_i$ for all $i\in\{1,\ldots,\ar(f)\}$
implies that $f(p_1,\ldots,p_{\ar(f)})\sim f(q_1,\ldots,q_{\ar(f)})$.
We call $\sim$ a congruence \emph{for a TSS} $(\Sigma,R)$, if it is a
congruence for all function symbols from the signature $\Sigma$ with respect to the LTS
generated by $(\Sigma,R)$. This is the case if for any open term
$t \in \mathbb{T}(\Sigma)$ and
any closed substitutions $\rho,\rho':V \rightarrow \mathbb{T}$ we have that
$$\forall x\in\var(t):~~~ \rho(x) \sim \rho'(x) ~~\Rightarrow~~ \rho(t)
\sim \rho'(t)\;.$$
A {\em congruence format} for $\sim$ is a list of syntactic
restrictions on TSSs, such that $\sim$ is guaranteed to be a
congruence for any TSS satisfying these restrictions.

We proceed to apply the decomposition method from the previous section to derive congruence formats
for delay bisimulation and rooted delay bisimulation semantics. The idea behind the construction of these congruence formats is that 
it must be guaranteed that a formula from the characterising logic of the equivalence
under consideration is always decomposed into formulas from this same logic. We prove that
the delay bisimulation format guarantees that a formula from $\IO{d}$ is always decomposed into
formulas from $\IO{d}^\equiv$ (see Prop.~\ref{prop:rdelay-rdelay}).
Likewise, the rooted delay bisimulation format guarantees that a formula from $\IO{rd}$ is always
decomposed into formulas from $\IO{rd}^\equiv$ (see Prop.~\ref{prop:rooted_delay_preservation}).
This implies the desired congruence results (see Thm.~\ref{thm:delay-congruence} and Thm.~\ref{thm:rooted-delay-congruence}, respectively).

At the end of this section it is sketched how these results can be transposed to (rooted) weak bisimilarity,
by adding one condition to the congruence format for (rooted) delay bisimilarity.

\subsection{Congruence format for rooted delay bisimilarity}\label{sec:formats}

We recall the notion of a rooted branching bisimulation safe rule, which underlies
the rooted branching bisimulation format from \cite{FvGdW12}. The congruence format
for rooted delay bisimilarity is obtained by additionally requiring delay resistance.

We assume two predicates on arguments of function symbols from \cite{Fok00,FvGdW12}.
The predicate $\Lambda$ marks arguments that contain processes that
have started executing (but may currently be unable to execute).
The predicate $\aleph$ marks arguments that contain
processes that can execute immediately.
For example, in process algebra, $\Lambda$ and $\aleph$ hold for the arguments of the
merge $t_1\|t_2$, and for the first argument of sequential composition $t_1\cdot t_2$;
they can contain processes that started to execute in the past, and these processes
can continue their execution immediately. On the other hand, $\Lambda$ and $\aleph$
typically do not hold for the second argument of sequential composition; it contains a process
that did not yet start to execute, and cannot execute immediately (in absence of the empty process).
$\Lambda$ does not hold and $\aleph$ holds for the arguments of
alternative composition $t_1+t_2$; they contain processes that did not yet
start to execute, but that can start executing immediately.

\begin{definition}\label{def:rooted_branching_bisimulation_safe}\rm\cite{FvGdW12}\,
A standard ntytt rule $r=\frac{H}{t\trans\alpha u}$ is {\em rooted branching bisimulation safe}
w.r.t.\ $\aleph$ and $\Lambda$ if it satisfies the following conditions.
Let $x\in\var(t)$.
\begin{enumerate}
\vspace{-.5ex}
\item \label{rhs}
Right-hand sides of positive premises occur only $\Lambda$-liquid in $u$.
\vspace{-.5ex}
\item \label{Lambda}
If $x$ occurs only $\Lambda$-liquid in $t$, then $x$ occurs only $\Lambda$-liquid in $r$.
\vspace{-.5ex}
\item \label{aleph}
If $x$ occurs only $\aleph$-frozen in $t$, then $x$ occurs only $\aleph$-frozen in $H$.
\vspace{-.5ex}
\item \label{main}
If $x$ has exactly one $\aleph$-liquid occurrence in $t$, which is also $\Lambda$-liquid,
then $x$ has at most one $\aleph$-liquid occurrence in $H$, which must be in a positive premise.
If moreover this premise is labelled $\tau$, then $r$ must be $\aL$-patient.
\setcounter{saveenumi}{\theenumi}
\vspace{-.5ex}
\end{enumerate}
\end{definition}

\begin{definition}\rm\label{def:rooted_delay_bisimulation_format}
A standard TSS $P$ is in {\em rooted delay bisimulation format} if it
is in ready simulation format and \dr,
and, for some $\aleph$ and $\Lambda$, it is $\aL$-patient and all its transition rules are
rooted branching bisimulation safe w.r.t.\ $\aleph$ and $\Lambda$.

This TSS is in {\em delay bisimulation format} if moreover $\Lambda$ is universal.
\end{definition}

\begin{remark}
If a standard TSS $P$ is in rooted delay bisimulation format, then there
are smallest predicates $\aleph_0$ and $\Lambda_0$ such that $P$ is in
rooted delay bisimulation format w.r.t.\ $\aleph_0$ and $\Lambda_0$.
Namely the $\Lambda_0$-liquid arguments are \emph{generated} by
conditions \ref{rhs} and \ref{Lambda} of Def.~\ref{def:rooted_branching_bisimulation_safe};
they are the smallest collection of arguments such that these two requirements are satisfied.
Likewise the $\aleph_0$-liquid arguments are generated by condition \ref{aleph},
which can be read as ``If $x$ occurs $\aleph$-liquid in $H$, then the
unique occurrence of $x$ in $t$ is $\aleph$-liquid.'' For any
standard TSS $P$ in ready simulation format, $\aleph_0$ and $\Lambda_0$ are determined in this way,
and whether $P$ is in rooted delay bisimulation format then depends solely on whether it is
delay resistant and $\aleph_0\mathord\cap\Lambda_0$-patient, and condition \ref{main} of
Def.~\ref{def:rooted_branching_bisimulation_safe} is satisfied by all rules in $P$.
\end{remark}

\subsection{Preservation of syntactic restrictions}
\label{sec:preservation}

In the definition of modal decomposition, we did not use the rules from the original {\dra}
standard TSS $P$, but the $P$-ruloids. Therefore we must verify that if $P$
is in rooted delay bisimulation format, then the $P$-ruloids are rooted branching
bisimulation safe (Prop.~\ref{prop:preservation_bra_bisimulation_safe}).
In the proof of this preservation result, rules with a negative conclusion
play an important role. For this reason, the notion of rooted branching bisimulation
safety is extended to non-standard rules.

\begin{definition}\label{def:nonstandard-branching}\rm\cite{FvGdW12}\,
An ntytt rule $r=\frac{H}{t\ntrans\alpha}$ is
{\em rooted branching bisimulation safe} w.r.t.\ $\aleph$ and $\Lambda$ if it satisfies conditions
\ref{Lambda} and \ref{aleph} of Def.~\ref{def:rooted_branching_bisimulation_safe}.
\end{definition}

\begin{proposition}{\rm\cite{FvGdW12}}
\label{prop:preservation_bra_bisimulation_safe}
Let $P$ be a standard TSS in ready simulation format, in which each transition rule is
rooted branching bisimulation safe w.r.t.\ $\aleph$ and $\Lambda$.
Then each $P$-ruloid is rooted branching bisimulation safe w.r.t.\ $\aleph$ and $\Lambda$.
\end{proposition}
The following lemma is a crucial step in the proof of Prop.~\ref{prop:preservation_bra_bisimulation_safe}.

\begin{lemma}{\rm\cite{FvGdW12}}
\label{lem:preservation_branching_bisimulation_safe}
Let $P$ be a TSS in decent ntyft format, in which each transition rule is
rooted branching bisimulation safe w.r.t.\ $\aleph$ and $\Lambda$.
Then any ntytt rule irredundantly provable from $P$ is
rooted branching bisimulation safe w.r.t.\ $\aleph$ and $\Lambda$.%
\end{lemma}

\subsection{Preservation of modal characterisations}\label{sec:bra}

Consider a standard TSS that is in rooted delay bisimulation format, w.r.t.\ some $\aleph$ and $\Lambda$.
Def.~\ref{def:decomposition-dr} yields decomposition mappings $\psi\mathbin\in
t^{-\!1}_{\rm dr}\hspace{-1pt}(\phi)$, with $\Gamma=\aL$. In this section we will first prove that
if $\phi\in\IO{d}$, then $\psi(x)\in\IO{d}^{\equiv}$ if $x$ occurs only
$\Lambda$-liquid in $t$. (That is why in the delay bisimulation
format, $\Lambda$ must be universal.) Next we will prove that if
$\phi\in\IO{rd}$, then $\psi(x)\in\IO{rd}^{\equiv}$ for all variables
$x$. From these preservation results we will, in
Sect.~\ref{sec:congruenceresults}, deduce the promised congruence
results for delay bisimilarity and rooted delay bisimilarity, respectively.

\begin{proposition}\label{prop:rdelay-rdelay}
Let $P$ be a {\dra}, $\aL$-patient standard TSS in ready simulation format,
in which each transition rule is rooted branching bisimulation safe w.r.t.\ $\aleph$ and $\Lambda$.
For any term $t$ and variable $x$ that occurs only $\Lambda$-liquid in $t$:
\[
\phi\in\IO{d}\ \Rightarrow\ \forall\psi\in t^{-1}_{\rm dr}(\phi):\psi(x)\in\IO{d}^\equiv\enspace .
\]
\end{proposition}

\begin{proof}
We apply simultaneous induction on the structure of $\phi\in\IO{d}$
and the construction of $\psi$. Let $\psi\in t^{-1}_{\rm dr}(\phi)$, and let
$x$ occur only $\Lambda$-liquid in $t$. First we treat the case where
$t$ is univariate. If $x\mathbin{\notin}\var(t)$, then by
Lem.~\ref{lem:psi1-dr}, $\psi(x)\mathbin{\equiv}\top\mathbin\in\IO{d}^\equiv$.
So suppose $x$ has exactly one, $\Lambda$-liquid occurrence in $t$.

We need to consider the four possible syntactic forms of $\phi$ in the BNF grammar
of $\IO{d}$ in Def.~\ref{def:modal-characterisations}.
\begin{itemize}
\item $\phi=\bigwedge_{i\in I}\phi_i$ with $\phi_i\in\IO{d}$ for each $i\in I$. By Def.~\ref{def:decomposition}.\ref{dec1}, $\psi(x)=\bigwedge_{i\in I}\psi_i(x)$ with $\psi_i\in t^{-1}_{\rm dr}(\phi_i)$ for each $i\in I$. By  induction on formula structure, $\psi_i(x)\in\IO{d}^\equiv$ for each $i\in I$, so $\psi(x)\in\IO{d}^\equiv$.\vspace{1mm}

\item $\phi=\neg\phi'$ with $\phi'\in\IO{d}$. By Def.~\ref{def:decomposition}.\ref{dec2}, there is a function $h:t^{-1}_{\rm dr}(\phi')\rightarrow \var(t)$ such that $\psi(x)=\bigwedge_{\chi\in h^{-1}(x)}\neg\chi(x)$. By  induction on formula structure, $\chi(x)\mathbin\in\IO{d}^\equiv$ for each $\chi\in h^{-1}(x)$, so $\psi(x)\in\IO{d}^\equiv$.\vspace{1mm}

\item
$\phi=\eps\phi'$ with $\phi'\mathbin\in\IO{d}$ (which implies that $\phi'$ is not of the form $\diam{a}\phi''$). According to Def.~\ref{def:decomposition-dr}.4, we can distinguish two cases.\vspace{1mm}

\begin{description}
\item[{\sc Case 1:}] $\psi(x)$ is defined on the basis of Def.~\ref{def:decomposition-dr}.4a.
  Then either $\psi(x)=\eps\chi(x)$ or $\psi(x)=\chi(x)$ for some $\chi\in t^{-1}_{\rm dr}(\phi')$.
  By  induction on formula structure, $\chi(x)\in\IO{d}^\equiv$. So $\psi(x)\in\IO{d}^\equiv$.\vspace{1mm}

\item[{\sc Case 2:}] $\psi(x)$ is defined on the basis of
  Def.~\ref{def:decomposition-dr}.4b, employing an $\aL$-impatient $P$-ruloid
  \plat{$\frac{H}{t\trans{\tau}u}$} and a $\chi\in u^{-1}_{\rm dr}(\eps\phi')$.
  The rest of this case proceeds exactly as the next case below.
\end{description}

\item
$\phi=\eps\diam{a}\phi'$ with $\phi'\in\IO{d}$. Then $\psi(x)$ is defined on the basis of Def.~\ref{def:decomposition-dr}.4(b),
employing either an $\aL$-impatient $P$-ruloid $\frac{H}{t\trans{\tau}u}$ and a $\chi\in u^{-1}_{\rm dr}(\eps\diam{a}\phi')$,
or a $P$-ruloid $\frac{H}{t\trans{a}u}$ and a $\chi\in u^{-1}_{\rm dr}(\phi')$.
By Prop.~\ref{prop:preservation_bra_bisimulation_safe} we can assume that $\frac{H}{t\trans{\tau}u}$ or $\frac{H}{t\trans{a}u}$
is rooted branching bisimulation safe w.r.t.\ $\aleph$ and $\Lambda$. Since the occurrence of $x$ in $t$ is
$\Lambda$-liquid, by condition \ref{Lambda} of Def.~\ref{def:rooted_branching_bisimulation_safe}, $x$ occurs only
$\Lambda$-liquid in $u$. Moreover, by condition~\ref{rhs} of Def.~\ref{def:rooted_branching_bisimulation_safe}, variables in ${\it rhs}(H)$
occur only $\Lambda$-liquid in $u$. So by induction on the construction of $\psi$ or on formula size, $\chi(x)\in\IO{d}^\equiv$,
and $\chi(y)\in\IO{d}^\equiv$ for each \nietplat{$x\trans{\beta}y\in H$}. We distinguish two cases.\vspace{1mm}

\begin{description}
\item[{\sc Case 1:}]
The occurrence of $x$ in $t$ is $\aleph$-frozen. By condition \ref{aleph} of
Def.~\ref{def:rooted_branching_bisimulation_safe}, $x$ does not occur in $H$.
Hence $\psi(x)=\chi(x)\in\IO{d}^\equiv$.\vspace{1mm}
\item[{\sc Case 2:}]
The occurrence of $x$ in $t$ is $\aleph$-liquid. Since moreover the occurrence of $x$ in $t$ is $\Lambda$-liquid,
by condition \ref{main} of Def.~\ref{def:rooted_branching_bisimulation_safe},
$x$ is the left-hand side of at most one premise in $H$, which is positive. And since
$\frac{H}{t\trans{\tau}u}$ or $\frac{H}{t\trans{a}u}$ is $\aL$-impatient, this positive premise does not carry the label $\tau$.
Hence $\psi(x)=\eps\chi(x)$ or $\psi(x)=\eps(\chi(x)\wedge\eps\diam{b}\chi(y))$ with $b\in A$ and \plat{$x\trans b y\in H$}.
In either case, $\psi(x)\in\IO{d}^\equiv$.
\end{description}
\end{itemize}
Finally, we treat the case where $t$ is not univariate. Then $t=\sigma(u)$ for some univariate term $u$ and non-injective mapping $\sigma:\var(u)\rightarrow V$. By Def.~\ref{def:decomposition}.\ref{dec6}, $\psi(x)=\bigwedge_{z\in\sigma^{-1}(x)}\chi(z)$ for some $\chi\in u^{-1}_{\rm dr}(\phi)$. Since $u$ is univariate, and for each $z\in\sigma^{-1}(x)$ the occurrence in $u$ is $\Lambda$-liquid, $\chi(z)\in\IO{d}^\equiv$ for all $z\in\sigma^{-1}(x)$. Hence $\psi(x)\in\IO{d}^\equiv$.
\qed
\end{proof}

\begin{proposition}\label{prop:rooted_delay_preservation}
Let $P$ be a {\dra}, $\aL$-patient standard TSS in ready simulation format,
in which each transition rule is rooted branching bisimulation safe w.r.t.\ $\aleph$ and $\Lambda$.
For any term $t$ and variable $x$:
\[
\phi\in\IO{rd}\ \Rightarrow\ \forall\psi\in t^{-1}_{\rm dr}(\phi):\psi(x)\in\IO{rd}^\equiv\enspace .
\]
\end{proposition}

\begin{proof}
We apply simultaneous induction on the structure of $\phi\in\IO{rd}$ and the construction of $\psi$.
Let $\psi\in t^{-1}_{\rm dr}(\phi)$. We restrict
attention to the case where $t$ is univariate; the general case then
follows just as at the end of the proof of Prop.~\ref{prop:rdelay-rdelay}. If
$x\notin\var(t)$, then by Lem.~\ref{lem:psi1-dr}, $\psi(x)\equiv\top\in\IO{rd}^\equiv$. So suppose
$x$ occurs once in $t$.

\begin{itemize}
\item The cases $\phi=\bigwedge_{i\in I}\phi_i$ and
$\phi=\neg\phi'$ proceed as in the proof of Prop.~\ref{prop:rdelay-rdelay}, replacing $\IO{d}$ by $\IO{rd}$.

\item
$\phi=\eps\diam\alpha\phi'$ with $\phi'\in\IO{d}$.
According to Def.~\ref{def:decomposition-dr}.4(b),
\[
\psi(x)=\left\{
\begin{array}{ll}
\displaystyle{\eps\left(\chi(x)\ \land\
\!\!\!\!\!\bigwedge_{x\trans{\beta}y\in H^+}\!\!\!\!\!\eps\diam{\beta}\chi(y)\ \land\ \!\!\!\!\! \bigwedge_{x\ntrans{\gamma}\in
                H^{s-}}\!\!\!\!\!\neg\diam{\gamma}\top\right)} & \mbox{\begin{tabular}{l}if
$x$ occurs\\ $\aL$-liquid in $t$\end{tabular}}\\
\displaystyle{\chi(x)\ \land\ \!\!\!\!\!\bigwedge_{x\trans{\beta}y\in H^+}\!\!\!\!\!\eps\diam{\beta}\chi(y)\ \land\ \!\!\!\!\! \bigwedge_{x\ntrans{\gamma}\in
                H^{s-}}\!\!\!\!\!\neg\diam{\gamma}\top} &
\mbox{\begin{tabular}{l}if $x$ occurs\\ $\aL$-frozen in $t$\end{tabular}}
\end{array}
\right.
\]
where there is a $P$-ruloid $\frac{H}{t\trans{\alpha}u}$ with $\chi\in u^{-1}_{\rm dr}(\phi')$,
or there is an $\aleph\mathord\cap\Lambda$-impatient $P$-ruloid \plat{$\frac{H}{t\trans{\tau}u}$}
with $\chi\mathbin\in u^{-1}_{\rm dr}(\xi)$. Here $\xi=\eps\diam\alpha\phi'$ if $\alpha\neq\tau$
and $\xi=\eps\phi'$ if $\alpha=\tau$. In either case $\xi\in\IO{d}$.

If $H^{s-}$ contains a negative premise with left-hand side $x$, then by
the definition of $H^{s-}$, it also contains $x\!\ntrans{\tau}$.
Clearly $p\models \neg\diam{\tau}\top\,\land\,\neg\diam{\gamma}\top$ if and only if $p\models \neg\diam{\tau}\top\,\land\,\neg\eps\diam{\gamma}\top$,
and moreover $p\models \neg\diam{\tau}\top$ if and only if $p\models \neg\eps\diam{\tau}\top$, for any process $p$ and action $\gamma$.
This implies that the conjuncts $\bigwedge_{x\ntrans{\gamma}\in H^{s-}}\neg\diam{\gamma}\top$ in $\psi(x)$ can be replaced by
\[
\bigwedge_{x\ntrans{\gamma}\in H^{s-}}\!\!\!\!\!\neg\eps\diam{\gamma}\top\enspace .
\]
By induction on the construction of $\psi$ or on formula size, $\chi(x)\in\IO{rd}^\equiv$.
By Prop.~\ref{prop:preservation_bra_bisimulation_safe} we can assume that \plat{$\frac{H}{t\trans{\tau}u}$} or \plat{$\frac{H}{t\trans\alpha u}$}
is rooted branching bisimulation safe w.r.t.\ $\aleph$ and $\Lambda$; so by condition \ref{rhs}
of Def.~\ref{def:rooted_branching_bisimulation_safe}, variables in ${\it rhs}(H)$ %\vspace{-2pt}
occur only $\Lambda$-liquid in $u$. Hence, by
Prop.~\ref{prop:rdelay-rdelay}, $\chi(y)\in\IO{d}^\equiv$ for each $x\trans{\beta}y\in H$.
In case the occurrence of $x$ in $t$ is $\aL$-frozen,
this immediately yields $\psi(x)\in\IO{rd}^\equiv$.

So suppose the occurrence of $x$ in $t$ is $\aL$-liquid.
By condition \ref{Lambda} of Def.~\ref{def:rooted_branching_bisimulation_safe},
$x$ occurs only $\Lambda$-liquid in $u$, so by Prop.~\ref{prop:rdelay-rdelay}, $\chi(x)\in\IO{d}^\equiv$.
And by condition \ref{main} of Def.~\ref{def:rooted_branching_bisimulation_safe},
$x$ is the left-hand side of at most one premise in $H$, which is positive. So either %\vspace{-2pt}
$\psi(x)=\eps\chi(x)$, or $\psi(x)=\eps(\chi(x)\wedge\eps\diam\beta\chi(y))$ with $x\trans{\beta}y\in H$.
In the first case, and in the second case with $\beta\neq \tau$,
this immediately yields $\psi(x)\in\IO{d}^\equiv \subset \IO{rd}^\equiv$.

Consider the second case with $\beta=\tau$. By condition \ref{main} of Def.~\ref{def:rooted_branching_bisimulation_safe},
\plat{$\frac{H}{t\trans{\alpha}u}$ or $\frac{H}{t\trans{\tau}u}$} is $\aL$-patient.
So clearly this $P$-ruloid it is of the form $\frac{x\trans\tau y}{C[x]\trans\tau C[y]}$.
%(see Remark~\ref{rem:patient} in \cite{FvGdW12}).
Since $t$ is univariate, it follows that $x\mathbin{\notin}\var(u)$.
Hence, by Lem.~\ref{lem:psi1-dr}, $\chi(x)\equiv\top$. Thus
$\psi(x) \equiv\eps\diam{\beta}\chi(y)\in\IO{rd}^\equiv$.\vspace{1mm}

\item
$\phi\in\IO{d}$. The cases $\phi=\bigwedge_{i\in I}\phi_i$ and
$\phi=\neg\phi'$ proceed as in the proof of Prop.~\ref{prop:rdelay-rdelay},
replacing $\IO{d}^\equiv$ by $\IO{rd}^\equiv$, and
the case $\phi\mathbin=\eps\diam{a}\phi'$ was already treated above.
The remaining case is $\phi=\eps\phi''$ with $\phi''\in\IO{d}$ not of the form $\diam{\alpha}\phi'$.
If the occurrence of $x$ in $t$ is $\Lambda$-liquid, then by Prop.~\ref{prop:rdelay-rdelay},
$\psi(x)\in\IO{d}^\equiv\subset\IO{rd}^\equiv$. So we can assume that this occurrence is
$\Lambda$-frozen. According to Def.~\ref{def:decomposition-dr}.4 we can distinguish two cases.

\begin{description}
\item[{\sc Case 1:}] $\psi(x)$ is defined on the basis of case 4a.
  Then $\psi(x)=\chi(x)$ for some $\chi\in t^{-1}_{\rm dr}(\phi'')$.
  By induction on formula structure, $\psi(x)=\chi(x)\in\IO{rd}^\equiv$.

\item[{\sc Case 2:}] $\psi(x)$ is defined on the basis of case 4b, employing an $\aleph\mathord\cap\Lambda$-impatient $P$-ruloid
  \plat{$\frac{H}{t\trans{\tau}u}$} and a $\chi\in u^{-1}_{\rm dr}(\eps\phi'')$. Then
  $$\psi(x) = \chi(x)\ \land\ \!\!\!\!\!\!\bigwedge_{x\trans{\beta}y\in H^+}\!\!\!\!\!\eps\diam{\beta}\chi(y)\ \land\
                \!\!\!\!\! \bigwedge_{x\ntrans{\gamma}\in H^{s-}}\!\!\!\!\!\neg\diam{\gamma}\top\enspace .$$
  By induction on the construction of $\psi$, $\chi(x)\in\IO{rd}^\equiv$.
  By Prop.~\ref{prop:preservation_bra_bisimulation_safe} we can assume that $\frac{H}{t\trans{\tau}u}$
  is rooted branching bisimulation safe w.r.t.\
  $\aleph$ and $\Lambda$; so by condition \ref{rhs} of Def.~\ref{def:rooted_branching_bisimulation_safe},
  variables in ${\it rhs}(H)$ occur only $\Lambda$-liquid in $u$. %\vspace{-2pt}
  Hence, by Prop.~\ref{prop:rdelay-rdelay}, $\chi(y)\in\IO{d}^\equiv$ for each $x\trans{\beta}y\in H$.
  And similar as in the previous case $\phi=\eps\diam\alpha\phi'$ with $\phi'\in\IO{d}$ we can argue that the conjuncts
  $\neg\diam{\gamma}\top$ in $\psi(x)$ can be replaced by $\neg\eps\diam{\gamma}\top$.
  Hence $\psi(x)\in\IO{rd}^\equiv$.
\qed
\end{description}
\end{itemize}
\end{proof}

\subsection{Congruence for rooted delay bisimilarity}
\label{sec:congruenceresults}

Now we are in a position to prove the promised congruence results for $\bis{d}$ and $\bis{rd}$.

\begin{theorem}\label{thm:delay-congruence}
Let $P$ be a complete standard TSS in delay bisimulation format.
Then $\bis{d}$ is a congruence for $P$.
\end{theorem}

\begin{proof}
By Def.~\ref{def:rooted_delay_bisimulation_format} each rule of $P$
is rooted branching bisimulation safe w.r.t\ some $\aleph$ and the
universal predicate $\Lambda$, and $P$ is \dr, $\aL$-patient and in 
ready simulation format.

Let $\rho,\rho'$ be {\valuation}s and $t$ a term. Suppose that
$\rho(x) \bis{d} \rho'(x)$ for all $x\in\var(t)$;
we need to prove that then $\rho(t) \bis{d} \rho'(t)$.

Let $\rho(t)\models\phi \in \IO{d}$.
By Thm.~\ref{thm:dr-decomposition}, taking $\Gamma=\aL$, there is a $\psi
\in t^{-1}_{\rm dr}(\phi)$ with $\rho(x) \models \psi(x)$ for all $x\in\var(t)$.
Since $x$ occurs $\Lambda$-liquid in $t$ (because $\Lambda$ is universal),
by Prop.~\ref{prop:rdelay-rdelay}, $\psi(x) \in\IO{d}^\equiv$ for all $x\in\var(t)$.
By Thm.~\ref{thm:characterisation}, $\rho(x) \bis{d} \rho'(x)$
implies $\rho(x) \sim_{\IO{d}^\equiv} \rho'(x)$ for all $x \in \var(t)$.
So $\rho'(x)\models\psi(x)$ for all $x \in \var(t)$. Therefore, by Thm.~\ref{thm:dr-decomposition},
$\rho'(t)\models\phi$.  Likewise, $\rho'(t)\models\phi \in \IO{d}$ implies $\rho(t)\models\phi$.
So $\rho(t) \sim_{\IO{d}} \rho'(t)$. Hence, by
Thm.~\ref{thm:characterisation}, $\rho(t) \bis{d} \rho'(t)$.
\qed
\end{proof}

\noindent
We can follow the same approach to prove that the rooted delay bisimulation format
guarantees that $\bis{rd}$ is a congruence.

\begin{theorem}\label{thm:rooted-delay-congruence}
Let $P$ be a complete standard TSS in rooted delay bisimulation format.
Then $\bis{rd}$ is a congruence for $P$.
\end{theorem}

\noindent
The proof of Thm.~\ref{thm:rooted-delay-congruence} is similar to the one of Thm.~\ref{thm:delay-congruence},
except that Prop.~\ref{prop:rooted_delay_preservation} is applied instead of
Prop.~\ref{prop:rdelay-rdelay}; therefore $x$ need not occur $\Lambda$-liquid in $t$,
which is why universality of $\Lambda$ can be dropped.

\subsection{Rooted weak bisimilarity as a congruence}
\label{sec:weak-congruence}

We now proceed to derive a congruence format for rooted weak bisimilarity.
It is obtained from the congruence format from \cite{FvGdW12} for rooted $\eta$-bisimilarity by
additionally requiring delay resistance. The format for rooted $\eta$-bisimilarity in turn is obtained
by strengthening condition \ref{rhs} in the definition of rooted branching bisimulation safeness.

\begin{definition}\label{def:rooted_eta_bisimulation_safe} 
{\rm\cite{FvGdW12}
An ntytt rule $r=\frac{H}{t\trans\alpha u}$ is {\em rooted $\eta$-bisimulation safe}
w.r.t.\ $\aleph$ and $\Lambda$ if it satisfies conditions
\ref{Lambda}--\ref{main} of Def.~\ref{def:rooted_branching_bisimulation_safe}, together with:
\begin{itemize}
\item[$\ref*{rhs}'$.] Right-hand sides of positive premises occur only $\aL$-liquid in $u$.
\end{itemize}
}
\end{definition}

\begin{definition}\label{def:weak_bisimulation_format}
{\rm
A standard TSS is in {\em rooted weak bisimulation format} if it is
in ready simulation format and \dr,
and, for some $\aleph$ and $\Lambda$, it is $\aL$-patient and contains
only rules that are rooted $\eta$-bisimulation safe w.r.t.\ $\aleph$ and $\Lambda$.

This TSS is in {\em weak bisimulation format} if moreover $\Lambda$ is universal.
}
\end{definition}
The proofs that these formats
guarantee that [rooted] weak bisimilarity is a congruence are largely identical to the proofs for the
[rooted] delay bisimulation format. We will therefore only explain where these proofs differ.

For non-standard ntytt rules, the notion of rooted $\eta$-bisimulation safeness coincides with the
notion of rooted branching bisimulation safeness (see Def.~\ref{def:nonstandard-branching}).

\begin{proposition}{\rm\cite{FvGdW12}}
\label{prop:preservation_eta_bisimulation_safe}
Let $P$ be a TSS in ready simulation format, in which each transition rule is
rooted $\eta$-bisimulation safe w.r.t.\ $\aleph$ and $\Lambda$.
Then each $P$-ruloid is rooted $\eta$-bisimulation safe w.r.t.\ $\aleph$ and $\Lambda$.
\end{proposition}
 
\noindent
The proof of the following proposition is very similar to the proof of the corresponding
Prop.~\ref{prop:rdelay-rdelay} for the rooted delay bisimulation format.

\begin{proposition}\label{prop:rweak-rweak}
Let $P$ be a {\dra}, $\aL$-patient standard TSS in ready simulation format,
in which each transition rule is rooted $\eta$-bisimulation safe w.r.t.\ $\aleph$ and $\Lambda$.
For any term $t$ and variable $x$ that occurs only $\Lambda$-liquid in $t$:
\[
\phi\in\IO{w}\ \Rightarrow\ \forall\psi\in t^{-1}_{\rm dr}(\phi):\psi(x)\in\IO{w}^\equiv\enspace .
\]
\end{proposition}

\noindent
In the case $\phi=\eps\diam{a}\eps\phi'$ with $\phi'\in\IO{w}$ the proof of
Prop.~\ref{prop:rweak-rweak} slightly deviates from the case $\phi=\eps\diam{a}\phi'$
with $\phi'\in\IO{d}$ in the proof of Prop.~\ref{prop:rdelay-rdelay}. At the end of the latter case,
%\vspace{-2pt}
in {\sc Case 2} where the occurrence of $x$ in univariate term $t$ is $\aleph$-liquid,
it may be that $\psi(x)=\eps(\chi(x)\wedge\eps\diam{b}\chi(y))$ with \plat{$x\trans b y\in H$}.
The additional observation we make in the proof of Prop.~\ref{prop:rweak-rweak} is that owing to
the stronger condition $\ref*{rhs}'$, $y$ only occurs $\aL$-liquid in $u$.
So according to Lem.~\ref{lem:psi2} with $\Gamma=\aL$, $\chi(y)\equiv\eps\chi(y)$.
Hence, $\psi(x)\equiv\eps(\chi(x)\wedge\eps\diam{b}\eps\chi(y))\in\IO{w}^\equiv$.

The same difference with the proof of Prop.~\ref{prop:rdelay-rdelay} appears in the case
$\phi=\eps\phi'$ with $\phi'\in\IO{w}$, {\sc Case 2}. For the rest, the proofs proceed in exactly the same way.

The proof of the following proposition is very similar to the proof of the corresponding
Prop.~\ref{prop:rooted_delay_preservation} for the rooted delay bisimulation format.

\begin{proposition}\label{prop:rooted_weak_preservation}
Let $P$ be a {\dra}, $\aL$-patient standard TSS in ready simulation format,
in which each transition rule is rooted $\eta$-bisimulation safe w.r.t.\ $\aleph$ and $\Lambda$.
For any term $t$ and variable $x$:
\[
\phi\in\IO{rw}\ \Rightarrow\ \forall\psi\in t^{-1}_{\rm dr}(\phi):\psi(x)\in\IO{rw}^\equiv\enspace .
\]
\end{proposition}

\noindent
Again the only real difference with the proof of Prop.~\ref{prop:rooted_delay_preservation} is that
we need to exploit the stronger condition $\ref*{rhs}'$ of Def.~\ref{def:rooted_eta_bisimulation_safe}:
for each $P$-ruloid $\frac{H}{t\trans{\alpha}u}$, %\vspace{-1pt}
each $y\in{\it rhs}(H)$ can occur only $\aL$-liquid in
$u$; so by Lem.~\ref{lem:psi2} with $\Gamma=\aL$, $\chi\in u^{-1}_{\rm dr}(\eps\phi)$ implies $\chi(y)\mathbin\equiv\eps\chi(y)$.
Moreover, we need to observe that $\neg\eps\diam{\gamma}\top\mathbin\equiv\neg\eps\diam{\gamma}\eps\top\mathbin\in\IO{rw}$.

The proofs of the following congruence theorems for weak bisimilarity are almost identical
to the proofs of the corresponding congruence theorems for delay bisimilarity.

\begin{theorem}
Let $P$ be a complete standard TSS in weak bisimulation format.
Then \mbox{}$\bis{w}$ is a congruence for $P$.
\end{theorem}

\begin{theorem}
\label{thm:rooted-weak-congruence}
Let $P$ be a complete standard TSS in rooted weak bisimulation format.
Then $\bis{rw}$ is a congruence for $P$.
\end{theorem}

\subsection{Counterexamples}
\label{sec:counterexamples}

In \cite{FvGdW12} it was shown that none of the syntactic requirements of the rooted branching bisimulation format
in Def.~\ref{def:rooted_branching_bisimulation_safe} can be omitted, and that the presence of $\aL$-patience rules is crucial.
Here we present a sequence of examples to show that none of the requirements that make up delay resistance
is redundant to guarantee that rooted delay bisimilarity is a congruence.

All TSSs in this section are standard, complete, in ready-simulation format and $\aL$-patient, and their rules are rooted branching bisimulation safe.

\begin{example}
\label{exa:negative}
Let $f$ be a unary function symbol with an $\aleph$-liquid, $\Lambda$-frozen argument, defined by the rule
$$\frac{x\ntrans{a}}{f(x)\trans{b}0}$$
This rule is \pdr, but fails to be \ndr.

Consider the LTS consisting of the transitions $p_0\trans{\tau}p_0$, $p_0\trans{a}0$, $p_1\trans{\tau}q$ and $q\trans{a}0$.
Note that $p_0\bis{rd}p_1$. However, $f(p_0)$ exhibits no transitions, while \plat{$f(p_1)\trans{b}0$}.
So \nietplat{$f(p_0)\notbis{rd}f(p_1)$}. Hence rooted delay bisimilarity is not a congruence.
\end{example}

\begin{example}
Let $f$ be a unary function symbol with an $\aleph$-liquid, $\Lambda$-frozen argument, defined by the rule
$$\frac{x\trans{a}y}{f(x)\trans{b}0}$$
This TSS is \ndr, but not \pdr.

Consider the LTS from Ex.~\ref{exa:negative}.
We have $f(p_0)\trans{b}0$, while the process $f(p_1)$ does not exhibit any transitions.
So $f(p_0)\notbis{rd}f(p_1)$. Hence rooted delay bisimilarity is not a congruence.
\end{example}

\noindent
The following example shows that the requirement that $H^d$ is finite
in Def.~\ref{def:positive} of positive delay resistance is essential.

\begin{example}
Let $A=\{a_k\mid k\in\mathbb{Z}_{>0}\}\cup\{b\}$, and let there be binary function symbols $f_k$ for all $k\in\mathbb{Z}_{>0}$,
of which both arguments are $\aleph$-liquid and only the second argument is $\Lambda$-liquid. They are defined by the rules
$$\frac{x_1\trans{\tau}y}{f_k(x_1,x_2)\trans{\tau}f_\ell(x_1,y)}\mbox{~~(for all $\ell>k$)}\qquad\qquad\frac{x_2\trans{\tau}y}{f_k(x_1,x_2)\trans{\tau}f_k(x_1,y)}\vspace{1mm}$$
$$\frac{\{x_1\trans{a_\ell}y_\ell\mid\ell>k\}\cup\{x_2\trans{a_k}y_k\}}{f_k(x_1,x_2)\trans{b}0}$$
where $\omega_1$, $\omega_2$ and $\omega_3$ are constants with $\omega_2\trans{a_k}0$ and $\omega_3\trans{a_k}0$ for all $k\geq 1$,
and $\omega_1\trans{\tau}\omega_3$ and $\omega_2\trans{\tau}\omega_3$. Clearly, $\omega_1\bis{rd}\omega_2$.

With the exception of the rule with infinitely many premises, its rules are \dr. That one rule is \ndr, and all its premises are delayable.
With regard to Def.~\ref{def:positive} it only violates the
requirement that $H^d$ needs to be chosen finite.
As a result Prop.~\ref{prop:delay-resistance} is violated.
Namely, although there are sequences $\omega_1\epsarrow\trans{a_{\ell}}$ for all $\ell\geq 1$,
there is no sequence \nietplat{$f_{k}(\omega_1,\omega_1)\epsarrow\trans{b}$} for any $k\geq 1$.
On the other hand, \nietplat{$f_{k}(\omega_2,\omega_2)\trans{b}0$} for all $k\geq 1$.
So rooted delay bisimilarity is not a congruence. 
\end{example}

\noindent
The following example shows that the strengthening of condition \ref{rhs} to condition $\ref*{rhs}'$ in
Def.~\ref{def:rooted_eta_bisimulation_safe} is essential for the rooted weak bisimulation format.

\begin{example}
Let $f$ be a unary function symbol with an $\aL$-liquid argument,
and $g$ a unary function symbol with an $\aleph$-frozen, $\Lambda$-liquid argument. They are defined by the rules
$$\frac{x\trans{a}y}{f(x)\trans{a}g(y)}\qquad\qquad\frac{x\trans{\tau}y}{f(x)\trans{\tau}f(y)}\qquad\qquad\frac{~}{g(x)\trans{a}x}$$
This TSS violates condition $\ref{rhs}'$ of Def.~\ref{def:rooted_eta_bisimulation_safe} (due to the first rule).
On the other hand, it is in rooted delay bisimulation format w.r.t.\ $\aleph$ and $\Lambda$, and thereby \dr.

Consider the following LTS, which was already depicted in Sect.~\ref{sec:equivalences_terms}:
$p_0\trans{a}q$, $p_0\trans{a}0$, $p_1\trans{a}q$, $q\trans{\tau}0$ and $q\trans{b}0$.
We have $p_0\bis{rw}p_1$ (but $p_0\notbis{rd}p_1$). The LTS rooted in $f(p_0)$ and $f(p_1)$ is as follows.

\vspace{4mm}

\centerline{\begin{picture}(0,0)%
\includegraphics{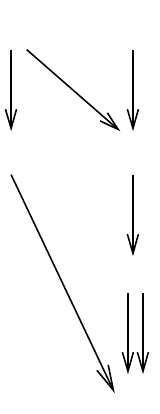}%
\end{picture}%
\setlength{\unitlength}{3947sp}%
\begingroup\makeatletter\ifx\SetFigFont\undefined%
\gdef\SetFigFont#1#2#3#4#5{%
  \reset@font\fontsize{#1}{#2pt}%
  \fontfamily{#3}\fontseries{#4}\fontshape{#5}%
  \selectfont}%
\fi\endgroup%
\begin{picture}(756,1967)(3697,-3081)
\put(3796,-1886){\makebox(0,0)[b]{\smash{{\SetFigFont{10}{12.0}{\rmdefault}{\mddefault}{\updefault}{\color[rgb]{0,0,0}$g(0)$}%
}}}}
\put(4351,-2461){\makebox(0,0)[b]{\smash{{\SetFigFont{10}{12.0}{\rmdefault}{\mddefault}{\updefault}{\color[rgb]{0,0,0}$q$}%
}}}}
\put(4326,-1886){\makebox(0,0)[b]{\smash{{\SetFigFont{10}{12.0}{\rmdefault}{\mddefault}{\updefault}{\color[rgb]{0,0,0}$g(q)$}%
}}}}
\put(4414,-2086){\makebox(0,0)[lb]{\smash{{\SetFigFont{10}{12.0}{\rmdefault}{\mddefault}{\updefault}{\color[rgb]{0,0,0}$a$}%
}}}}
\put(3796,-1261){\makebox(0,0)[b]{\smash{{\SetFigFont{10}{12.0}{\rmdefault}{\mddefault}{\updefault}{\color[rgb]{0,0,0}$f(p_0)$}%
}}}}
\put(4326,-1261){\makebox(0,0)[b]{\smash{{\SetFigFont{10}{12.0}{\rmdefault}{\mddefault}{\updefault}{\color[rgb]{0,0,0}$f(p_1)$}%
}}}}
\put(4051,-1486){\makebox(0,0)[lb]{\smash{{\SetFigFont{10}{12.0}{\rmdefault}{\mddefault}{\updefault}{\color[rgb]{0,0,0}$a$}%
}}}}
\put(4351,-3026){\makebox(0,0)[b]{\smash{{\SetFigFont{10}{12.0}{\rmdefault}{\mddefault}{\updefault}{\color[rgb]{0,0,0}$0$}%
}}}}
\put(4414,-1486){\makebox(0,0)[lb]{\smash{{\SetFigFont{10}{12.0}{\rmdefault}{\mddefault}{\updefault}{\color[rgb]{0,0,0}$a$}%
}}}}
\put(3712,-1486){\makebox(0,0)[rb]{\smash{{\SetFigFont{10}{12.0}{\rmdefault}{\mddefault}{\updefault}{\color[rgb]{0,0,0}$a$}%
}}}}
\put(4276,-2710){\makebox(0,0)[rb]{\smash{{\SetFigFont{10}{12.0}{\rmdefault}{\mddefault}{\updefault}{\color[rgb]{0,0,0}$\tau$}%
}}}}
\put(4438,-2710){\makebox(0,0)[lb]{\smash{{\SetFigFont{10}{12.0}{\rmdefault}{\mddefault}{\updefault}{\color[rgb]{0,0,0}$b$}%
}}}}
\put(3925,-2461){\makebox(0,0)[rb]{\smash{{\SetFigFont{10}{12.0}{\rmdefault}{\mddefault}{\updefault}{\color[rgb]{0,0,0}$a$}%
}}}}
\end{picture}%
}

\vspace{4mm}

\noindent
$f(p_0)\notbis{rw}f(p_1)$ because the transition $f(p_0)\trans{a}g(0)$ cannot be mimicked by $f(p_1)$.
Hence rooted weak bisimilarity is not a congruence.
\end{example}

\section{Checking for delay resistance} \label{sec:checking-delay-resistance}

On top of the congruence formats for (rooted) delay and weak bisimilarity we have
imposed delay resistance, a non-syntactic restriction that is based on concepts from \cite{Blo95,vGl11}.
To show that a TSS $P$ is delay resistant one has to establish a property for each
$P$-ruloid, and a non-trivial TSS has infinitely many of them.
This section introduces tools that lighten the burden of checking that a TSS is delay resistant.
In particular, syntactic criteria are proposed which imply that a TSS is \dr.

\subsection[Delay resistance w.r.t.\ Lambda]{Delay resistance w.r.t.\ $\Lambda$}\label{sec:Lambda}

We introduce the notion ``{\dr} w.r.t.\ $\Lambda$'', in which the conditions of
Defs.~\ref{def:delayable},~\ref{def:positive} and~\ref{def:negative} need to be checked only for premises containing
a variable that occurs only $\Lambda$-frozen in the source. We show that delay resistance w.r.t.~$\Lambda$,
together with the $\aL$-patience rules, conditions \ref{aleph} and \ref{main} of
Def.~\ref{def:rooted_branching_bisimulation_safe}, and one additional condition
(see Def.~\ref{def:smooth}) imply delay resistance.
So in the context of our congruence formats, when incorporating the additional condition,
delay resistance w.r.t.\ $\Lambda$ is sufficient to check that a TSS is \dr.
Moreover, it follows that the delay and weak bisimulation formats do not require delay resistance at all,
since with $\Lambda$ universal there are no $\Lambda$-frozen occurrences.

\begin{definition}\rm\label{def:delay-resistance-Lambda}
Given a predicate $\Lambda$ on the arguments of function symbols, %\vspace{1pt}
a premise \nietplat{$w \trans\beta y$} or \nietplat{$w\ntrans\beta$}
of a transition rule \nietplat{$\frac{H}{t \trans\alpha u}$}
is called $\Lambda$-liquid if all variables in $w$ occur
$\Lambda$-liquid in $t$. Let $H^\Lambda$ be the set of $\Lambda$-liquid premises in $H$.

An ntytt rule \nietplat{$\frac{ H}{t\trans{\alpha}u}$} %\vspace{3pt}
is \emph{{\pdr} w.r.t.\ $\Lambda$} and a TSS $P$ if there exists a finite set $H^d \subseteq H^+$ of delayable positive premises such that
  for each set $M \subseteq H^+{\setminus}(H^d \cup H^\Lambda)$ there is a rule
  \nietplat{$r_M=\frac{H_M}{t\trans{\alpha}u}$}, linearly provable from $P$, where
  $H_M \subseteq (H{\setminus}M) \cup M_\tau$.

A rule \nietplat{$\frac{H}{t\trans{\alpha}u}$} %\vspace{1pt}
is \emph{{\ndr} w.r.t.}\ $\Lambda$ and $P$ if
there is a rule \nietplat{$\frac{H'}{t\trans{\alpha}u}$}, linearly provable from $P$,
with $H' \subseteq H^+ \cup H^{s-} \cup H^\Lambda$.

An ntytt rule \nietplat{$\frac{ H}{t\trans{\alpha}u}$} %\vspace{2pt}
is \emph{delay resistant w.r.t.\ $\Lambda$} and $P$
if it is {\pdr} as well as {\ndr} w.r.t.\ $\Lambda$ and $P$.
A standard TSS $P$ in ready simulation format is \emph{delay resistant w.r.t.\ $\Lambda$} if all its linear ruloids
with a positive conclusion are delay resistant w.r.t.\ $\Lambda$ and $\hat P^+$.
\end{definition}

\noindent
By taking $\Lambda:=\emptyset$ we retrieve the notion of delay resistance from Def.~\ref{def:delay-resistant}.

The syntactic condition in the following definition prevents that a running process $x$ is tested twice.

\begin{definition}\label{def:smooth}
{\rm
Given a standard ntytt rule $r=\frac{H}{t\trans\alpha u}$, we define the following syntactic condition:
\begin{enumerate}
\setcounter{enumi}{\thesaveenumi}
\setcounter{saveenumi}{\theenumi}
\item \label{smooth}
If $x$ has exactly one occurrence in $t$, which is $\Lambda$-liquid, and an $\aleph$-liquid occurrence in $H$,
then these are the only two occurrences of $x$ in $r$.
\end{enumerate}
}
\end{definition}
For non-standard rules we take this condition to be vacuously satisfied.

In line with Prop.~\ref{prop:preservation_bra_bisimulation_safe}, it can be proved that condition \ref{smooth}
is preserved by the construction of ruloids.

\begin{lemma}
\label{lem:preservation_smooth}
Let $P$ be a TSS in decent ntyft format, in which each transition rule is
rooted branching bisimulation safe and satisfies condition \ref{smooth} of Def.~\ref{def:smooth} w.r.t.\ $\aleph$ and $\Lambda$.
Then any ntytt rule irredundantly provable from $P$
satisfies condition \ref{smooth} of Def.~\ref{def:smooth} w.r.t.\ $\aleph$ and $\Lambda$.
\end{lemma}

\begin{proof}
Let an ntytt rule $\frac{H}{t\trans{\alpha} u}$
be irredundantly provable from $P$, by means of a proof $\pi$. We prove,
using structural induction with respect to $\pi$, that this rule satisfies condition \ref{smooth}
of Def.~\ref{def:smooth} w.r.t.\ $\aleph$ and $\Lambda$.

\vspace{2mm}

\noindent
{\em Induction basis}:
Suppose $\pi$ has only one node, marked ``hypothesis''. Then $\frac{H}{t\trans{\alpha} u}$ equals
$\frac{t\trans{\alpha}u}{t\trans{\alpha}u}$ (so $t$ and $u$ are distinct variables).
This rule satisfies condition \ref{smooth} of Def.~\ref{def:smooth} w.r.t.\ $\aleph$ and $\Lambda$.

\vspace{2mm}

\noindent
{\em Induction step}:
Let $r\in R$ be the rule and $\sigma$ the substitution used at the bottom of $\pi$. By assumption,
$r$ is decent, ntyft, and rooted branching bisimulation safe w.r.t.\ $\aleph$ and $\Lambda$.
Let
\[
\{v_k\trans{\beta_k}y_k\mid k\in K\}\cup\{w_\ell\ntrans{\gamma_\ell}\mid \ell\in L\}
\]
be the set of premises of $r$, and
\[
f(x_1,\ldots,x_{\ar(f)})\trans{\alpha}v
\]
the conclusion of $r$. Then $\sigma(f(x_1,\ldots,x_{\ar(f)}))=t$ and $\sigma(v)=u$. Moreover, rules
$r_k = \frac{H_k}{\sigma(v_k)\trans{\beta_k}\sigma(y_k)}$ for each $k\in K$ and
$r_\ell = \frac{H_\ell}{\sigma(w_\ell)\ntrans{\gamma_\ell}}$ for each $\ell\in L$ are irredundantly provable
from $P$ by means of strict subproofs of $\pi$, where $H=\bigcup_{k\in K}H_k\cup\bigcup_{\ell\in L}H_\ell$.

As $r$ is decent, $\var(v_k) \subseteq \{x_1,\ldots,x_{\ar(f)}\}$, so $\var(\sigma(v_k)) \subseteq
\var(t)$ for each $k\in K$. Likewise,\linebreak $\var(\sigma(w_\ell))\subseteq\var(t)$ for each $\ell\in
L$. From ${\it rhs}(H) \cap \var(t) = \emptyset$ it follows that ${\it rhs}(H_k) \cap
\var(\sigma(v_k)) = \emptyset$ for each $k\in K$, and ${\it
  rhs}(H_\ell)\cap\var(\sigma(w_\ell))=\emptyset$ for each $\ell\in L$. So for each $k\in K$ and
$\ell\in L$, the rules $r_k$ and $r_\ell$ are ntytt rules. By Lem.~\ref{lem:preservation_decency},
they are decent. By Lem.~\ref{lem:preservation_branching_bisimulation_safe} they are rooted branching bisimulation safe w.r.t.\ $\aleph$ and $\Lambda$.
And by induction, they satisfy condition \ref{smooth} of Def.~\ref{def:smooth} w.r.t.\ $\aleph$ and $\Lambda$.

Suppose that $x$ has exactly one occurrence in $t$, which is $\Lambda$-liquid, and an
$\aleph$-liquid occurrence in $H$. Then there is an $i_0\in\{1,\ldots,\ar(f)\}$ with
$\Lambda(f,i_0)$ such that $x$ has exactly one occurrence in
$\sigma(x_{i_0})$, which is $\Lambda$-liquid; moreover,
$x\notin\var(\sigma(x_i))$ for each $i\neq i_0$. And
$x$ occurs $\aleph$-liquid in the left-hand side of a premise in $H_{m}$ for some $m\in K\cup L$.
Since $r_{m}$ is rooted branching bisimulation safe w.r.t.\ $\aleph$ and $\Lambda$,
by condition \ref{aleph} of Def.~\ref{def:rooted_branching_bisimulation_safe} or Def.~\ref{def:nonstandard-branching},
$x$ must occur $\aleph$-liquid in $\sigma(v_{m})$ or $\sigma(w_{m})$.
By the decency of $r$, this implies that $x_{i_0}$ occurs $\aleph$-liquid in $v_{m}$ or $w_{m}$.
Since $r$ is rooted branching bisimulation safe w.r.t.\ $\aleph$ and $\Lambda$, and $\Lambda(f,i_0)$,
by condition \ref{smooth} of Def.~\ref{def:smooth}, this
is the only occurrence of $x_{i_0}$ in the $v_k$ or $w_\ell$ for each
$k\in K$ and $\ell\in L$, and $x_{i_0}\notin\var(v)$.
Moreover, either by condition \ref{aleph} (if $\neg\aleph(f,i_0)$) or condition \ref{main} (if
$\aleph(f,i_0)$, using that $\Lambda(f,i_0)$) of Def.~\ref{def:rooted_branching_bisimulation_safe}, $m \in K$.
And by condition \ref{Lambda} of Def.~\ref{def:rooted_branching_bisimulation_safe}, the occurrence
of $x_{i_0}$ in $v_{m}$ is $\Lambda$-liquid.
It follows that $x$ has exactly one occurrence in $\sigma(v_{m})$,
which is $\Lambda$-liquid, and that $x$ does not occur in $\sigma(v_k)$ for
each $k\in K\backslash\{m\}$ and $\sigma(w_\ell)$ for each $\ell\in L$.
So $x$ does not occur in $r_k$ for each $k\in K\backslash\{m\}$ and
$r_\ell$ for each $\ell\in L$, in view of the decency of these rules,
and the fact that $x$ cannot occur in ${\it rhs}(H)$.
That is, $x$ does not occur in $H_k$ and $\sigma(y_k)$ for each $k\in K\backslash\{m\}$, and $H_\ell$ for
each $\ell\in L$. And since $r_{m}$ satisfies
condition \ref{smooth} of Def.~\ref{def:smooth}, while $x$ occurs $\aleph$-liquid in $H_m$ and
has exactly one occurrence in $\sigma(v_m)$, which is $\Lambda$-liquid,
$x$ has only one occurrence in $H_{m}$, and $x\notin\var(\sigma(y_{m}))$. 
Concluding, $x$ has only one occurrence in $H$; and since
$x_{i_0}\notin\var(v)$, $x\notin\var(\sigma(x_i))$ for each $i\neq
i_0$, and $x\notin\var(\sigma(y_k))$ for each $k\in K$, it follows that $x\notin\var(\sigma(v))$.
\qed
\end{proof}

\begin{corollary}
\label{cor:smooth}
Let $P$ be a standard TSS in ready simulation format, in which each rule is rooted branching bisimulation
safe and satisfies condition \ref{smooth} of Def.~\ref{def:smooth} w.r.t.\ $\aleph$ and $\Lambda$.
Then any $P$-ruloid satisfies condition \ref{smooth} of Def.\,\ref{def:smooth} w.r.t.\ $\aleph$ and $\Lambda$.
\end{corollary}

\begin{proof}
We recall from Sect.~\ref{sec:ruloids}, that the standard TSS $P$ can be transformed into a
TSS $P^+$ in decent ntyft format;
the $P$-ruloids are those decent nxytt rules that are irredundantly provable from $P^+$.

As the rules of $P$ are rooted branching bisimulation safe and satisfying condition \ref{smooth} of
Def.~\ref{def:smooth} w.r.t.\ $\aleph$ and $\Lambda$, then so are the rules in $P^+$.
Namely, as described in Sect.~\ref{sec:ruloids}, the rules in $P^+$ are
constructed in three steps. The first step (the conversion of $P$ to
decent ntyft format) clearly preserves the
rooted branching bisimulation format, as well as condition \ref{smooth}.
The second step (the construction to reduce
left-hand sides of positive premises to variables) yields an
intermediate TSS, all of whose rules are irredundantly provable from $P$,
and thus is covered by Lem.~\ref{lem:preservation_branching_bisimulation_safe} for rooted branching bisimulation safety, and by 
Lem.~\ref{lem:preservation_smooth} for condition~\ref{smooth}.
The the final step adds rules with negative conclusions to the TSS;
as pointed out in the proof of Prop.~\ref{prop:preservation_bra_bisimulation_safe}
(in~\cite{FvGdW12}), these added rules are also rooted branching 
bisimulation safe. Trivially, they satisfy condition \ref{smooth}.

Since the rules in $P^+$ are rooted branching bisimulation safe and satisfy condition \ref{smooth}
w.r.t.\ $\hspace{-1pt}\aleph\hspace{-1pt}$ and $\hspace{-1pt}\Lambda$,
by Lem.~\ref{lem:preservation_smooth}, each $P$-ruloid satisfies condition
\ref{smooth} w.r.t.\ $\hspace{-.5pt}\aleph\hspace{-.5pt}$ and $\Lambda$.
\qed
\end{proof}

\noindent
The next lemma is the crucial step in showing that delay resistance w.r.t.\ $\Lambda$,
together with the $\aL$-patience rules, conditions \ref{aleph} and \ref{main} of
Def.~\ref{def:rooted_branching_bisimulation_safe}, and condition~\ref{smooth} of Def.~\ref{def:smooth},
implies delay resistance.

\begin{lemma}\label{lem:delay-resistant}
Let $P$ be an $\aL$-patient TSS,
and let $r=\frac{H}{t \trans\alpha u}$ be an xyntt rule with $t$ univariate, linearly provable from $P$,
that is rooted branching bisimulation safe and satisfies condition \ref{smooth} of
Def.~\ref{def:smooth} w.r.t.\ $\aleph$ and $\Lambda$.
If $r$ is {\dr} w.r.t.\ $\Lambda$ and $P$, then it is {\dr} w.r.t.\ $P$.
\end{lemma}

\begin{proof}
To show that $r$ is {\pdr} w.r.t.\ $P$ it suffices to show that $H^\Lambda \cap H^+$ is finite, and that all premises in $H^\Lambda\cap H^+$ are delayable.
To show that $r$ is {\ndr} w.r.t.\ $P$ it suffices to show that $H^\Lambda$ does not contain negative premises at all.

By definition, for each premise \nietplat{$x \trans\beta y$} or \nietplat{$x\ntrans\beta$} in $H^\Lambda$,
$x$ occurs $\Lambda$-liquid in $t$. By condition \ref{aleph} of Def.~\ref{def:rooted_branching_bisimulation_safe},
the (unique) occurrence of $x$ in $t$ is also $\aleph$-liquid.
By condition \ref{main} of Def.~\ref{def:rooted_branching_bisimulation_safe},
$H^\Lambda$ contains only one premise with left-hand side $x$, which must be positive.
So $H^\Lambda$ contains no negative premises.
And since $\var(t)$ is finite, it follows that $H^\Lambda\cap H^+$ is finite too.

Let \nietplat{$H=H_0 \uplus \{x \trans\beta y\}$} where $x$ occurs $\aL$-liquid in $t$.
We need to show that there exist xyntt rules $\frac{ H_1}{t \trans{\tau} v}$ %\vspace{-1pt}
and $\frac{ H_2}{v \trans{\alpha} u}$, linearly provable from $P$, with $H_1\subseteq H_0 \cup\{x \trans{\tau} z\}$
  and $H_2\subseteq H_0 \cup\{z \trans{\beta} y\}$ for some term $v$ and fresh variable $z$.
By condition \ref{smooth} of Def.~\ref{def:smooth},
$x$ does not occur in $u$ or $H_0$. Let $v$ be obtained from $t$ by
substituting a fresh variable $z$ for $x$.
Then $\frac{H_2}{v\trans{\alpha}u}$ is a substitution instance of $\frac{H}{t\trans{\alpha}u}$,
and hence linearly provable from $P$, using Lem.~\ref{linear proof composition}.
As \nietplat{$\frac{x \trans{\tau} z}{t \trans{\tau} v}$} is $\aL$-patient, it is also linearly provable from $P$.
\qed
\end{proof}

\noindent
The following lemmas ensure that to verify delay resistance of a TSS $P$, it suffices to check
delay resistance w.r.t.\ $P$ for linear $P$-ruloids with a univariate source.

\begin{lemma}\label{lem:univariate}
Let $P$ be a TSS in ntyft format. Any ntytt rule linearly provable from $P$ is the instance under
a substitution $\sigma:V\rightarrow V$ of a rule, linearly provable from $P$, with a univariate source $t$.
Moreover, $\textrm{dom}(\sigma) = \var(t)$.
\end{lemma}

\begin{proof}
Straightforward by induction on the linear proof of the rule.
\qed
\end{proof}

\begin{lemma}\label{lem:substitution}
Let \nietplat{$r\mathbin=\frac{H}{t\trans\alpha u}$} be an ntytt rule, and $\sigma:\var(t)\rightarrow V$.\\[1pt]
If $r$ is  {\dr} w.r.t.\ a TSS $P$, then so is the rule
\nietplat{$\sigma(r)\mathbin=\frac{\sigma(H)}{\sigma(t)\trans\alpha \sigma(u)}$}.
\end{lemma}

\begin{proof}
Let $r$ be {\dr} w.r.t.\ $P$. %\vspace{1pt}
Since $r$ is {\ndr} w.r.t.\ $P$, there exists a rule \nietplat{$\frac{H'}{t\trans\alpha u}$}, %\vspace{6pt}
linearly provable from $P$, with
$H' \subseteq H^+\cup H^{s-}$. By Lem.~\ref{linear proof composition}, the rule
\nietplat{$\frac{\sigma(H')}{\sigma(t)\trans\alpha \sigma(u)}$} %\vspace{4pt}
is also linearly provable from $P$,
and $\sigma(H') \subseteq \sigma(H)^+\cup \sigma(H)^{s-}$. It follows that $\sigma(r)$ is {\ndr} w.r.t.\ $P$.

As $\sigma$ does not affect ${\it rhs}(H)$, $\sigma$ is injective on $H^+$.
Let $H^d \subseteq H^+$ be the finite set of delayable premises that exists by Def.~\ref{def:positive}.
Take $\sigma(H)^d := \sigma(H^d)$. Then any subset of $\sigma(H)^+\setminus 
\sigma(H)^d$ can be written as $\sigma(M)$ with $M\subseteq H\setminus H^d$.
By Def.~\ref{def:positive} there exists a rule $r_M=\frac{H_M}{t\trans\alpha u}$,
linearly provable from $P$,
where $H_M \subseteq (H{\setminus}M)\cup M_\tau$. Hence the rule
\nietplat{$\sigma(r_M)=\frac{\sigma(H_M)}{\sigma(t)\trans\alpha \sigma(u)}$}, %\vspace{2pt}
where $\sigma(H_M) \subseteq (\sigma(H){\setminus}\sigma(M))\cup \sigma(M)_\tau$, is
linearly provable from $P$ by Lem.~\ref{linear proof composition}.

It remains to show that the premises in $\sigma(H^d)$ are delayable.
This follows immediately from the delayability of the premises in $H^d$, by applying $\sigma$.
\qed
\end{proof}

\noindent
The next proposition states that delay resistance w.r.t.\ $\Lambda$, together with
the presence of the $\aL$-patience rules, rooted branching bisimulation safeness and
condition \ref{smooth} of Def.~\ref{def:smooth}, is sufficient to guarantee delay resistance.

\begin{proposition}\label{prop:delay-resistant}
Let $P$ be a standard TSS in ready simulation format, in which each transition rule is rooted branching bisimulation
safe and satisfies condition \ref{smooth} of Def.~\ref{def:smooth} w.r.t.\ $\aleph$ and $\Lambda$.
Let moreover $P$ be $\aL$-patient and {\dr} w.r.t.\ $\Lambda$. Then $P$ is \dr.
\end{proposition}

\begin{proof}
Let \nietplat{$r=\frac{H}{t \trans\alpha u}$} %\vspace{4pt}
be a linear $P$-ruloid, i.e.\ an nxytt rule, linearly
provable from the TSS $\hat P^+$ constructed in Sec.~\ref{sec:ruloids}. We need to show that $r$ is \dr.
Since $P$ is $\aL$-patient, so is $\hat P^+$.
Using Lemmas~\ref{lem:univariate} and ~\ref{lem:substitution} (and that $\hat P^+$ is in ntyft format)
we may restrict attention to the case that $t$ is univariate.
By Prop.\,\ref{prop:preservation_bra_bisimulation_safe}, $r$ is rooted branching bisimulation safe
w.r.t.\ $\aleph$ and $\Lambda$, and by Cor.~\ref{cor:smooth} it moreover satisfies condition
\ref{smooth} of Def.~\ref{def:smooth} w.r.t.\ $\aleph$ and $\Lambda$. 
By assumption, $r$ is {\dr} w.r.t.\ $\Lambda$ and $P$.
The result now follows from Lem.~\ref{lem:delay-resistant}.
\qed
\end{proof}
If $\Lambda$ is universal, all premises are $\Lambda$-liquid.
Hence in the presence of condition \ref{smooth} of Def.~\ref{def:smooth}, the requirement of delay
resistance can be dropped from the delay and weak bisimulation formats.

\begin{definition}\rm\label{def:delay_bisimulation_format}
A standard TSS $P$ is in {\em syntactic delay bisimulation format} if it is in ready simulation format,
and, for some $\aleph$ and the universal predicate $\Lambda$, it is $\aL$-patient and all its transition rules are
rooted branching bisimulation safe and satisfy condition \ref{smooth} of Def.~\ref{def:smooth}
w.r.t.\ $\aleph$ and $\Lambda$.

The {\em syntactic weak bisimulation format} is defined likewise, but using condition $\ref*{rhs}'$
of Def.~\ref{def:rooted_eta_bisimulation_safe} instead of condition $\ref{rhs}$
of Def.~\ref{def:rooted_branching_bisimulation_safe}.
\end{definition}

\begin{corollary}
Let $P$ be a complete standard TSS in syntactic delay bisimulation format.
Then $\bis{d}$ (as well as $\bis{rd}$) is a congruence for $P$.

Let $P$ be a complete standard TSS in syntactic weak bisimulation format.
Then $\bis{w}$ (as well as $\bis{rw}$) is a congruence for $P$.
\end{corollary}

\noindent
The following example from \cite{vGl11} shows that in the above corollary (as well as in
Prop. \ref{prop:delay-resistant}), condition \ref{smooth} of Def.\ \ref{def:smooth} cannot be omitted.

\begin{example}
The operator $s$ of \cite[Sec.\ 10(4)]{GlWe96} allows a process (its argument) to
proceed normally, but in addition can report that the process is ready
to perform a visible action, without actually doing it.  It supposes
an alphabet $A := {\cal L} \uplus \{\mbox{\footnotesize Can do `$a$'}
\mid a \in {\cal L}\}$ for some nonempty set ${\cal L}$, and its rules are
$$\frac{x \trans{\alpha} y}{s(x) \trans{\alpha} s(y)}~(\alpha\in A\cup\{\tau\})
~~~~~~~~~~
  \frac{x \trans{a} y}{s(x) \stackrel{\raisebox{-2pt}[0pt][0pt]
  {\tiny Can do	`$a$'}}
  {-\!\!\!-\!\!\!-\!\!\!\longrightarrow} s(x)}~(a\in {\cal L})$$
Consider the TSS that consists of these two rules together with the transitions from the first LTS in Ex.~\ref{ex:weak bisimulations},
defining the rooted delay bisimilar processes $p_0$ and $p_1$.
This TSS is in syntactic delay bisimulation format, except that it violates condition \ref{smooth} of Def.\ \ref{def:smooth}: by the first rule for $s$
the argument of $s$ is $\Lambda$-liquid, so in the second rule $x$ occurs $\Lambda$-liquid in the
source, $\aleph$-liquid in the premise, and also in the target.

On this TSS $\bis{d}$ and $\bis{w}$ are not congruences, for $p_0 \bis{d} p_1$ whereas
$s(p_0) \notbis{w} s(p_1)$. Namely, only $s(p_0)$ can report ``can do
{`$b$'}'' and then do $a$ (as depicted in \cite[Fig.~7]{GlWe96}). Thus
condition \ref{smooth} cannot be skipped from the syntactic delay and weak bisimulation formats.
\end{example}

\subsection{Semi-syntactic criteria for delay resistance}

We now introduce requirements on the rules of a TSS $P$
that imply delay resistance of $P$. This yields what could be called
\emph{semi-syntactic} congruence formats. They are not purely syntactic, because
one of the conditions (in Def.~\ref{def:manifestly delayable}) requires the existence of certain
linearly provable rules;
however, all conditions need to be checked for rules in $P$ only (rather than for $P$-ruloids).

\begin{definition}\rm\label{def:manifestly delayable}
A premise \nietplat{$w \trans\beta y$} of an ntytt rule \nietplat{$r=\frac{H}{t\trans\alpha u}$} %\vspace{1pt}
is \emph{manifestly delayable} in a TSS $P=(\Sigma,R)$ if, for some term $v$ and fresh variable $z$,
there is a transition rule $\frac{H_1}{t \trans{\tau}v}$ in $\overline R$,
as well as an ntytt rule $\frac{ H_2}{v \trans{\alpha} u}$ linearly
provable from $P$, with \nietplat{$H_1\subseteq (H{\setminus}\{w \trans{\beta} y\})\cup\{w \trans{\tau} z\}$}
and \nietplat{$H_2\subseteq (H{\setminus}\{w \trans{\beta} y\})\cup\{z \trans{\beta} y\}$}.
\end{definition}
Here $\overline R$ denotes the set $R$ of transition rules up to a bijective renaming of variables.
The difference with Def.~\ref{def:delayable} is that here the rule $\frac{H_1}{t \trans{\tau} v}$ needs to be
in $\overline R$, rather than merely being linearly provable from $P$.
So clearly each manifestly delayable premise is delayable.

\begin{definition}\rm
A transition rule $\frac{H}{t\trans\alpha u}$ is \emph{manifestly negative delay resistant}, or more briefly
\emph{negative-stable}, if for every premise $w\ntrans\alpha$ in $H\!$,
also \nietplat{$w\ntrans\tau$} is in $H\!$. A TSS is \emph{negative-stable} if all its rules are.
\end{definition}
The difference with Def.~\ref{def:negative} is that here the requirement also applies to redundant
premises \nietplat{$w\ntrans\alpha$}. Clearly each negative-stable rule is {\ndr} w.r.t.\ any TSS.

\begin{definition}\rm\label{def:manifestly delay resistant}
An ntytt rule \nietplat{$\frac{ H}{t\trans{\alpha}u}$} %\vspace{2pt}
is \emph{manifestly delay resistant}
w.r.t.\ \Brac{a predicate $\Lambda$ and} a TSS $P=(\Sigma,R)$ if it is negative-stable
and there exists a finite set $H^d \subseteq H^+$ of manifestly delayable
positive premises, such that for each set $M \subseteq H^+{\setminus}(H^d \Brac{\cup H^\Lambda})$ there is a rule
\nietplat{$r_M=\frac{H_M}{t\trans{\alpha}u}$} in $\overline R$, where $H_M \subseteq (H{\setminus}M) \cup M_\tau$.
\end{definition}
Again, the rule $r_M$ needs to be in $\overline R$, rather than merely being linearly provable from $P$.
The material in this section and in the appendix comes in two flavours: incorporating a predicate
$\Lambda$ on arguments of function symbols, or omitting it.  The latter is equivalent to taking
$\Lambda=\emptyset$. Notationally, we will capture both by putting the optional material, pertaining
to $\Lambda$, between square brackets.

Clearly, a manifestly delay resistant rule w.r.t.\ \Brac{$\Lambda$ and} $P$ is delay resistant
w.r.t.\ \Brac{$\Lambda$ and} $P$ (cf.\ Def.~\ref{def:delay-resistant} \Brac{or Def.~\ref{def:delay-resistance-Lambda}}).

\begin{definition}\rm\label{def:manifestly}
A standard TSS $P$ in decent ntyft format is \emph{manifestly delay resistant} \Brac{w.r.t.\ $\Lambda$}
if all its transition rules are manifestly delay resistant w.r.t.\ \Brac{$\Lambda$ and} $P$.
A standard TSS $P$ in ready simulation format is \emph{manifestly delay resistant} \Brac{w.r.t.\ $\Lambda$} if its conversion
$P^\dagger$ to decent ntyft format (see Sec.~\ref{sec:ruloids}) is manifestly delay resistant \Brac{w.r.t.\ $\Lambda$}.
\end{definition}
Note that in contrast to the notion of a delay resistant TSS from Def.~\ref{def:delay-resistant},
here the property is only required for the rules of $P$, instead of all linear $P$-ruloids.
The following theorem, whose proof is presented in Appendix~\ref{app:manifest}, provides a semi-syntactic version
of all our congruence formats.

\begin{theorem}\label{thm:manifest}
Any manifestly delay resistant standard TSS in ready simulation format is delay resistant.
\end{theorem}
The following variant of this theorem, whose proof is also presented in Appendix~\ref{app:manifest}, mixes in the
insights of Sec.~\ref{sec:Lambda}, and provides semi-syntactic versions of our congruence formats
that are normally easier to apply.

\begin{theorem}\label{thm:manifest-Lambda}
Let $P$ be a standard TSS in ready simulation format, in which each transition rule is rooted branching bisimulation
safe and satisfies condition \ref{smooth} of Def.~\ref{def:smooth} w.r.t.\ $\aleph$ and $\Lambda$.
Let moreover $P$ be $\aL$-patient and manifestly {\dr} w.r.t.\ $\Lambda$. Then $P$ is \dr.
\end{theorem}

\begin{definition}\rm\label{def:manifest_rooted_delay_bisimulation_format}
A standard TSS $P$ is in {\em manifest rooted delay bisimulation format} if it
is in ready simulation format, and, for some $\aleph$ and $\Lambda$,
it is $\aL$-patient and manifestly delay resistant w.r.t.\ $\Lambda$, and it only contains transition rules that are
rooted branching bisimulation safe and satisfy condition \ref{smooth} of
Def.~\ref{def:smooth} w.r.t.\ $\aleph$ and $\Lambda$.

The {\em manifest rooted weak bisimulation format} is defined likewise, but using condition $\ref*{rhs}'$
of Def.~\ref{def:rooted_eta_bisimulation_safe} instead of condition $\ref{rhs}$
of Def.~\ref{def:rooted_branching_bisimulation_safe}.
\end{definition}

\begin{corollary}\label{cor:manifest}
Let $P$ be a complete standard TSS in manifest rooted delay bisimulation format.
Then \mbox{}$\bis{rd}$ is a congruence for $P$.

Let $P$ be a complete standard TSS in manifest rooted weak bisimulation format.
Then $\bis{rw}$ is a congruence for $P$.
\end{corollary}

\noindent
For most applications, to check delay resistance it suffices to check the following, simpler property.

\begin{definition}\rm
A standard ntytt rule is \emph{simply $\Lambda$-delay resistant} in a TSS $P$
if it is negative-stable and has finitely many positive premises, all of which are either manifestly
delayable in $P$ or $\Lambda$-liquid. 
A standard TSS in decent ntyft format is \emph{simply $\Lambda$-delay resistant} if all its transition rules are.
\end{definition}
Clearly, a simply $\Lambda$-delay resistant transition rule \plat{$r=\frac{H}{t\trans\alpha u}$} is manifestly delay
resistant w.r.t.\ $\Lambda$, by taking $H^d:= H^{+}\setminus H^\Lambda$ and $r_\emptyset:=r$.

\subsection{Syntactic criteria for delay resistance} \label{sec:getting-rid}

We show how delay resistance can be replaced by additional syntactic requirements.
Def.~\ref{def:rooted_delay_bisimulation_format} is adapted as follows. On the one hand the requirement that the TSS is {\dr}
is dropped. On the other hand, rules must satisfy condition \ref{smooth} of Def.~\ref{def:smooth},
and the TSS must be in nxytt format and negative-stable.
And there are additional syntactic restrictions if a $\Lambda$-frozen argument of the source is tested in the premises
(condition \ref{strengthened-3}).
Furthermore, there is a syntactic requirement with regard to predicates $\Delta_\alpha\subseteq\aL$ (condition \ref{strengthened-4}).

\begin{definition}\label{def:strengthened_delay_bisimulation_format}
{\rm
A standard TSS $P\mathbin=(\Sigma,R)$ is in {\em syntactic rooted delay bisimulation format} if,
for some $\aleph$ and $\Lambda$ and predicates $\Delta_\alpha\subseteq\aL$ where $\alpha$ ranges over $A\cup\{\tau\}$:
\begin{enumerate}
\item \label{strengthened-1}
$P$ is in decent nxytt format and $\aL$-patient.
\item \label{strengthened-2}
Each rule in $R$ is rooted branching bisimulation safe and
negative-stable, satisfies condition \ref{smooth} of Def.~\ref{def:smooth} w.r.t.\ $\aleph$ and $\Lambda$,
and has finitely many positive premises.
\item \label{strengthened-3}
If $R$ contains a rule $\frac{H\uplus\{x_i\trans{\beta}y\}}{f(x_1,\ldots,x_{\ar(f)})\trans{\alpha}u}$
where $\neg\Lambda(f,i)$, then:
\begin{enumerate}
\item \label{strengthened-3a}
$\beta=\alpha$;
\item \label{strengthened-3b}
$R$ contains a rule \nietplat{$\frac{H' \cup \{x_i\trans{\tau}y\}}
{f(x_1,\ldots,x_{\ar(f)})\trans{\tau}u}$} with $H'\subseteq H$; and\vspace{1mm}
\item \label{strengthened-3c}
$y$ has exactly one, $\Delta_\alpha$-liquid occurrence in $u$.
\end{enumerate}
\item \label{strengthened-4}
If $\Delta_\alpha(f,i)$, then $R$ contains
$\frac{x_i\trans{\alpha}y}{f(x_1,\ldots,x_i,\ldots,x_{\ar(f)})\trans{\alpha}f(x_1,\ldots,y,\ldots,x_{\ar(f)})}$.
\end{enumerate}
$P$ is in {\em syntactic rooted weak bisimulation format} if its rules
moreover satisfy condition \ref*{rhs}$'$ of Def.~\ref{def:smooth}.
}
\end{definition}

\noindent
The introduction of predicates $\Delta_\alpha\subseteq\aL$ is of practical importance.
If in Def.~\ref{def:strengthened_delay_bisimulation_format} one would replace the occurrences of $\Delta_\alpha$ by $\aL$,
then for instance the \emph{encapsulation} operator $\partial_H$, which blocks all actions in the set $H$, would violate
condition \ref{strengthened-4} of Def.~\ref{def:strengthened_delay_bisimulation_format}. Namely, the argument of $\partial_H$ is
$\aL$-liquid, but there is no rule $\frac{x\trans a y}{\partial_H(x)\trans a\partial_H(y)}$ if $a\in H$.

Actually, in many applications $\Delta_\alpha$ can be empty, as a rule that tests an $\aleph$-liquid, $\Lambda$-frozen
argument of the source in practice tends to have a single $y$ as right-hand side of the conclusion, so that
condition \ref{strengthened-3c} of Def.~\ref{def:strengthened_delay_bisimulation_format}
is trivially satisfied; a notable example is the rule $\frac{x_1\trans\alpha y}{x_1+x_2\trans\alpha y}$ for alternative composition.

\begin{proposition}
\label{prop:strengthened}
Let $P$ be a TSS in syntactic rooted delay bisimulation format.
Then it is in manifest rooted delay bisimulation format.
\end{proposition}

\begin{proof}
Let $\aleph$, $\Lambda$ and $\Delta_\gamma$ for all $\gamma\in A\cup\{\tau\}$ be such that $P$ satisfies the restrictions
in Def.~\ref{def:strengthened_delay_bisimulation_format}.
It suffices to show that $P$ is simply $\Lambda$-delay resistant.
Consider a rule $\frac{H\uplus\{x_i\trans\beta y\}}{f(x_1,\dots,x_n)\trans\alpha u}$ of $P$ with
$\neg \Lambda(f,i)$. By condition~\ref{strengthened-3a} of Def.~\ref{def:strengthened_delay_bisimulation_format},
$\beta=\alpha$. And by condition~\ref{strengthened-3c} of Def.~\ref{def:strengthened_delay_bisimulation_format},
$y$ has exactly one, $\Delta_\alpha$-liquid occurrence in $u$.
It suffices to show that $x_i\trans\beta y$ is manifestly delayable in $P$.
So, for some term $v$ and fresh variable $z$, there must be a rule
$\frac{H_1}{f(x_1,\dots,x_n) \trans{\tau} v}$ in $\overline R$,
as well as a rule $\frac{H_2}{v \trans{\alpha} u}$ linearly
provable from $P$, with \nietplat{$H_1\subseteq H\cup\{x_i \trans{\tau} z\}$}
and \nietplat{$H_2\subseteq H\cup\{z \trans{\beta} y\}$}.
Let $v$ be obtained by substituting $z$ for $y$ in $u$.
The first of these rules exists by condition~\ref{strengthened-3b} of
Def.~\ref{def:strengthened_delay_bisimulation_format}, substituting $z$ for $y$.
The second is the rule $\frac{z \trans{\beta}y}{v \trans{\alpha}u}$,
which can be derived by condition~\ref{strengthened-4} of Def.~\ref{def:strengthened_delay_bisimulation_format}.
Here we use that $\beta=\alpha$ and $y$ has exactly one, $\Delta_\alpha$-liquid occurrence in $u$.
\qed
\end{proof}

\noindent
Prop.~\ref{prop:strengthened}, together with Cor.~\ref{cor:manifest}, gives rise to the following corollary.

\begin{corollary}\label{cor:strenghtened}
Let the complete standard TSS $P$ be in syntactic rooted delay bisimulation format.
Then $\bis{rd}$ is a congruence for $P$.

Let the complete standard TSS $P$ be in syntactic rooted weak bisimulation format.
Then $\bis{rw}$ is a congruence for $P$.
\end{corollary}

\section{Applications}
\label{sec:applications}

In this section we revisit some applications of our congruence formats that were already considered in \cite{FvGdW12}:
the basic process algebra BPA$_{\varepsilon\delta\tau}$, extended with
binary Kleene star as an example where the predicates $\Delta_\alpha$ from Def.~\ref{def:strengthened_delay_bisimulation_format}
are non-empty, and initial priority because it includes negative premises. We also consider a
deadlock test that is outside the syntactic rooted delay bisimulation format.
In all these cases our formats provide congruence results for rooted delay and weak bisimilarity, while
they are outside the congruence formats for rooted delay and weak bisimilarity from \cite{Blo95,vGl11}.

The TSSs in this section are all $\aL$-patient and in decent xynft format.

\subsection{Basic process algebra}
\label{sec:bpa}

Consider the basic process algebra BPA$_{\varepsilon\delta\tau}$, consisting of: constants from an alphabet ${\it Act}\cup\{\tau\}$;
the empty process $\varepsilon$; the deadlock $\delta$;
alternative composition $t_1+t_2$; and sequential composition $t_1\cdot t_2$. Let $\ell$ range over ${\it Act}\cup\{\tau\}$
and $\alpha$ over ${\it Act}\cup\{\tau,\surd\}$. The transition rules are: \vspace{-1ex}
$$\frac{~}{\ell\trans\ell \varepsilon}\qquad\frac{~}{\varepsilon\trans\surd\delta}\qquad
\frac{x_1\trans\alpha y}{x_1+x_2\trans\alpha y}\qquad\frac{x_2\trans\alpha y}{x_1+x_2\trans\alpha y}$$
$$\frac{x_1\trans\ell y}{x_1\cdot x_2\trans\ell y\cdot x_2}\qquad
\frac{x_1\trans\surd y_1~~~x_2\trans\alpha y_2}{x_1\cdot x_2\trans\alpha y_2}$$
To show that $\bis{rd}$ and $\bis{rw}$ are congruences,
we argue that this TSS satisfies the conditions of Def.~\ref{def:strengthened_delay_bisimulation_format}.
In \cite{FvGdW12} it was shown that it is in rooted $\eta$-bisimulation format, with $\aleph$ and $\Lambda$ defined as follows.
Since the arguments of alternative and sequential composition can all execute immediately,
$\aleph$ holds for all these arguments. Since only the first argument of sequential composition can contain running processes,
it is the only argument for which $\Lambda$ holds.
Since the TSS is positive, it surely is negative-stable.
With regard to condition \ref{strengthened-2} of Def.~\ref{def:strengthened_delay_bisimulation_format},
we still need to check that the rules satisfy condition \ref{smooth} of Def.~\ref{def:smooth}:
only the two rules for sequential composition contain a $\Lambda$-liquid occurrence of a variable, $x_1$, in their source;
and in both cases $x_1$ has only one other occurrence in the rule, in the left-hand side of a premise.
Condition \ref{strengthened-3} of Def.~\ref{def:strengthened_delay_bisimulation_format} needs to be verified with regard to the two rules for
alternative composition and the second rule for sequential composition, since in these rules a $\Lambda$-frozen argument of
the source is tested in a premise. It is not hard to see that condition \ref{strengthened-3} is satisfied for these rules, where we
can take $\Delta_\gamma=\emptyset$ for all $\gamma$. Hence condition \ref{strengthened-4} of Def.~\ref{def:strengthened_delay_bisimulation_format}
is trivially satisfied.

Concluding, by Cor.~\ref{cor:strenghtened} rooted delay and weak bisimilarity are congruences for BPA$_{\varepsilon\delta\tau}$.

\subsection{Binary Kleene star}\label{sec:binary Kleene star}

The {\em binary Kleene star} $t_1{}^\ast t_2$ \cite{Kle56} repeatedly executes
$t_1$ until it executes $t_2$. This operational behaviour is captured by the
following rules, which are added to the rules for BPA$_{\varepsilon\delta\tau}$.\vspace{-1ex}
\[
{\displaystyle\frac{x_1\trans \ell y}{x_1{}^\ast x_2\trans \ell y\cdot(x_1{}^\ast x_2)}}
\hspace{2cm}{\displaystyle\frac{x_2\trans{\alpha} y}{x_1{}^\ast x_2\trans{\alpha} y}}\\
\]
Again, to show that $\bis{rd}$ and $\bis{rw}$ are congruences,
we argue that the resulting TSS satisfies the conditions of Def.~\ref{def:strengthened_delay_bisimulation_format}.
In \cite{FvGdW12} it was shown that it is in rooted $\eta$-bisimulation format,
if we take the arguments of the binary Kleene star to be $\Lambda$-frozen (they do not contain
running processes) and $\aleph$-liquid (they can start executing immediately).\pagebreak[2]
Since the arguments of the binary Kleene star are $\Lambda$-frozen, condition
condition \ref{smooth} of Def.~\ref{def:smooth} is trivially satisfied.
Condition \ref{strengthened-3} of Def.~\ref{def:strengthened_delay_bisimulation_format}
needs to be verified for the two rules for binary Kleene star.
It is easy to see that conditions \ref{strengthened-3}(a,b) are satisfied by both rules, and
that the second rule for binary Kleene star trivially satisfies condition \ref{strengthened-3}(c).
In view of the latter condition with regard to the first rule for binary Kleene star,
we mark the first argument of sequential composition by $\Delta_\ell$ for all $\ell\in {\it Act}\cup\{\tau\}$.
No other arguments are marked by the $\Delta_\gamma$.
It is easy to see that condition \ref{strengthened-4} of Def.~\ref{def:strengthened_delay_bisimulation_format}
is satisfied with respect to the $\Delta_\gamma$.
(Note that for this last condition it is essential that the first argument of sequential composition is not marked by $\Delta_\surd$.)

Concluding, by Cor.~\ref{cor:strenghtened}
rooted delay and weak bisimilarity are congruences for BPA$_{\varepsilon\delta\tau}$ with the
binary Kleene star.

\subsection{Initial priority}\label{sec:initial priority}

{\em Initial priority} is a unary function that assumes an ordering on atomic actions.
The term $\theta(t)$ executes the transitions of $t$,
with the restriction that an initial transition \nietplat{$t\trans {\ell} t_1$} only gives
rise to an initial transition \nietplat{$\theta(t)\trans \ell t_1$} %\vspace{1pt}
if there does not exist an initial transition \nietplat{$t\trans{\ell'} t_2$} with $\ell<\ell'$.
This intuition is captured by the first rule for the initial priority operator below,
which is added to the rules for BPA$_{\epsilon\delta\tau}$.
\[
\frac{x\trans \ell y~~~~~~~~x\ntrans{\ell'}\mbox{ for all } \ell'>\ell}{\theta(x)\trans \ell y}
\hspace{2cm}
\frac{x\trans \surd y}{\theta(x)\trans \surd y}
\]

\vspace{2mm}
\noindent
We take the argument of initial priority to be $\Lambda$-frozen (it does not contain
running processes) and $\aleph$-liquid (it can start executing immediately).
In \cite{FvGdW12} it was observed that the resulting TSS is in rooted $\eta$-bisimulation format,
irrespective of the ordering on atomic actions.

If we take $\tau$ to be greater than all atomic actions in ${\it Act}$, then
both rules are negative-stable, because instances of the first
rule for initial priority with a premise $x\ntrans{a}$ for some $a\in {\it Act}$ are guaranteed to also contain
the premise $x\ntrans{\tau}$. In fact it is sufficient to require $\forall\ell:(\exists\ell':\ell'>\ell)\Rightarrow \tau>\ell$.

To show that $\bis{rd}$ and $\bis{rw}$ are congruences,
we argue that the TSS satisfies the conditions of Def.~\ref{def:strengthened_delay_bisimulation_format}.
Condition \ref{smooth} of Def.~\ref{def:smooth} is trivially satisfied by the rules for initial priority, because its argument is $\Lambda$-frozen.
Condition \ref{strengthened-3} of Def.~\ref{def:strengthened_delay_bisimulation_format}
needs to verified with regard to the two rules for initial priority,
since in these rules the $\Lambda$-frozen argument of the source is tested in a premise.
It is not hard to see that condition \ref{strengthened-3} is satisfied for these rules, where we
can take $\Delta_\gamma=\emptyset$ for all $\gamma$. In particular, condition \ref{strengthened-3}(b)
is satisfied by the first rule for initial priority, because this rule with $\ell=\tau$ contains no negative
premises. Since the $\Delta_\gamma$ are empty, condition \ref{strengthened-4} of Def.~\ref{def:strengthened_delay_bisimulation_format}
is trivially satisfied.

Concluding, if $\tau$ is greater than all atomic actions in ${\it Act}$,
rooted delay and weak bisimilarity are congruences for BPA$_{\epsilon\delta\tau}$ with initial priority.

We note that if $\tau$ is smaller than some atomic action $a$ in ${\it Act}$, then
rooted delay and weak bisimilarity are not congruences for BPA$_{\epsilon\delta\tau}$ with initial priority.
For example, consider $\tau\cdot a$ and $(\tau\cdot a)+a$. These process terms are rooted delay bisimilar.
However, the transition $\theta(\tau\cdot a)\trans\tau a$ cannot be mimicked by $\theta((\tau\cdot a)+a)$,
as the latter term can only perform an $a$-transition to $\varepsilon$. So these terms are not rooted weakly bisimilar.

\subsection{Deadlock testing}

Finally we give an example that is outside the format from Def.~\ref{def:strengthened_delay_bisimulation_format}, but that
is covered by the more general format induced by Thm.~\ref{thm:manifest}. Let ${\it yes},{\it no}\in{\it Act}$.
The unary operator $f$, defined by the following two rules, tests whether its argument is a deadlock.\pagebreak[2]
\[
\frac{x\trans \alpha y}{f(x)\trans{{\it no}}\delta}
\hspace{2cm}
\frac{x\ntrans\alpha y \mbox{ for all }\alpha}{f(x)\trans{{\it yes}}\delta}
\]

\noindent
The argument of $f$ is $\Lambda$-frozen and $\aleph$-liquid. Clearly the first rule violates condition \ref{strengthened-3}
of Def.~\ref{def:strengthened_delay_bisimulation_format}.

The TSS is complete, in rooted $\eta$-bisimulation format and manifestly delay resistant.
In particular, for both rules, $H^d:=\emptyset$, while $r_\emptyset$ is the rule itself.
Furthermore, for the first rule, $r_{\{x\trans \alpha y\}}$ is $\frac{x\trans \tau y}{f(x)\trans{{\it no}}\delta}$.
So according to Thm.~\ref{thm:manifest} the TSS is delay resistant. Hence, by Thms.~\ref{thm:rooted-delay-congruence} and
\ref{thm:rooted-weak-congruence}, rooted delay and weak bisimilarity are congruences
for BPA$_{\epsilon\delta\tau}$ with deadlock testing.

\section{Conclusions} \label{sec:conclusions}

We have extended the method from \cite{FvGdW12} for modal decomposition and the derivation of congruence formats
so that it applies to delay and weak bisimilarity. This research line gives a deeper insight into the link between
modal logic and structural operational semantics, and provides a framework for the derivation of congruence formats
for the spectrum of weak semantics from \cite{vGl93}.

Admittedly, the whole story is quite technical and intricate. Partly this is because we build on a rich body of earlier work
in the realm of structural operational semantics: the notions of well-supported proofs and complete TSSs from \cite{vGl04}
(or actually \cite{GRS91} in logic programming); the ntyft format from \cite{BolG96,Gro93}; the transformation to ruloids,
which for the main part goes back to \cite{FvG96}; and the work on modal decomposition and congruence formats from \cite{BFvG04}
and \cite{FvGdW12}.

In spite of these technicalities, we have arrived at a relatively simple framework for the derivation of congruence
formats for weak semantics. Namely, for this one only needs to: (1) provide a modal characterisation of the weak semantics under
consideration; (2) study the class of modal formulas that result from decomposing this modal characterisation, and formulate
syntactic restrictions on TSSs to bring this class of modal formulas within the original modal characterisation; and (3) check
that these syntactic restrictions are preserved under the transformation to ruloids. As shown in Sect.~\ref{sec:weak-congruence},
steps (2) and (3) are very similar in structure for delay and weak bisimilarity. And as said, the end results are congruence formats
that are more general and at the same time more elegant than existing congruence formats for these semantics in the literature.

Our intention is to carve out congruence formats for all weak semantics in the spectrum from \cite{vGl93} that have reasonable congruence properties.
The work presented in this paper constitutes an essential step in this direction, as the majority (103 out of 155) of the weak semantics in this spectrum have
a modal characterisation that contains modalities $\eps\diam{a}\phi$. However, further work is needed to cover the entire spectrum in \cite{vGl93}.
In the follow-up paper \cite{FvGL17} another significant step in this direction is made by
extending the current framework to stability-respecting and divergence-preserving semantics.

In \cite{FvG17} the framework for concrete semantics was extended with lookahead; an open question is to do the same for weak semantics.
For future research it would also be interesting to see whether the bridge between modal logic and congruence formats
could be employed in the realm of logics and semantics for e.g.\ probabilities and security. As a first step in this direction,
in \cite{GF12,CGT16} the decomposition method for Hennessy-Milner logic was lifted to probabilistic systems.

%%%%
\newpage
\appendix

\section{Modal Characterisations}\label{app:modal}

We first prove the first part of Thm.~\ref{thm:characterisation}, which states that $\IO{w}$ is a modal
characterisation of weak bisimilarity. We need to prove, given an LTS $(\mathbb{P},\rightarrow)$, that $p\bis{w}q\Leftrightarrow p\sim_{\IO{w}}q$ for all $p,q\in\mathbb{P}$.

\begin{proof}
($\Rightarrow$) Suppose $p\bis{w}q$, and $p\models\phi$ for some $\phi\in\IO{w}$. We prove $q\models\phi$, by structural induction on $\phi$. The reverse implication ($q\models\phi$ implies $p\models\phi$) follows by symmetry.
\begin{itemize}
\item
$\phi=\bigwedge_{i\in I}\phi_i$. Then $p\models\phi_i$ for $i\in I$. By induction $q\models\phi_i$ for $i\in I$, so $q\models\bigwedge_{i\in I}\phi_i$.
\item
$\phi=\neg\phi'$. Then $p\not\models\phi'$. By induction $q\not\models\phi'$, so $q\models\neg\phi'$.
\item
$\phi=\eps\phi'$. Then $p \epsarrow p' \models\phi'$ for some term $p'$.
Since $p\bis{w}q$, according to Def.~\ref{def:bb}, $q\epsarrow q'$ for some $q'$ with $p'\bis{w}q'$.
 Since $p'\models\phi'$, by induction, $q'\models\phi'$.
 Hence $q\models\eps\phi'$.
\item
$\phi=\eps\diam{a}\eps\phi'$, with $a\in A$. Then $p \epsarrow\trans{a}\epsarrow p' \models\phi'$
  for some term $p'$.
 Since $p\bis{w}q$, according to Def.~\ref{def:bb}, $q\epsarrow\trans{a}\epsarrow q'$ for some $q'$
 with $p'\bis{w}q'$. Since $p'\models\phi'$, by induction, $q'\models\phi'$.
 Hence $q\models\eps\diam{a}\eps\phi'$.
\end{itemize}
We conclude that $p\sim_{\IO{w}}q$.

\vspace{2mm}
\noindent
($\Leftarrow$)
We prove that $\sim_{\IO{w}}$ is a weak bisimulation. The relation is clearly symmetric. Let
$p\sim_{\IO{w}}q$. Suppose \nietplat{$p\trans{\alpha}p'$}. If $\alpha=\tau$ and $p'\sim_{\IO{w}}q$, then
the first condition of Def.~\ref{def:bb} is fulfilled. So we can assume that either (i)
$\alpha\neq\tau$ or (ii) $p'\not\sim_{\IO{w}}q$. Let
\[
\begin{array}{lcl}
Q' &:=& \{q'\in\mathbb{P}\mid q\epsarrow \trans\alpha \epsarrow q'\land p'\not\sim_{\IO{w}}q'\}\;.
\end{array}
\]
For each $q'\in Q'$, let $\psi_{q'}$ be a formula in $\IO{w}$ such that $p'\models\psi_{q'}$ and $q'\not\models\psi_{q'}$. We define
\[
\psi=\bigwedge_{q'\in Q'}\psi_{q'}\;.
\]
Clearly, $\psi\in\IO{w}$ and $p'\models\psi$. Moreover, $q'\models\psi$ for no $q'\in Q'$.
We distinguish two cases.
\begin{enumerate}
\item
$\alpha\neq\tau$. Since $p\models\eps\diam\alpha\eps\psi\in\IO{w}$ and $p\sim_{\IO{w}} q$, also
  $q\models\eps\diam\alpha\eps\psi$. Hence \nietplat{$q\epsarrow \trans\alpha \epsarrow q'$} with
  $q'\models\psi$. It follows that $q'\notin Q'$ and thus $p'\sim_{\IO{w}}q'$.
\item
$\alpha=\tau$ and $p'\not\sim_{\IO{w}}q$. Let $\tilde\phi\in\IO{w}$ such that $p'\models\tilde\phi$
  and $p,q\not\models\tilde\phi$. Since $p\models\eps(\tilde\phi\land\psi)\in\IO{w}$ and
  $p\sim_{\IO{w}}q$, also $q\models\eps(\tilde\phi\land\psi)$. So
  \nietplat{$q\epsarrow q'$} with $q'\models \tilde\phi\land\psi$.
  It follows that $q'\notin Q'$.
  As $q\not\models\tilde\phi$ we have $q'\neq q$ and thus $q\epsarrow\trans\tau\epsarrow q'$.
  Hence $p'\sim_{\IO{w}}q'$.
\end{enumerate}
Both cases imply that the second condition of Def.~\ref{def:bb} is fulfilled. We therefore conclude that $\sim_{\IO{w}}$ is a weak bisimulation.
\qed
\end{proof}
Using the first part of Thm.~\ref{thm:characterisation}, which was proved above, it is not hard to derive the second part of Thm.~\ref{thm:characterisation}, i.e.\ that $\IO{rw}$ is a modal characterisation of rooted weak bisimilarity.

\begin{proof}
($\Rightarrow$) Suppose $p\bis{rw}q$, and $p\models\phi$ for some $\phi\in\IO{rw}$. We prove $q\models\phi$, by structural induction on $\phi$. The reverse implication ($q\models\phi$ implies $p\models\phi$) follows by symmetry.
\begin{itemize}
\item
The cases $\phi=\bigwedge_{i\in I}\phi_i$ and $\phi=\neg\phi'$ go exactly as in the previous proof.
\item
  $\phi=\eps\diam{\alpha}\eps\hat\phi$, with $\hat\phi\in\IO{w}$.
  Then $p \epsarrow\trans{\alpha}\epsarrow p' \models\hat\phi$ for some term $p'$.
 Since $p\bis{rw}q$, according to Def.~\ref{def:rbb}, $q\epsarrow\trans{\alpha}\epsarrow q'$ for some $q'$
 with $p'\bis{w}q'$. Since $p'\models\phi'$, by the previous result, $q'\models\phi'$.
 Hence $q\models\eps\diam{a}\eps\phi'$.
\item
  $\phi\in\IO{w}$. Since $p\bis{rw}q$ implies $p\bis{w}q$, the previously result yields $q\models\phi$.
\end{itemize}
We conclude that $p\sim_{\IO{rw}}q$.

\vspace{2mm}
\noindent
($\Leftarrow$)
We prove that $\sim_{\IO{rw}}$ is a rooted weak bisimulation. The relation is clearly symmetric. Let
$p\sim_{\IO{rw}}q$. Suppose \nietplat{$p\trans{\alpha}p'$}. Let
\[
\begin{array}{lcl}
Q' &:=& \{q'\in\mathbb{P}\mid q\epsarrow \trans\alpha \epsarrow q'\land p'\not\sim_{\IO{w}}q'\}\;.
\end{array}
\]
For each $q'\in Q'$, let $\psi_{q'}$ be a formula in $\IO{w}$ such that $p'\models\psi_{q'}$ and $q'\not\models\psi_{q'}$. We define
\[
\psi=\bigwedge_{q'\in Q'}\psi_{q'}\;.
\]
Clearly, $\psi\in\IO{w}$ and $p'\models\psi$. Moreover, $q'\models\psi$ for no $q'\in Q'$.

Since $p\models\eps\diam\alpha\eps\psi\in\IO{rw}$ and $p\sim_{\IO{rw}} q$, also
  $q\models\eps\diam\alpha\eps\psi$. Hence \nietplat{$q\epsarrow \trans\alpha \epsarrow q'$} with
  $q'\models\psi$. It follows that $q'\notin Q'$ and thus $p'\sim_{\IO{w}}q'$.
  By the previous result this implies $p'\bis{w}q'$.

Hence the condition of Def.~\ref{def:rbb} is fulfilled. We therefore conclude that $\sim_{\IO{rw}}$
is a rooted weak bisimulation.
\qed
\end{proof}
The validity of the modal characterisation of (rooted) delay bisimilarity can be proved in a similar fashion.

\section{Manifest delay resistance}\label{app:manifest}

This appendix contains the proofs of Thms.~\ref{thm:manifest} and~\ref{thm:manifest-Lambda}.

\subsection[tau-Pollable rules]{$\tau$-Pollable rules}

As a major step towards Thms.~\ref{thm:manifest} and~\ref{thm:manifest-Lambda}, we would like to
show that if all rules in a TSS $P$ are {\dr} w.r.t.\ \Brac{$\Lambda$ and} $P$, then so are all
ntytt rules that are linearly provable from $P$. However, we can prove this only if we
assume all rules in $P$ to have an additional property, which we call $\tau$-pollability w.r.t.\ \Brac{$\Lambda$ and} $P$.
In Sec.~\ref{sec:delay-resistance} we introduced the concept of a $\tau$-pollable premise in a rule $r$:
a positive premise that could be replaced by a similar premise with label $\tau$ and a fresh
right-hand side. Below, a rule will be called $\tau$-pollable if \emph{all} its premises \Brac{with
the exception of $\Lambda$-liquid ones} can be replaced in this manner; however, for standard rules
this possibly comes at the expense of changing the label of the conclusion of $r$ to $\tau$,
and its target to some term $w$, possibly containing the fresh right-hand sides mentioned above.
It turns out that in a manifestly delay resistant TSS all rules are $\tau$-pollable.

\begin{definition}\rm\label{def:pollable}
An ntytt rule \nietplat{$\frac{ H}{t\ntrans{\alpha}}$}, resp.\ \nietplat{$\frac{ H}{t\trans{\alpha}u}$},
is \emph{$\tau$-pollable} w.r.t.\ \Brac{a predicate $\Lambda$ and} a TSS $P$
if for each set $M \subseteq H^+\Brac{\setminus H^\Lambda}$ there is a rule
\nietplat{$\frac{H_M}{t\ntrans{\alpha}}$}, resp.\ 
\nietplat{$\frac{H_M}{t\trans{\alpha}u}$} or \nietplat{$\frac{H_M}{t\trans{\tau}w}$} for some term $w$,
linearly provable from $P$, where  $H_M \subseteq (H{\setminus}M) \cup M_\tau$.
\end{definition}
We recall that $M_\tau$ was defined in Def~\ref{def:positive}; it is obtained from $M$ by replacing the transition
labels by $\tau$, and the right-hand sides by fresh variables, not occurring in
\plat{$\frac{H}{t\ntrans{\alpha}}$}, resp.\ \plat{$\frac{ H}{t\trans{\alpha}u}$}.
Those variables may occur in $w$, however.

Again, we introduce a ``manifest'' version of this concept, where the required rule needs to be in
$\overline R$, rather than merely being linearly provable from $P$; we consider this notion only for a TSS as a whole.

\begin{definition}\rm\label{def:within}
A TSS $P=(\Sigma,R)$ in ntytt format is called \emph{manifestly $\tau$-pollable} \Brac{w.r.t.\ $\Lambda$} if for each
rule \nietplat{$\frac{ H}{t\ntrans{\alpha}}$} or \nietplat{$\frac{ H}{t\trans{\alpha}u}$}
in $R$, and for each set $M \subseteq H^+\Brac{\setminus H^\Lambda}$,
$\overline R$ also contains a rule \nietplat{$\frac{H_M}{t\ntrans{\alpha}}$}, resp.\ 
\nietplat{$\frac{H_M}{t\trans{\alpha}u}$} or \nietplat{$\frac{H_M}{t\trans{\tau}w}$} for some $w$,
where  $H_M \subseteq (H{\setminus}M) \cup M_\tau$.
\end{definition}
Clearly, in a manifestly $\tau$-pollable TSS $P$, each rule is $\tau$-pollable w.r.t.\ $P$.

\begin{lemma}\label{lem:pollable}
Any manifestly delay resistant \Brac{w.r.t.\ $\Lambda$} standard TSS $P=(\Sigma,R)$ in decent ntyft format is manifestly
$\tau$-pollable \Brac{w.r.t.\ $\Lambda$}.
\end{lemma}

\begin{proof}
Consider a rule $\frac{ H}{t\trans{\alpha}u}$ in $R$ and an $M \subseteq H^+\Brac{\setminus H^\Lambda}$.
We apply induction on $|M|$, taking into account that $M$ may be infinite. The induction base $M=\emptyset$ is trivial.
In the induction step, $M\neq\emptyset$.
Pick a finite set $H^d$ of manifestly delayable positive premises with the property formulated in Def.~\ref{def:manifestly delay resistant}.
First we deal with the case that $M\not\subseteq H^d$. Let $M_1 := M\setminus H^d$.
By Def.~\ref{def:manifestly delay resistant} there is a rule
\nietplat{$\frac{H_1}{t\trans{\alpha}u}$} in $\overline R$ with $H_1 \subseteq (H{\setminus}M_1) \cup (M_1)_\tau$.
Let $M_2:=M \cap H^d\cap H_1^+$. As $H_d$ is finite, so is $M_2$. Since  $M_2 \subseteq H_1^+\Brac{\setminus H_1^\Lambda}$ and
$|M_2| < |M|$, by induction $\overline R$
contains a rule \nietplat{$\frac{H_2}{t\trans{\alpha}u}$} or \nietplat{$\frac{H_2}{t\trans{\tau}w}$} for some term $w$,
where $H_2 \subseteq (H_1{\setminus}M_2) \cup (M_2)_\tau \subseteq (H{\setminus}M) \cup M_\tau$.

Next assume that $M\subseteq H^d$. As $H_d$ is finite, so is $M$. Pick a $v \trans\beta y$ in $M$.
Since it is a manifestly delayable premise of $r$, there exists a rule
\plat{$\frac{H_1}{t \trans{\tau} w}$} in $\overline R$, for some term $w$ and fresh variable $z$,
with \nietplat{$H_1\subseteq (H{\setminus}\{v \trans{\beta} y\})\cup\{v \trans{\tau} z\}$}.
Let \nietplat{$M_1 := M \cap (H_1^+\setminus \{v \trans{\tau} z\})$}. Since $M_1 \subseteq H_1^+\Brac{\setminus H_1^\Lambda}$ and
$|M_1|<|M|$, by induction $\overline R$ contains a rule \plat{$\frac{H_2}{t\trans{\tau}w}$} for some term $w$,
where $H_2 \subseteq (H_1{\setminus}M_1) \cup (M_1)_\tau \subseteq (H{\setminus}M) \cup M_\tau$.
(Here we use that there is a $v \trans{\tau} z_y$ in $M_\tau$, and that we can choose $z_y:=z$.)
\qed
\end{proof}

\subsection[Lifting manifest delay resistance]{Lifting manifest delay resistance from $P$ to $\hat P^\ddagger$}

We extend the predicate ``\pdr'' to non-standard rules by declaring it vacuously true.

\begin{lemma}\label{lem:preservation positive}
Let $P$ be a TSS in decent ntytt format, in which each transition rule is {\pdr} as well as
$\tau$-pollable w.r.t.\  $P$.
Then any ntytt rule linearly provable from $P$ is {\pdr} as well as $\tau$-pollable w.r.t.\  $P$.
\end{lemma}

\begin{proof}
Let an ntytt rule \nietplat{$r=\frac{ H}{t\ntrans{\alpha}}$} [resp.\ \nietplat{$\frac{H}{t\trans\alpha u}$}]
be linearly provable from $P$, by means of a proof $\pi$.
We will prove, by structural induction on $\pi$, that this rule is {\pdr} as well as $\tau$-pollable w.r.t.\  $P$.

\vspace{2mm}

\noindent
{\em Induction basis}:
Suppose $\pi$ has only one node, marked ``hypothesis''. Then rule $r$ has the form
\plat{$\frac{t\ntrans{\alpha}}{t\ntrans{\alpha}}$} [resp.\
\nietplat{$\frac{t\trans{\alpha}y}{t\trans{\alpha}y}$}].\vspace{-3pt}
The non-standard rule is trivially $\tau$-pollable, by lack of positive premises.
The standard rule satisfies the requirements of Defs.~\ref{def:positive} and~\ref{def:delayable}
by taking $H^d\mathbin{=}H_\emptyset\mathbin{:=}\{t\trans{\alpha}y\}$,
%$w\mathbin{:=}t$, $\beta\mathbin{:=}\alpha$, $y\mathbin{:=}y$,
$H_1\mathbin{:=}\{t\trans\tau z\}$, $v\mathbin{:=}z$ and $H_2\mathbin{:=}\{z\trans\alpha y\}$.
Furthermore, the requirement of Def.~\ref{def:pollable} is satisfied through
the rule \nietplat{$\frac{t\trans{\tau}y}{t\trans{\tau}y}$}.

\vspace{2mm}

\noindent
{\em Induction step}:
Let $r'=\frac{K}{t'\ntrans\alpha}$ [resp.\ $\frac{K}{t'\trans\alpha u'}$] be the rule and $\sigma$
the substitution used at the bottom of $\pi$---by assumption, $r'$ is decent, ntytt and {\pdr} as well as $\tau$-pollable w.r.t.\ $P$.
Then $\sigma(t')=t$ [and $\sigma(u')=u$]. Moreover, rules
\nietplat{$r_\mu \mathbin= \frac{H_\mu}{\sigma(\mu)}$} for each $\mu\mathbin\in K$ are
linearly provable from $P$ by means of strict subproofs of $\pi$, where $H=\bigcup_{\mu\in K}H_\mu$,
and the sets $H_\mu$ are pairwise disjoint.

For each $\mu\in K$, let $t_\mu$ be the left-hand side of $\mu$.
As $r'$ is decent, $\var(t_\mu)\subseteq \var(t')$, so $\var(\sigma(t_\mu))\subseteq\var(\sigma(t'))=\var(t)$.
From ${\it rhs}(H)\cap\var(t)=\emptyset$ it follows that ${\it rhs}(H_{\mu})\cap\var(\sigma(t_\mu))=\emptyset$.
So $r_{\mu}$ is an ntytt rule. By induction, $r_{\mu}$ is {\pdr} as well as $\tau$-pollable w.r.t.\ $P$.

To show that $r$ is $\tau$-pollable, pick any set $M\subseteq H^+$.
It can be written as $M= \bigcup_{\mu \in K} M_\mu$ with $M_\mu \subseteq H_\mu^+$ for all $\mu\mathbin\in K$.
For each negative premise $\mu\in K$, as $r_\mu$ is $\tau$-pollable,
a rule \nietplat{$\frac{H_\mu^M}{\sigma(\mu)}$}, with
$H_\mu^M\subseteq (H_\mu{\setminus}M_\mu) \cup (M_\mu)_\tau$, is linearly provable from $P$.
For each positive $\mu\in K$, say of the form $t_\mu \trans{\gamma_\mu} y_\mu$, as $r_\mu$ is
$\tau$-pollable, a rule \nietplat{$\frac{H_\mu^M}{\sigma(t_\mu) \trans{\gamma_\mu} \sigma(y_\mu)}$}
or \nietplat{$\frac{H_\mu^M}{\sigma(t_\mu) \trans\tau w_\mu}$}, with
$H_\mu^M\subseteq (H_\mu{\setminus}M_\mu) \cup (M_\mu)_\tau$, is linearly provable from $P$.
In the construction of the sets $(M_\mu)_\tau$ (see Def.~\ref{def:positive}), we make sure that the
fresh right-hand sides $z_y$ are all different.
Let $M^K$ be the set of literals $\mu$ from $K^+$ for which only the second of these two possibilities applies.
Since $r'$ is $\tau$-pollable, there is a rule
\nietplat{$\frac{K^M}{t'\ntrans{\alpha}}$} [resp.\
\nietplat{$\frac{K^M}{t'\trans{\alpha}u'}$} or
\nietplat{$\frac{K^M}{t'\trans{\tau}w'}$}], linearly provable from $P$, with
$K^M\mathbin\subseteq (K{\setminus} M^K) \cup M^K_\tau$. 
Lem.~\ref{linear proof composition} yields a linear proof from $P$, which uses this rule and $\sigma$
at the bottom, of a rule \nietplat{$\frac{H^M}{t\ntrans{\alpha}}$} [resp.\ \nietplat{$\frac{H^M}{t\trans{\alpha}u}$} or
\nietplat{$\frac{H^M}{t\trans{\tau}w}$}] with $H^M\subseteq \bigcup_{\mu\in K}H^M_\mu = (H{\setminus}M) \cup M_\tau$,
as required by Def.~\ref{def:pollable}.

It remains to show that $r$ is \pdr.
So assume \nietplat{$r=\frac{H}{t\trans\alpha u}$} and \nietplat{$r'=\frac{K}{t'\trans\alpha u'}$}.
Let $K^d \subseteq K^{+}$ and $H^d_\mu \subseteq H^+_\mu$ for each $\mu\in K$ be the finite sets of
delayable positive premises of $r'$ resp.\ $r_\mu$, which exist by Def.~\ref{def:positive}. Take $H^d := \bigcup _{\mu \in K^d} H_\mu^d$.
Pick any set $M\subseteq H^+\setminus H^d$.
It can be written as $M= \bigcup_{\mu \in K} M_\mu$ with $M_\mu \subseteq H^+_\mu$ for each $\mu\mathbin\in K$;
moreover, $M_\mu \subseteq H_\mu^{+}\setminus H^d_\mu$ for each \nietplat{$\mu\mathbin\in K^d$}.
For each negative premise $\mu\in K$, as $r_\mu$ is $\tau$-pollable,
a rule \nietplat{$\frac{H_\mu^M}{\sigma(\mu)}$} with
$H_\mu^M\subseteq (H_\mu\setminus M_\mu) \cup (M_\mu)_\tau$ is linearly provable from $P$.
For each $\mu\in K^{+}$, say of the form $t_\mu \trans{\gamma_\mu} y_\mu$, as $r_\mu$ is
$\tau$-pollable, a rule $\frac{H_\mu^M}{\sigma(t_\mu) \trans{\gamma_\mu} \sigma(y_\mu)}$
or $\frac{H_\mu^M}{\sigma(t_\mu) \trans\tau w_\mu}$ with $H_\mu^M \subseteq (H_\mu{\setminus}M_\mu) \cup (M_\mu)_\tau$
is linearly provable from $P$.
In the construction of the sets $(M_\mu)_\tau$ we make sure that the fresh right-hand sides $z_y$ are all different.
In the special case that $\mu\in K^d$ we have $M_\mu\subseteq H_\mu^{+}\setminus H^d_\mu$, so that
the positive delay resistance of $r_\mu$ guarantees that the first of these two possibilities applies.
Let $M^K$ be the set of $\mu\in K^+$ for which only the second possibility
applies. Then $M^K \subseteq K^+\setminus K^d$. Since $r'$ is \pdr, there is a rule
\nietplat{$\frac{K^M}{t'\trans{\alpha}u'}$}, linearly provable from $P$, with
$K^M\subseteq (K{\setminus}M^K) \cup M^K_\tau\!$.
Lem.~\ref{linear proof composition} yields a linear proof from $P$, which uses this rule and $\sigma$
at the bottom, of a rule \nietplat{$\frac{H^M}{t\trans{\alpha}u}$} with
$H^M\subseteq \bigcup_{\mu\in K}H^M_\mu = (H{\setminus}M) \cup M_\tau$, as required by
Def.~\ref{def:positive}.

It remains to show that any literal in \nietplat{$H^d$}, say of the form $w\trans\beta y$, is a delayable premise of $r$.
By the definition of $H^d$, $w\trans\beta y$ is in $H_{\mu_0}^d$ for some $\mu_0\in K^d$; say $\mu_0$ is of the form $t_0\trans\gamma y_0$.
Since $w\trans\beta y$ is a delayable premise of $r_{\mu_0}$, there are rules
$\frac{H^1_{\mu_0}}{\sigma(t_0)\trans{\tau}v}$ and $\frac{H^2_{\mu_0}}{v\trans{\gamma}\sigma(y_0)}$,
linearly provable from $P$, with $H^1_{\mu_0}\subseteq (H_{\mu_0}{\setminus}\{w \trans{\beta} y\})\cup\{w \trans{\tau} z\}$
and $H^2_{\mu_0}\subseteq (H_{\mu_0}{\setminus}\{w \trans{\beta} y\})\cup\{z \trans{\beta} y\}$
for some term $v$ and fresh variable $z$.
Likewise, since $t_0\trans\gamma y_0$ is a delayable premise of $r'$, there are rules
$\frac{K_1}{t'\trans{\tau}v'}$ and $\frac{K_2}{v'\trans{\alpha}u'}$,
linearly provable from $P$, with $K_1\subseteq (K{\setminus}\{t_0 \trans{\gamma} y_0\})\cup\{t_0 \trans{\tau} z'\}$
and \nietplat{$K_2\subseteq (K{\setminus}\{t_0 \trans{\gamma} y_0\})\cup\{z' \trans{\gamma} y_0\}$}
for some term $v'$ and fresh variable $z'$.
Without limitation of generality, we pick $z'$ so that it does not occur in $\pi$.
Let $\sigma'(z')=v$ and $\sigma'$ coincides with $\sigma$ on all other variables.
Let $L_1$ and $L_2$ denote $\bigcup_{\mu\in K_1{\setminus}\{t_0 \trans{\tau} z'\}}H_\mu \cup H^1_{\mu_0}$
resp.\ $\bigcup_{\mu\in K_2{\setminus}\{z' \trans{\gamma} y_0\}}H_\mu \cup H^2_{\mu_0}$.
Lem.~\ref{linear proof composition} yields linear proofs from $P$, which use
$\frac{K_1}{t'\trans{\tau}v'}$ and $\frac{K_2}{v'\trans{\alpha}u'}$ and $\sigma'$ at the bottom, of the rules
$\frac{L_1}{t\trans{\tau}\sigma'(v')}$ and $\frac{L_2}{\sigma'(v')\trans{\alpha}u}$.
Moreover, $L_1\subseteq (H{\setminus }\{w \trans{\beta} y\})\cup\{w \trans{\tau} z\}$ and
$L_2\subseteq (H{\setminus}\{w \trans{\beta} y\})\cup\{z \trans{\beta} y\}$.
\qed
\end{proof}
Only in the above proof, together with its forthcoming variant proving Lem.~\ref{lem:preservation positive-Lambda},
does it make a difference that linear provability is used instead
of irredundant provability. It allows us to infer that the sets $H_\mu$ are pairwise disjoint.
Without that, the last sentence in the proof would fail.
See also Ex.~\ref{ex:linearity}.

\begin{lemma}\label{lem:preservation stable}
Let $P$ be a negative-stable standard TSS in decent ntytt format.
Then any ntytt rule irredundantly provable from $P$ is negative-stable.
\end{lemma}

\begin{proof}
Let an ntytt rule \nietplat{$r=\frac{H}{t\trans\alpha u}$}
be irredundantly provable from $P$, by means of a proof $\pi$.
We will prove, by structural induction on $\pi$, that $r$ is negative-stable.

\vspace{2mm}

\noindent
{\em Induction basis}:
If $\pi$ has only one node, marked ``hypothesis'', then \nietplat{$\frac{H}{t\trans{\alpha} u}$} equals \nietplat{$\frac{t\trans{\alpha}u}{t\trans{\alpha}u}$},
which trivially is negative-stable.

\vspace{2mm}

\noindent
{\em Induction step}:
Let $r'=\frac{K}{t'\trans\alpha u'}$ be the rule and $\sigma$
the substitution used at the bottom of $\pi$---by assumption, $r'$ is decent, ntytt and negative-stable;
so $K=K^+\cup K^{s-}$.
Then $\sigma(t')=t$ and $\sigma(u')=u'$. Moreover, rules
\nietplat{$r_\mu \mathbin= \frac{H_\mu}{\sigma(\mu)}$} for each $\mu\mathbin\in K^+$ are
irredundantly provable from $P$ by means of strict subproofs of $\pi$, where $H=\bigcup_{\mu\in K^+}H_\mu \cup \sigma(K^{s-})$.

For each $\mu\mathbin\in K^+$, let $t_\mu$ be the left-hand side of $\mu$.
As $r'$ is decent, $\var(t_\mu)\subseteq \var(t')$, so $\var(\sigma(t_\mu))\subseteq\var(\sigma(t'))=\var(t)$.
From ${\it rhs}(H)\cap\var(t)=\emptyset$ it follows that ${\it rhs}(H_{\mu})\cap\var(\sigma(t_\mu))=\emptyset$.
So $r_{\mu}$ is an ntytt rule.
By induction, $r_{\mu}$ is negative-stable.
From this it follows that $r$ is negative-stable.
\qed
\end{proof}

\begin{proposition}\label{prop:manifest preservation}
Let $P$ be a standard TSS in ready simulation format. If $P$ is manifestly delay resistant,
then so is the standard TSS $\hat P^\ddagger$ in xynft format constructed in Secs.~\ref{sec:ruloids}--\ref{sec:linearity}.
\end{proposition}

\begin{proof}
Let $P$ be manifestly delay resistant.
By Def.~\ref{def:manifestly}, the conversion $P^\dagger$ of $P$ to decent ntyft format, defined in
Sec.~\ref{sec:ruloids}, is manifestly delay resistant.
By Lem.~\ref{lem:pollable} it is also manifestly $\tau$-pollable.
Hence all its rules are $\tau$-pollable w.r.t.\ $P^\dagger$, and manifestly delay resistant
w.r.t.\ $P^\dagger$---thus negative-stable and {\pdr} w.r.t.\ $P^\dagger$.

The TSS $\hat P^\ddagger$ contains all xynft rules linearly provable from $P^\dagger$.
By Lem.~\ref{lem:preservation stable} those rules are negative-stable.
By Lem.~\ref{lem:preservation positive} they are {\pdr} and
$\tau$-pollable w.r.t.\ $P^\dagger$. The definitions of {\pdr} and
$\tau$-pollability of rules from $\hat P^\ddagger$ w.r.t.\ $P^\dagger$ require the existence of certain
xynft rules that are linearly provable from $P^\dagger$. By definition, these are rules of $\hat P^\ddagger$.
Hence $\hat P^\ddagger$ is manifestly delay resistant, as well as manifestly $\tau$-pollable.
\qed
\end{proof}

\subsection[Eliminating the Lambda-restriction]{Eliminating the $\Lambda$-restriction}

Prop.~\ref{prop:manifest preservation} above is a crucial step (or ``halfway marker'')
in the proof of Thm.~\ref{thm:manifest}. In this section we take a similar step
(Prop.~\ref{prop:eliminating-Lambda}) towards the proof of Thm.~\ref{thm:manifest-Lambda}.
It requires $P$ to be merely manifestly delay resistant w.r.t.\ $\Lambda$,
rather than outright manifestly delay resistant. On the other hand, it assumes the extra
antecedents of Thm.~\ref{thm:manifest-Lambda}. The conclusion of Prop.~\ref{prop:eliminating-Lambda}
is that $P$ is manifestly delay resistant. This allows us to dispense with $\Lambda$ in the
second half of the proof of Thm.~\ref{thm:manifest-Lambda},
which therefore will equal the second half of the proof of Thm.~\ref{thm:manifest}.

Our initial proof strategy was to extend Lem.~\ref{lem:preservation positive}---similar to
Lem.~\ref{lem:pollable}---by replacing ``w.r.t.\ $P$'' by ``w.r.t.\ [$\Lambda$ and] $P$''
both in the antecedent and in the conclusion. Then
Lemmas~\ref{lem:delay-resistant}--\ref{lem:substitution} would suffice to establish
Prop.~\ref{prop:eliminating-Lambda}, eliminating $\Lambda$.
This strategy failed, however. It turns out we have to integrate
the proof of Lem.~\ref{lem:delay-resistant} into the proof of the modified 
Lem.~\ref{lem:preservation positive}, and eliminate $\Lambda$ while lifting positive delay
resistance and $\tau$-pollability from the rules of $P$ to the linearly provable xynft rules.
This yields the forthcoming Lem.~\ref{lem:preservation positive-Lambda}.

\newcommand{\irr}{\vdash_{\rm lin}}		    % linearly provable
\newcommand{\pvb}{\vdash_{\rm lin}}		    % linearly provable
The following lemma is a variant of Prop.~\ref{prop:ruloid}, needed in the proof of 
Lem.~\ref{lem:preservation positive-Lambda}.
We write $P\irr r$ to say that a rule $r$ is linearly provable from a TSS $P$.

\begin{lemma}\label{lem:ruloids}
Let $P=(\Sigma,R)$ be a TSS in decent ntyft format, $t,t' \in \mathbb{T}(\Sigma)$ and $\sigma$ a substitution.
If $P\pvb r$ with $r=\frac{H}{\sigma(t)\ntrans{\alpha}}$ [resp.\ $\frac{H}{\sigma(t)\trans \alpha t'}$] an xyntt rule,
then there are a decent xyntt rule $r'=\frac{G}{t\ntrans \alpha}$ [resp.\ $\frac{G}{t\trans \alpha u}$] with $P\irr r'$
and a substitution $\sigma'$ with $\sigma'(t)=\sigma(t)$ [and $\sigma'(u)=t'$] such that $H$ can be
written as $\biguplus_{\nu \in G} H_\nu$ with $P\pvb\frac{H_\nu}{\sigma'(\nu)}$ for all $\nu\in G$.
Moreover, any proof $\pi$ of $r$ from $P$ has subproofs $\pi_\nu$ of \nietplat{$\frac{H_\nu}{\sigma'(\nu)}$};
and if neither $t$ is a variable nor $\pi$ a 1-node proof, the $\pi_\nu$ are strict subproofs of $\pi$.
\end{lemma}

\begin{proof}
First, suppose $t$ is a variable. By default, the decent xyntt
rule \nietplat{$\frac{t\ntrans \alpha}{t\ntrans \alpha}$} [resp.\ $\frac{t\trans \alpha
y}{t\trans \alpha y}$] is linearly provable from $P$. Let $\sigma'$
be a substitution with $\sigma'(t)=\sigma(t)$ [and $\sigma'(y)=t'$].
Clearly, $P\pvb\frac{H}{\sigma'(t\ntrans \alpha)}$ [resp.\ $P\pvb\frac{H}{\sigma'(t\trans \alpha y)}$],
taking $\pi_\nu := \pi$.

Next, suppose $t=f(t_1,\ldots,t_{{\it ar}(f)})$.
We apply structural induction on the proof $\pi$ of $r$ from $P$.

\vspace{2mm}

\noindent
{\em Induction basis}:
If $\pi$ has only one node, marked ``hypothesis'', using that $r$ is an xyntt rule, \plat{$r=\frac{t\ntrans{\alpha}}{t\ntrans{\alpha}}$}.
Take $r':=r$, $\sigma':=\sigma$, $H_\nu:=H$ and $\pi_\nu:=\pi$.

\vspace{2mm}

\noindent
{\em Induction step}:
Let $r''\in R$ be the decent ntyft rule and $\rho$ the
substitution used at the bottom of $\pi$, where $r''$ is of the form
$\frac{\{v_k\trans{\gamma_k}y_k\mid k\in K\}\cup
\{w_\ell\ntrans{\delta_\ell}\mid \ell\in L\}}
{f(x_1,\ldots,x_{{\it ar}(f)})\ntrans \alpha}$
[resp.\ $\frac{\{v_k\trans{\gamma_k}y_k\mid k\in K\}\cup
\{w_\ell\ntrans{\delta_\ell}\mid \ell\in L\}}
{f(x_1,\ldots,x_{{\it ar}(f)})\trans \alpha v}$].
Then $\rho(x_i)=\sigma(t_i)$ for $i=1,\ldots,{\it ar}(f)$, [$\rho(v)=t'$,]
and rules $\frac{H_k}{\rho(v_k)\trans{\gamma_k}\rho(y_k)}$ for $k\in K$ and
$\frac{H_\ell}{\rho(w_\ell)\ntrans{\delta\ell}}$ for $\ell\in L$ are linearly provable from $P$
by means of strict subproofs of $\pi$, where $H=\biguplus_{k\in K}H_k \cup\biguplus_{\ell\in L}H_\ell$.
Since $r''$ is decent, $\var(v_k)$ for $k\in K$ and
$\var(w_\ell)$ for $\ell\in L$ are included in $\{x_1,\ldots,x_{{\it ar}(f)}\}$.
Let $\rho_0$ be a substitution with $\rho_0(x_i)=t_i$ for $i=1,\ldots,{\it ar}(f)$.
As $\rho(x_i)=\sigma(t_i)=\sigma(\rho_0(x_i))$ for $i=1,\ldots,{\it
ar}(f)$, we have $\rho(v_k)=\sigma(\rho_0(v_k))$ for $k\in K$ and
$\rho(w_\ell)=\sigma(\rho_0(w_\ell))$ for $\ell\in L$.
So $\frac{H_k}{\sigma(\rho_0(v_k))\trans{\gamma_k}\rho(y_k)}$ for $k\in K$ and
$\frac{H_\ell}{\sigma(\rho_0(w_\ell))\ntrans{\delta_\ell}}$ for $\ell\in L$ are linearly provable from $P$
by means of strict subproofs $\pi_k$ and $\pi_\ell$ of $\pi$. According to the induction hypothesis,
for $k\in K$ there are a decent xyntt rule
$\frac{G_k}{\rho_0(v_k)\trans{\gamma_k}u_k}$ and a substitution $\sigma_k'$
with $P\irr\frac{G_k}{\rho_0(v_k)\trans{\gamma_k}u_k}$,
$\sigma_k'(\rho_0(v_k))=\sigma(\rho_0(v_k))$
and $\sigma_k'(u_k)=\rho(y_k)$, such that $H_k$ can be written
as $\biguplus_{\nu \in G_k} H_\nu$ with $P\pvb\frac{H_\nu}{\sigma'_k(\nu)}$ for all $\nu\in G_k$.
Moreover, $\pi_k$ has subproofs $\pi_\nu$ of \nietplat{$\frac{H_\nu}{\sigma'_k(\nu)}$}.
Likewise, for $\ell\in L$ there are a
decent xyntt rule $\frac{G_\ell}{\rho_0(w_\ell)\ntrans{\delta_\ell}}$
and a substitution $\sigma_\ell'$
with $P\irr\frac{G_\ell}{\rho_0(w_\ell)\ntrans{\delta_\ell}}$
and $\sigma_\ell'(\rho_0(w_\ell))=\sigma(\rho_0(w_\ell))$, such that $H_\ell$ can be written
as $\biguplus_{\nu \in G_\ell} H_\nu$ with $P\pvb\frac{H_\nu}{\sigma'_\ell(\nu)}$ for all $\nu\in G_\ell$.
Moreover, $\pi_\ell$ has subproofs $\pi_\nu$ of \nietplat{$\frac{H_\nu}{\sigma'_\ell(\nu)}$}.
As observed in \cite{BFvG04}, using that $|\Sigma|,|A|\leq |V|$,
we can choose the sets of variables in the
right-hand sides of the positive premises in the $G_k$ (for $k\in K$)
and $G_\ell$ (for $\ell\in L$) pairwise disjoint, and disjoint from
$\var(t)$. This allows us to define a substitution $\sigma'$ with:
\begin{itemize}
\item
$\sigma'(z)=\sigma(z)$ for $z\in\var(t)$;
\item
$\sigma'(z)=\sigma_k'(z)$ for right-hand sides $z$ of positive premises
in $G_k$ for $k\in K$;
\item
$\sigma'(z)=\sigma_\ell'(z)$ for right-hand sides $z$ of positive premises
in $G_\ell$ for $\ell\in L$.\pagebreak[3]
\end{itemize}
Let $G$ denote $\bigcup_{k\in K}G_k\cup\bigcup_{\ell\in L}G_\ell$.
Moreover, let $\rho_1$ be a substitution with $\rho_1(x_i)=t_i$ for
$i=1,\ldots,{\it ar}(f)$ and $\rho_1(y_k)=u_k$ for $k\in K$.
We verify that the rule \nietplat{$\frac{G}{t\ntrans \alpha}$}
[resp.\ \nietplat{$\frac{G}{t\trans \alpha\rho_1(v)}$}] together
with the substitution $\sigma'$ satisfy the desired properties.

As $\var(v_k)\subseteq\{x_1,\ldots,x_{{\it ar}(f)}\}$, it follows that
$\var(\rho_0(v_k))\subseteq\var(t)$. Since $\sigma'$ and $\sigma$ agree
on $\var(t)$, $\sigma'(\rho_0(v_k))=\sigma(\rho_0(v_k))=\sigma_k'(\rho_0(v_k))$
for $k\in K$. Thus, by the decency of $\frac{G_k}{\rho_0(v_k)\trans{\gamma_k}u_k}$,
$\sigma'$ and $\sigma'_k$ agree on all variables
occurring in this rule for $k\in K$. Likewise,
$\sigma'$ and $\sigma'_\ell$ agree on all variables occurring in
\nietplat{$\frac{G_\ell}{\rho_0(w_\ell)\ntrans{\delta_\ell}}$} for $\ell \in L$.

\begin{enumerate}
\item
$\rho_1$ and $\rho_0$ agree on $\var(v_k)\subseteq
\{x_1,\ldots,x_{{\it ar}(f)}\}$, so $\rho_1(v_k)=\rho_0(v_k)$ for $k\!\in\! K$.\vspace{-2pt}
Likewise, $\rho_1(w_\ell)\linebreak[3]=\rho_0(w_\ell)$ for $\ell\!\in\! L$.\vspace{-4pt}
Since $P\irr r''$, we have $P\irr\rho_1(r'')=
\frac{\{\rho_0(v_k)\trans{\gamma_k}u_k\mid k\in K\}\cup
\{\rho_0(w_\ell)\ntrans{\delta_\ell}\mid \ell\in L\}}{t\ntrans \alpha}$
[resp.\ $\frac{\{\rho_0(v_k)\trans{\gamma_k}u_k\mid k\in K\}\cup
\{\rho_0(w_\ell)\ntrans{\delta_\ell}\mid \ell\in L\}}{t\trans \alpha\rho_1(v)}$].\\
Furthermore, $P\irr \frac{G_k}{\rho_0(v_k)\trans{\gamma_k}u_k}$ for $k\in K$ and
$P\irr\frac{G_\ell}{\rho_0(w_\ell)\ntrans{\delta_\ell}}$ for $\ell\in L$.
As \nietplat{$G=\bigcup_{k\in K}G_k\cup\bigcup_{\ell\in L}G_\ell$}, it follows that
$P\irr\frac{G}{t\ntrans \alpha}$ [resp.\ $P\irr\frac{G}{t\trans \alpha\rho_1(v)}$].

\item
The right-hand sides of the positive premises in any $G_k$
or $G_\ell$ are distinct variables.  By construction, these sets of
variables (one for every $k\in K$ and $\ell\in L$) are pairwise
disjoint, and disjoint from $\var(t)$. Hence $\frac{G}{t\ntrans \alpha}$
[resp.\ $\frac{G}{t\trans \alpha\rho_1(v)}$] is an ntytt rule.
Since the positive premises in $G$ originate from $G_k$ (for $k\in K$) and
$G_\ell$ (for $\ell\in L$), their left-hand sides are variables.
This makes the rule an xyntt rule. The rule is decent by
Lem.~\ref{lem:preservation_decency}.

\item
Since $\sigma'$ and $\sigma$ agree on $\var(t)$, $\sigma'(t)=\sigma(t)$.

\item

[$\sigma'(\rho_1(x_i))=\sigma'(t_i)=\sigma(t_i)=\rho(x_i)$
for $i=1,\ldots,{\it ar}(f)$.
Moreover, since $\sigma'$ and $\sigma_k'$ agree on $\var(u_k)$,
$\sigma'(\rho_1(y_k))=\sigma'(u_k)=\sigma_k'(u_k)=\rho(y_k)$ for $k\in
K$. As $\var(v)\subseteq
\{x_1,\ldots,x_{{\it ar}(f)}\}\cup\{y_k\mid k\in K\}$, it follows that
$\sigma'(\rho_1(v))=\rho(v)=t'$.]

\item
$H=\biguplus_{k\in K}H_k \cup\biguplus_{\ell\in L}H_\ell =
\biguplus_{k\in K}\biguplus_{\nu \in G_k} H_\nu \cup\biguplus_{\ell\in L}\biguplus_{\nu \in G_\ell} H_\nu =
\biguplus_{\nu \in G} H_\nu$.

\item
$P\pvb\frac{H_\nu}{\sigma'(\nu)}$ for all $\nu\mathbin\in G$,
using that $\sigma'(\nu)\mathbin=\sigma'_k(\nu)$ when $\nu\mathbin\in G_k$ and
$\sigma'(\nu)\mathbin=\sigma'_\ell(\nu)$ when $\nu\mathbin\in G_\ell$.

\item
For $\nu\mathbin\in G$ the proof $\pi$ has strict subproofs $\pi_\nu$ of $\frac{H_\nu}{\sigma'(\nu)}$.
\hfill $\Box$
\end{enumerate}
\end{proof}
With this result in hand, we obtain the following variant of Lem.~\ref{lem:preservation positive}.

\begin{lemma}\label{lem:preservation positive-Lambda}
Let $P$ be an $\aL$-patient TSS in decent ntyft format, in which each transition rule is rooted branching bisimulation
safe and satisfies condition \ref{smooth} of Def.~\ref{def:smooth} w.r.t.\ $\aleph$ and $\Lambda$,
and is {\pdr} as well as $\tau$-pollable w.r.t.\ $\Lambda$ and $P$.
Then any xyntt rule linearly provable from $P$ is {\pdr} as well as $\tau$-pollable w.r.t.\ $P$.
\end{lemma}
Note the absence of the disclaimer ``w.r.t.\ $\Lambda$'' in the conclusion of the lemma.
\begin{proof}
This proof expands the proof of Lem.~\ref{lem:preservation positive}; we only explain the non-trivial differences.
The first difference is that we apply structural induction on the \emph{skeleton} of a proof $\pi$,
rather than on $\pi$ itself. Here the skeleton is the tree labelled with the rules that
are applied in each node, but not with the substitutions used, or the resulting literals.
This allows us to apply the induction hypothesis on proofs $\pi'$ whose skeleton is a proper
subtree of the skeleton of the proof $\pi$ currently under investigation, even when $\pi'$ itself is
not a subtree of $\pi$.

In the proof of Lem.~\ref{lem:preservation positive} we constructed, for each proof $\pi$ of an (ntytt) rule
$r\mathbin=\frac{H}{t\trans\alpha u}$, a finite set $H^d \mathbin\subseteq H^+$ of delayable premises of $r$.
Here, we do this so that if $\pi^\dagger$ is a strict subproof of $\pi$, proving an (xyntt)
rule $r^\dagger=\frac{H^\dagger}{\nu}$ (so that $H^\dagger \subseteq H$)
with $H^\dagger \cap H^d \neq \emptyset$, then $\nu$ is positive and
any $\mu \in H^\dagger \cap H^d$ is a delayable premise also of $r^\dagger$.\hfill (*)

When starting with an xyntt rule $r=\frac{H}{t\ntrans\alpha}$ [resp.\ $\frac{H}{t \trans\alpha u}$],
linearly provable from $P$ by means of a proof $\pi$, we first deal with the case that $t$ is univariate.
The general case is dealt with at the end of this proof.

Following the proof of Lem.~\ref{lem:preservation positive},
in the induction step $r'=\frac{K}{t'\ntrans\alpha}$ [resp.\ $\frac{K}{t'\trans\alpha u'}$]\vspace{-3pt}
is {\pdr} and $\tau$-pollable w.r.t.\ $\Lambda$ and $P$,
whereas the $r_\mu \mathbin= \frac{H_\mu}{\sigma(\mu)}$ for $\mu\mathbin\in K$ are {\pdr} and $\tau$-pollable w.r.t.~$P$;
they also satisfy (*) for subproofs of $\pi$.
By Lem.~\ref{lem:preservation_branching_bisimulation_safe}, $r$ is rooted branching bisimulation
safe w.r.t.\ $\aleph$ and $\Lambda$.

To show that $r$ is $\tau$-pollable, proceed as in the proof of Lem.~\ref{lem:preservation positive},
until the appearance of the set $M^K$. In case $M^K \subseteq K^+\setminus K^\Lambda$ we proceed
exactly as in the proof of Lem.~\ref{lem:preservation positive}, using that $r'$ is $\tau$-pollable
w.r.t.\ $\Lambda$ and $P$. So now assume that $\mu \in M^K \cap K^\Lambda$.
Consider a premise \nietplat{$\lambda = (x \trans\beta y)$} in $M_\mu$.
By the decency of $r$ (cf.~Lem.~\ref{lem:preservation_decency}), $x\in\var(t)$.
Using that $t$ is univariate, there is exactly one variable $x_0 \in \var(t')$ with $x \in \var(\sigma(x_0))$.
% and $x$ has exactly one occurrence in $\sigma(x_0)$. 
By the decency of $r_\mu$
(cf.~Lem.~\ref{lem:preservation_decency}), $x\in\var(\sigma(t_\mu))$, and using the decency of $r'$
this implies that $x_0\in\var(t_\mu)$.\vspace{-3pt} Given that $\mu\in K^\Lambda$, the occurrence of $x_0$ in
$t'$ must be $\Lambda$-liquid.
By Lem.~\ref{lem:ruloids}, the proof $\pi_\mu$ must have a
subproof $\pi'$ of a rule $\frac{H'}{\sigma(x_0)\ntrans\gamma}$ or $\frac{H'}{\sigma(x_0) \trans\gamma u''}$
with $(x \trans\beta y) \in H' \subseteq H_\mu$.\vspace{-3pt} By induction, this rule is $\tau$-pollable w.r.t.\ $P$,
so for $M':= M_\mu \cap H'$ there is a rule 
\nietplat{$\frac{H'_{M'}}{\sigma(x_0)\ntrans{\alpha}}$}, resp.\ 
\nietplat{$\frac{H'_{M'}}{\sigma(x_0)\trans{\alpha}u''}$} or \nietplat{$\frac{H'_{M'}}{\sigma(x_0)\trans{\tau}w}$} for some term $w$,
linearly provable from $P$, where $H'_{M'} \subseteq (H'{\setminus}M') \cup M'_\tau$.
Call the premise $\lambda$ \emph{innocent} if in fact there is a rule 
\nietplat{$\frac{H'_{M'}}{\sigma(x_0)\ntrans{\alpha}}$}, resp.\ \nietplat{$\frac{H'_{M'}}{\sigma(x_0)\trans{\alpha}u''}$}.
If all premises in $M_\mu$ are innocent, given the construction of rules required by Def.~\ref{def:pollable},
there is no reason for $\mu$ to be in $M^K$.\vspace{-3pt} So there must be a guilty premise
\nietplat{$\lambda = (x \trans\beta y)$} in $M_\mu$ such that a rule
\nietplat{$\frac{H'_{M'}}{\sigma(x_0)\trans{\tau}w}$} for some term $w$ is linearly provable from $P$.
Since $x$ occurs $\aleph$-liquid in a premise of $H$, its unique occurrence in $t$ must also be
$\aleph$-liquid, using condition~\ref{aleph} of Def.~\ref{def:rooted_branching_bisimulation_safe}.
This implies that the unique occurrence of $x_0$ in $t'$ is $\aL$-liquid.
Hence an $\aL$-patient rule \nietplat{$\frac{x_0 \trans\tau z_0}{t' \trans\tau w'}$}, where $w'$ is $t'$
with $x_0$ replaced by $z_0$, is linearly provable from $P$.
Let $\sigma'(z_0)=w$ and $\sigma'$ coincides with $\sigma$ on all other variables.
Lem.~\ref{linear proof composition} yields a linear proof, which uses \plat{$\frac{x_0 \trans\tau z_0}{t' \trans\tau w'}$}
and $\sigma'$ at the bottom, of the rule \nietplat{$\frac{H'_{M'}}{t\trans{\tau}\sigma'(w')}$}, where
$H'_{M'} \subseteq  (H'{\setminus}M') \cup M'_\tau \subseteq  (H_\mu{\setminus}M_\mu) \cup (M_\mu)_\tau
\subseteq  (H{\setminus}M) \cup M_\tau$, as required.

To show that $r$ is {\pdr}, and satisfies (*),
assume that $r$ has the form \nietplat{$r=\frac{H}{t\trans\alpha u}$},
and \nietplat{$r'=\frac{K}{t'\trans\alpha u'}$}. As in the proof of Lem.~\ref{lem:preservation positive},
let $K^d \subseteq K$ and, for each $\mu\in K$, let $H^d_\mu \subseteq H^+_\mu$ be the finite sets of
positive premises that exist by Def.~\ref{def:positive}. This time, take
$H^d := \bigcup _{\mu \in K^d \cup (K^\Lambda)^+} H_\mu^d$.
The requirement of Def.~\ref{def:positive} is established exactly as in the proof of 
Lem.~\ref{lem:preservation positive}, but substituting $K^d \cup (K^\Lambda)^+$ for $K^d$.
By construction (and induction) it follows that $r$ satisfies (*) w.r.t.\ $\pi$.

It remains to show that any literal in \nietplat{$H^d$}, say of the form $x\trans\beta y$, is a delayable premise of $r$.
There is a $\mu_0$ of the form $t_0\trans\gamma y_0$ in $K^d \cup K^\Lambda$ with $x\trans\beta y$ in $H_{\mu_0}^d$.
The case that $\mu_0 \in K^d$ proceeds exactly as in the proof of
Lem.~\ref{lem:preservation positive}. So let $\mu_0\in K^\Lambda$.
By the decency of $r$ (cf.~Lem.~\ref{lem:preservation_decency}), $x\in\var(t)$.
Using that $t$ is univariate, there is exactly one variable $x_0 \in \var(t')$ with $x \in \var(\sigma(x_0))$.
By the decency of $r_{\mu_0}$
(cf.~Lem.~\ref{lem:preservation_decency}), $x\in\var(\sigma(t_{\mu_0}))$, and using the decency of $r'$
this implies that $x_0\in\var(t_{\mu_0})$. Given that $\mu_0\in K^\Lambda$, the occurrence of $x_0$ in
$t'$ must be $\Lambda$-liquid.

By Lem.~\ref{lem:ruloids} there is a decent xyntt rule $r''=\frac{G}{t' \trans\alpha u''}$ with
$P\irr r''$, and a substitution $\sigma'$ with $\sigma'(t')=\sigma(t')=t$ and $\sigma'(u'')=\sigma(u')=u$,
such that $H$ can be written as $\biguplus_{\nu\in G} H_\nu$ with $P\pvb \frac{H_\nu}{\sigma'(\nu)}$
for all $\nu\in G$. Moreover, proofs $\pi_\nu$ of \nietplat{$\frac{H_\nu}{\sigma'(\nu)}$} from $P$ occur
as strict subproofs of $\pi$.\vspace{-3pt}
Let $\nu_0$ be the unique literal in $G$ be such that \nietplat{$x\trans\beta y$} occurs in $H_{\nu_0}$.
By (*), $\nu_0$ is positive---so has the form \nietplat{$x_0 \trans\delta y'_0$}---and
\nietplat{$x\trans\beta y$} is a delayable premise of \nietplat{$\frac{H_{\nu_0}}{\sigma'(\nu_0)}$}.
Hence, there are rules $\frac{H^1_{\nu_0}}{\sigma'(x_0)\trans{\tau}v}$ and
$\frac{H^2_{\nu_0}}{v\trans{\delta}\sigma(y'_0)}$,
linearly provable from $P$, with $H^1_{\nu_0}\subseteq (H_{\nu_0}{\setminus}\{x \trans{\beta} y\})\cup\{x \trans{\tau} z\}$
and $H^2_{\nu_0}\subseteq (H_{\nu_0}{\setminus}\{x \trans{\beta} y\})\cup\{z \trans{\beta} y\}$
for some term $v$ and fresh variable $z$.

Since $x$ occurs $\aleph$-liquid in a premise of $H$, its unique occurrence in $t$ must also be
$\aleph$-liquid, using condition~\ref{aleph} of Def.~\ref{def:rooted_branching_bisimulation_safe}.
This implies that the unique occurrence of $x_0$ in $t'$ is $\aL$-liquid.
Hence an $\aL$-patient rule \nietplat{$\frac{x_0 \trans\tau z_0}{t' \trans\tau w'}$}, where $w'$ is $t'$
with $x_0$ replaced by $z_0$, is linearly provable from $P$.
Let $\sigma''(z_0)=v$ and $\sigma''$ coincides with $\sigma$ on all other variables.
As $z_0$ is chosen fresh, $\sigma''(x_0)=\sigma'(x_0)$,
$\sigma''(t')=\sigma'(t')=t$ and $\sigma''(u')=\sigma'(u')=u$.
Lem.~\ref{linear proof composition} yields a linear proof, which uses $\frac{x_0 \trans\tau z_0}{t' \trans\tau w'}$\vspace{-3pt}
and $\sigma''$ at the bottom, of the rule $\frac{H^1_{\nu_0}}{t\trans{\tau}\sigma''(w')}$,
with  $H^1_{\nu_0}\subseteq (H_{\nu_0}{\setminus}\{x {\trans{\beta}} y\})\cup\{x {\trans{\tau}} z\} \subseteq
H{\setminus}\{x {\trans{\beta}} y\})\cup\{x \trans{\tau} z\}$.

By Lem.~\ref{lem:preservation_smooth}, the rule $r''$
satisfies condition \ref{smooth} of Def.~\ref{def:smooth} w.r.t.\ $\aleph$ and $\Lambda$.
So the occurrences of $x_0$ in $t'$ and in $\nu_0$ are the only two occurrences of $x_0$ in $r''$.
Replacing these occurrences of $x_0$ in $r''$ by $z_0$ produces a rule $r'''$ with source $w'$
and a premise $z_0 \trans\delta y'_0$.
Let $L$ denote $\bigcup_{\nu\in G{\setminus}\{\nu_0\}}H_\nu \cup H^2_{\nu_0}$.\vspace{-3pt}
Lem.~\ref{linear proof composition} yields a linear proof, which uses $r'''$
and $\sigma''$ at the bottom, of the rule $\frac{L}{\sigma''(w')\trans{\alpha}u}$,
with $L \subseteq (H{\setminus}\{x \trans{\beta} y\})\cup\{x \trans{\tau} z\}$.

We have now finished the case that $t$ is univariate.
Suppose an xyntt rule $r^\ddagger=\frac{H^\ddagger}{t^\ddagger{\ntrans\alpha}}$
[resp.\ $\frac{H^\ddagger}{t^\ddagger \trans\alpha u^\ddagger}$],
with $t^\ddagger$ not univariate, is linearly provable from $P$ by means of a proof $\pi^\ddagger$.
By Lem.~\ref{lem:univariate} there is an xyntt rule
$r=\frac{H}{t{\ntrans\alpha}}$ [resp.\ $\frac{H}{t \trans\alpha u}$]
with $t$ univariate, such that $r^\ddagger=\rho(r)$, using a substitution
$\rho:\textrm{dom}(t)\rightarrow V$. From the (trivial) proof of Lem.~\ref{lem:univariate}
we learn that $r$ has a proof $\pi$ that has the same skeleton as $\pi$;
in fact, $\pi^\ddagger = \rho(\pi)$. Hence, the induction hypothesis applies to strict subproofs of $\pi$
just as much as to strict subproofs of $\pi^\ddagger$.

We have shown above that $r$ is {\pdr} as well as $\tau$-pollable w.r.t.\ $P$, and also
satisfies (*) w.r.t.\ the proof $\pi$. Now Lem.~\ref{lem:substitution} says that also $r^\ddagger$ is {\pdr} w.r.t.\ $P$.
In the very same way one shows that $r^\ddagger$ is $\tau$-pollable w.r.t.\ $P$.
Requirement (*) w.r.t.\ $\pi$ says that if a premise from $H^d$ occurs in a rule $r^\dagger$
obtained by a subproof $\pi^\dagger$ of $\pi$, then this premise is a delayable premise of
$r^\dagger$ as well. Any subproof of $\pi^\ddagger$ can be obtained as $\rho(\pi^\dagger)$ with
$\pi^\dagger$ a subproof of $\pi$; the rule proven by $\rho(\pi^\dagger)$ is $\rho(r^\dagger)$.
Again, $\rho$ renames variables occurring in the source of $r^\dagger$, but leaves other variables
occurring in $r^\dagger$ alone. Hence Lem.~\ref{lem:substitution} applies to conclude that also
$r^\ddagger$ satisfies (*) w.r.t.\ $\pi^\ddagger$.
\qed
\end{proof}

\noindent
The following example shows the essentiality of limiting Lem.~\ref{lem:preservation positive} and Lem.~\ref{lem:preservation positive-Lambda}
to \emph{linearly} provable rules. In this example two delayable positive premises in a linear ruloid
are collapsed to a single premise that is not delayable in the resulting non-linear ruloid.
\begin{example}\label{ex:linearity}
Let $f,g$ be unary, $kb,kc$ binary, and $h$ a ternary function symbol. Let $A=\{a,b,c,d\}$,
and consider the TSS with the following rules:
\[
\frac{x\trans a y}{f(x)\trans b x} \qquad \frac{x\trans a y}{f(x)\trans c x} \qquad \frac{x\trans \tau y}{f(x)\trans \tau kb(x,y)} \qquad \frac{x\trans \tau y}{f(x)\trans \tau kc(x,y)}
\]
\[
\frac{x\trans b y \quad x\trans c z}{g(x)\trans d x} \qquad \frac{x\trans \tau y}{g(x)\trans \tau h(x,y,x)} \qquad \frac{x\trans \tau y}{g(x)\trans \tau h(x,x,y)}
\]
\[
\frac{x_2\trans a y}{kb(x_1,x_2)\trans b x_1} \qquad \frac{x_2\trans \tau y}{kb(x_1,x_2)\trans \tau kb(x_1,y)}
\]
\[
\frac{x_2\trans a y}{kc(x_1,x_2)\trans c x_1} \qquad \frac{x_2\trans \tau y}{kc(x_1,x_2)\trans \tau kc(x_1,y)}
\]
\[
\frac{x_2\trans b y \quad x_3\trans c z}{h(x_1,x_2,x_3)\trans d x_1}
\]
\[
\frac{x_2\trans \tau y}{h(x_1,x_2,x_3)\trans \tau h(x_1,y,x_3)} \qquad \frac{x_3\trans \tau y}{h(x_1,x_2,x_3)\trans \tau h(x_1,x_2,y)}\vspace*{1.5mm}
\]
This TSS is positive and in xynft format.
The argument of $f$ and of $g$ are taken to be $\aleph$-liquid and $\Lambda$-frozen.
The first argument of $kb$, $kc$ and $h$ are $\aleph$-frozen and $\Lambda$-frozen,
while their other arguments are $\aleph$-liquid and $\Lambda$-liquid.
The TSS is $\aL$-patient, since there are patience rules for the latter four arguments.
The rules are moreover rooted $\eta$-bisimulation safe and satisfy condition \ref{smooth} of Def.~\ref{def:smooth} w.r.t.\ $\aleph$ and $\Lambda$.

The TSS is manifestly delay resistant w.r.t.\ $\Lambda$. To prove this, we only need to consider the
rules for $f$ and $g$.
That the premise of the first rule for $f$ is manifestly delay resistant follows by the rules
$\frac{x\trans \tau z}{f(x)\trans \tau kb(x,z)}$ and $\frac{z\trans a y}{kb(x,z)\trans b x}$.
Likewise for the second rule for $f$.\vspace{-3pt}
That the premise of the third rule for $f$ is manifestly delay resistant follows by the rules
$\frac{x\trans \tau z}{f(x)\trans \tau kb(x,z)}$ and $\frac{z\trans \tau y}{kb(x,z)\trans \tau kb(x,y)}$.
Likewise for the fourth rule for $f$.
Furthermore, that the first premise of the first rule for $g$ is manifestly delay resistant follows
by the rules $\frac{x\trans \tau z}{g(x)\trans \tau h(x,z,x)}$ and $\frac{z\trans b y\quad x\trans c y'}{h(x,z,x)\trans d x}$.\vspace{-3pt}
Likewise for the second premise of this rule.
That the premise of the second rule for $g$ is manifestly delay resistant follows by
$\frac{x\trans \tau z}{g(x)\trans \tau h(x,z,x)}$ and $\frac{z\trans \tau y}{h(x,z,x)\trans d h(x,y,x)}$.
Likewise for the third rule for $g$.
Finally, that the TSS is manifestly delay resistant (without taking into account $\Lambda$) now follows
easily by employing the four patience rules in the TSS\@.

The non-linear ruloid $\frac{x\trans a y}{g(f(x))\trans d x}$ is obtained by substituting $y$ for $z$ in the linear ruloid $\frac{x\trans a y\quad x\trans a z}{g(f(x))\trans d x}$.
We argue that this non-linear ruloid is not delay resistant. Its premise $x\trans a y$ is not $\tau$-pollable, as clearly $\frac{x\trans \tau z}{g(f(x))\trans d x}$ is not a ruloid.\vspace{-3pt}
Suppose, toward a contradiction, that it is delayable. This would mean there exist ruloids $\frac{x\trans \tau z}{g(f(x))\trans \tau v}$ and $\frac{z\trans a y}{v\trans d x}$
for some term $v$ and fresh variable $z$. The first of these two ruloids allows four possibilities for $v$: $h(f(x),t(x,z),f(x))$ or $h(f(x),f(x),t(x,z))$ where
$t(x,z)$ equals either $kb(x,z)$ or $kc(x,z)$. However, for none of these four possibilities does there exist a ruloid $\frac{z\trans a y}{v\trans d x}$.
\end{example}

\begin{proposition}\label{prop:eliminating-Lambda}
Let $P$ be a standard TSS in ready simulation format, in which each transition rule is rooted branching bisimulation
safe and satisfies condition \ref{smooth} of Def.~\ref{def:smooth} w.r.t.\ $\aleph$ and $\Lambda$.
Let moreover $P$ be $\aL$-patient and manifestly {\dr} w.r.t.\ $\Lambda$.
Then the standard TSS $\hat P^\ddagger$ in xynft format, constructed in Secs.~\ref{sec:ruloids}--\ref{sec:linearity},
is manifestly \dr.
\end{proposition}

\begin{proof}
By Def.~\ref{def:manifestly}, the conversion $P^\dagger$ of $P$ to decent ntyft format, defined in
Sec.~\ref{sec:ruloids}, is manifestly delay resistant w.r.t.\ $\Lambda$.
By Lem.~\ref{lem:pollable} it is also manifestly $\tau$-pollable w.r.t.\ $\Lambda$.
Hence the rules of $P^\dagger$ are negative-stable and {\pdr} w.r.t.\ $\Lambda$ and $P^\dagger$, as
well as $\tau$-pollable w.r.t.\ $\Lambda$ and $P^\dagger$.
Clearly, $P^\dagger$ is $\aL$-patient, and all of its rules are rooted branching bisimulation
safe and satisfy condition \ref{smooth} of Def.~\ref{def:smooth} w.r.t.\ $\aleph$ and $\Lambda$
(as already observed in the proof of Cor.~\ref{cor:smooth}).

The TSS $\hat P^\ddagger$ contains all xynft rules linearly provable from $P^\dagger$.
By Lem.~\ref{lem:preservation stable} those rules are negative-stable,
and by Lem.~\ref{lem:preservation positive-Lambda} they are {\pdr} and
$\tau$-pollable w.r.t.\ $P^\dagger$. The definitions of {\pdr} and
$\tau$-pollability of rules from $\hat P^\ddagger$ w.r.t.\ $P^\dagger$ require the existence of certain
xynft rules that are linearly provable from $P^\dagger$. By definition, these are rules of $\hat P^\ddagger$.
Hence $\hat P^\ddagger$ is manifestly delay resistant, as well as manifestly $\tau$-pollable.
\qed
\end{proof}
\newpage

\subsection[Transfer of tau-pollability to non-standard rules]{Transfer of $\tau$-pollability to non-standard rules}

\begin{lemma}\label{lem:transfer}
Let $P$ be a standard TSS in ready simulation format. If the TSS $\hat P^\ddagger$
is manifestly $\tau$-pollable and negative-stable, then its augmentation $\hat P^+$ with non-standard rules is
manifestly $\tau$-pollable.
\end{lemma}

\begin{proof}
Let \nietplat{$\frac{H}{t\ntrans\alpha}$} be a rule of $\hat P^+$,
and let $M\subseteq H^+$.\vspace{-3pt} It suffices to show that also \nietplat{$\frac{H_M}{t\ntrans\alpha}$} is a rule of $\hat P^+$,
where $H_M := (H{\setminus}M) \cup M_\tau$.
Let $\hat P^\ddagger {\upharpoonright} (t{\trans\alpha})$ denote the set of rules of $\hat P^\ddagger$ with a conclusion of the form
$t\trans\alpha u$ for some $u$. By the construction of $\hat P^+$
there exists a surjective function $f$ from
\nietplat{$\hat P^\ddagger {\upharpoonright} (t{\trans\alpha})$} to $H$ such that each rule
\nietplat{$r \in \hat P^\ddagger {\upharpoonright} (t{\trans\alpha})$} contains a premise that denies the premise $f(r)$.
It suffices to find a surjective function $f_M$ from \nietplat{$\hat P^\ddagger {\upharpoonright} (t{\trans\alpha})$}
to $H_M$ such that each rule
\nietplat{$r \in \hat P^\ddagger {\upharpoonright} (t{\trans\alpha})$} contains a premise that denies the premise $f_M(r)$.
In case $f(r) \in (H\setminus M)$, we take $f_M(r):=f(r)$.
In case $f(r) = (w \trans\beta y)\in M$, we take $f_M(r) :=  (w \trans\tau z_y)\in M_\tau$.
As $r$ must contain the premise $w{\ntrans\beta}$, it also contains the premise \nietplat{$w{\ntrans\tau}$},
using that $\hat P^\ddagger$ is negative-stable. This premise denies $f_M(r)$.
Surjectivity is guaranteed by construction.
\qed
\end{proof}

\subsection{Lifting delay resistance to ruloids}

\begin{lemma}\label{lem:tau-negative}
Let $P$ be a standard TSS in ready simulation format, such that the TSS $\hat P^\ddagger$
is manifestly $\tau$-pollable.
If the TSS $\hat P^+$ contains rules \plat{$\frac{H}{t\ntrans\alpha}$} and \plat{$\frac{H^\tau}{t\ntrans\tau}$},
then it also contains a rule \plat{$\frac{H'}{t\ntrans\alpha}$}, with $H'\mathbin\subseteq (H{\cup} H^\tau)^+
\cup (H {\cup} H^\tau)^{s-}\!$.
\end{lemma}

\begin{proof}
Suppose $\hat P^+$ contains rules $\frac{H}{t\ntrans\alpha}$ and $\frac{H^\tau}{t\ntrans\tau}$.
Then there exists a surjective function $f$ from\vspace{-3pt}\linebreak
{$\hat P^\ddagger {\upharpoonright} (t{\trans\alpha}) \cup \hat P^\ddagger {\upharpoonright} (t{\trans\tau})$}
to $H \cup H^\tau$ such that each rule $r \in \textrm{dom}(f)$ contains a premise that denies the premise $f(r)$.
It suffices to construct a function $h$ from \nietplat{$\hat P^\ddagger {\upharpoonright} (t{\trans\alpha})$}
to $(H{\cup} H^\tau)^+ \cup (H {\cup} H^\tau)^{s-}$ such that each rule $r \in \textrm{dom}(h)$ contains a
premise that denies the premise $h(r)$.

Let \nietplat{$r = \frac{K}{t \trans\alpha u}\in \hat P^\ddagger {\upharpoonright} (t{\trans\alpha})$}.
Let $M$ be the collection premises \nietplat{$v \trans\beta y$} in
$K$ for which $v {\ntrans\tau}$ does not occur in $H \cup H^\tau$.
Since $\hat P^\ddagger$ is manifestly $\tau$-pollable, there must be a rule
$r_M=\frac{K_M}{t \trans\alpha u}$ or $\frac{K_M}{t \trans\tau w}$ in $\hat P^\ddagger$
with $K_M\subseteq (K{\setminus}M) \cup M_\tau$.
So a premise $\mu \in K_M$ denies the premise $f(r_M) \in H\cup H^\tau$.
Given the definition of $M$, this premise $\mu$ cannot occur in $M_\tau$.
Hence $\mu\in K\setminus M$ and $f(r_M)\in (H{\cup} H^\tau)^+\cup (H {\cup} H^\tau)^{s-}$. Define $h(r):=f(r_M)$.
\qed
\end{proof}

\begin{lemma}\label{lem:negative}
\label{lem:preservation_delay_bisimulation_safe}
Let $P$ be a standard TSS in ready simulation format, such that the TSS $\hat P^\ddagger$
is manifestly $\tau$-pollable and negative-stable.
\begin{enumerate}
\item[(1)] For any linear $P$-ruloid \nietplat{$\frac{H}{t \trans\alpha u}$} there is a 
linear $P$-ruloid \nietplat{$\frac{H'}{t \trans\alpha u}$} with $H' \subseteq H^+\cup H^{s-}$.
\vspace{-2pt}
\item[(2)] For any linear $P$-ruloids \nietplat{$\frac{H}{t{\ntrans\alpha}}$} and \nietplat{$\frac{H^\tau}{t{\ntrans\tau}}$},
with $t$ not a variable,
there is a linear $P$-ruloid {$\frac{H'}{t{\ntrans\alpha}}$} with
$H'\mathbin\subseteq (H{\cup} H^\tau)^+\linebreak[1] \cup (H {\cup} H^\tau)^{s-}\!$.
\end{enumerate}
\end{lemma}
\begin{proof}
We prove the claims by simultaneous structural induction on the linear proofs of the ruloids from
$\hat P^+\!$.

First consider the case that $t$ is a variable $x$. Then any standard $P$-ruloid with source $t$ must have
the form \nietplat{$\frac{x \trans\alpha y}{x \trans\alpha y}$}; it trivially satisfies (1).

Now let $t \notin V$.
First let $\pi$ be a linear proof of \plat{$\frac{H}{t\trans\alpha u}$}.
Let the decent ntyft rule $r$ in $\hat P^+$ (and hence in $\hat P^\ddagger$) of the form
\[
\frac{\{t_k\trans{\beta_k}y_k\mid k\in K\}\cup\{t_\ell\ntrans{\gamma_\ell\;}\,\mid \ell\in L\}\cup\{x_j\ntrans{\gamma_j\;}\,\mid j\in J\}}{f(x_1,\ldots,x_{\ar(f)})\trans{\alpha}v}
\]
and the substitution $\sigma$ be used at the bottom of $\pi$.
Here we require $\sigma(t_\ell)\notin V$ for all $\ell \in L$ and $\sigma(x_j)\in V$; so we have split the negative premises in two
sets, depending on whether after application of $\sigma$ their left-hand sides are variables.
We have $\sigma(f(x_1,\ldots,x_{\ar(f)}))\mathbin=t$ and $\sigma(v)\mathbin=u$.
Moreover, rules $r_k \mathbin= \frac{H_k}{\sigma(t_k)\trans{\beta_k}\sigma(y_k)}$ for each $k\mathbin\in K$ and
$r_\ell \mathbin= \frac{H_\ell}{\sigma(t_\ell)\ntrans{\gamma_\ell}}$ for each $\ell\mathbin\in L$ are linearly provable
from $\hat P^+$ by means of strict subproofs $\pi_k$ and $\pi_\ell$ of $\pi$,
where $H=\bigcup_{k\in K}H_k\cup\bigcup_{\ell\in L}H_\ell \cup \{\sigma(x_j){\ntrans{\gamma_j\;}} \mid j\mathbin\in J\}$.
Here we use that $\hat P^+$ is in ntyft format.
Using that $\hat P^\ddagger$ is negative-stable, $\{\sigma(x_j){\ntrans{\gamma_j\;}} \mid j\mathbin\in J\}
\subseteq H^{s-}$; moreover, for each $\ell\mathbin\in L$ there is a unique
$\ell'\mathbin\in L$ with $t_{\ell'}=t_\ell$ and $\gamma_{\ell'}=\tau$; we write $H^\tau_\ell := H_{\ell'}$.

As $r$ is decent, $\var(t_k) \subseteq \{x_1,\ldots,x_{\ar(f)}\}$, so $\var(\sigma(t_k)) \subseteq
\var(t)$ for each $k\in K$. Likewise,\linebreak $\var(\sigma(t_\ell))\subseteq\var(t)$ for each $\ell\in
L$. From ${\it rhs}(H) \cap \var(t) = \emptyset$ it follows that ${\it rhs}(H_k) \cap
\var(\sigma(t_k)) = \emptyset$ for each $k\mathbin\in K$, and ${\it
  rhs}(H_\ell)\cap\var(\sigma(t_\ell))=\emptyset$ for each $\ell\mathbin\in L$. So for each $k\mathbin\in K$ and
$\ell\mathbin\in L$, the rules $r_k$ and $r_\ell$ are nxytt rules. By Lem.~\ref{lem:preservation_decency},
they are decent, and thus ruloids.

By induction there is a $P$-ruloid $\frac{H'_k}{\sigma(t_k)\trans{\beta_k}\sigma(y_k)}$
with $H'_k \subseteq H_k^+\cup H_k^{s-}$, for each $k\mathbin\in K$.
Likewise, there is a $P$-ruloid $\frac{H'_\ell}{\sigma(t_\ell){\ntrans{\gamma_\ell}}}$
with $H'_\ell\mathbin\subseteq (H_\ell{\cup} H^\tau_\ell)^+ \cup (H_\ell {\cup} H^\tau_\ell)^{s-}\!$,
for each $\ell\mathbin\in L$.
By composition of proofs, we obtain a linear ruloid $\frac{H'}{t \trans\alpha u}$
with \nietplat{$H'=\bigcup_{k\in K}H'_k\cup\bigcup_{\ell\in L}H'_\ell  \cup \{\sigma(x_j){\ntrans{\gamma_j\;}} \mid j\mathbin\in J\}
\subseteq H^+ \cup H^{s-}$}.

Secondly, let $\pi$ be a linear proof of \plat{$\frac{H}{t{\ntrans\alpha}}$}, and $\pi^\tau$ a
linear proof of \plat{$\frac{H^\tau}{t{\ntrans\tau}}$}.
Let the decent ntyft rule $r$ in $\hat P^+$ (and hence in $\hat P^\ddagger$) of the form
\[
\frac{\{t_k\trans{\beta_k}y_k\mid k\in K\}\cup\{t_\ell\ntrans{\gamma_\ell\;}\,\mid \ell\in L\}\cup\{x_j\ntrans{\gamma_j\;}\,\mid j\in J\}}{f(x_1,\ldots,x_{\ar(f)}){\ntrans{\alpha}}}
\]
and the substitution $\sigma$ be used at the bottom of $\pi$, again with $\sigma(t_\ell)\mathbin{\notin} V$ for $\ell \mathbin\in L$
and $\sigma(x_j)\mathbin\in V$ for $j\mathbin\in J$.
Likewise, let the decent ntyft rule $r^\tau$ in $\hat P^+$ of the form
\[
\frac{\{t_k\trans{\beta_k}y_k\mid k\in K^\tau\}\cup\{t_\ell\ntrans{\gamma_\ell\;}\,\mid \ell\in L^\tau\}\cup\{x_j\ntrans{\gamma_j\;}\,\mid j\in J^\tau\}}{f(x_1,\ldots,x_{\ar(f)}){\ntrans{\tau}}}
\]
and the substitution $\sigma^\tau$ be used at the bottom of $\pi^\tau$, with
$\sigma^\tau(t_\ell)\mathbin{\notin} V$ for $\ell \mathbin\in L$ and $\sigma^\tau(x_j)\mathbin{\in} V$ for $j \mathbin\in J$.
Note that $\sigma(x_i)\mathbin=\sigma^\tau(x_i)$ for $i\mathbin=1,\dots,\ar(f)$. By choosing the right-hand sides of
premises in $r^\tau$ different from the ones in $r$, we may assume, w.l.o.g., that $\sigma^\tau=\sigma$.

Exactly as above, $P$-ruloids $r_k \mathbin= \frac{H_k}{\sigma(t_k)\trans{\beta_k}\sigma(y_k)}$
for each $k\mathbin\in K \cup K^\tau$ and
$r_\ell \mathbin= \frac{H_\ell}{\sigma(t_\ell)\ntrans{\gamma_\ell}}$
for each $\ell\mathbin\in L\cup L^\tau$ are linearly provable
from $\hat P^+$ by means of strict subproofs $\pi_k$ and $\pi_\ell$ of $\pi$ or $\pi^\tau$,
where $H=\bigcup_{k\in K}H_k\cup\bigcup_{\ell\in L}H_\ell \cup \{\sigma(x_j){\ntrans{\gamma_j\;}} \mid j\mathbin\in J\}$
and $H^\tau=\bigcup_{k\in K^\tau}H_k\cup\bigcup_{\ell\in L^\tau}H_\ell \cup \{\sigma(x_j){\ntrans{\gamma_j\;}} \mid j\mathbin\in J^\tau\}$.

By Lem.~\ref{lem:tau-negative} there is a rule in $\hat P^+$ of the form
\[
\frac{\{t_k\trans{\beta_k}y_k\mid k\in K'\}\cup\{t_\ell\ntrans{\gamma_\ell\;}\,\mid \ell\in L'\}\cup\{x_j\ntrans{\gamma_j\;}\,\mid j\in J'\}}{f(x_1,\ldots,x_{\ar(f)}){\ntrans{\alpha}}}
\]
with $K' \subseteq K \cup K^\tau$,  $L' \subseteq L \cup L^\tau$ and  $J' \subseteq J \cup J^\tau$,
and such that (1) for any $\ell \in L'$ there is an $\ell'\in L \cup L^\tau$ with $t_{\ell'}=t_\ell$
and $\gamma_{\ell'}=\tau$, and (2) for any $j \in J'$ there is an $j'\in J \cup J^\tau$ with $t_{j'}=t_j$
and $\gamma_{j'}=\tau$; making an arbitrary choice for $\ell'$ in case of ambiguity, let $H^\tau_\ell:=H_{\ell'}$.

By induction there is a $P$-ruloid $\frac{H'_k}{\sigma(t_k)\trans{\beta_k}\sigma(y_k)}$
with $H'_k \subseteq H_k^+\cup H_k^{s-}$, for each $k\mathbin\in K'$.
Likewise, there is a $P$-ruloid $\frac{H'_\ell}{\sigma(t_\ell){\ntrans{\gamma_\ell}}}$
with $H'_\ell\mathbin\subseteq (H_\ell{\cup} H^\tau_\ell)^+ \cup (H_\ell {\cup} H^\tau_\ell)^{s-}\!$,
for each $\ell\mathbin\in L'$.
By composition of proofs, we obtain a linear ruloid $\frac{H'}{t {\ntrans\alpha}}$ with
\nietplat{$H'=\bigcup_{k\in K'}H'_k\cup\bigcup_{\ell\in L'}H'_\ell \cup \{\sigma(x_j){\ntrans{\gamma_j\;}} \mid j\mathbin\in J'\}
\subseteq (H{\cup} H^\tau)^+\linebreak[1] \cup (H {\cup} H^\tau)^{s-}\!$}.$\!$\qed
\end{proof}

\begin{trivlist}
\item [\hspace{\labelsep}{\it Proof of Thm.~\ref{thm:manifest}.}]
Let $P$ be a manifestly delay resistant standard TSS in ready simulation format.
By Prop.~\ref{prop:manifest preservation}, the TSS $\hat P^\ddagger$, constructed in
Sec.~\ref{sec:ruloids}, is manifestly delay resistant, and thus negative-stable.
By Lem.~\ref{lem:pollable}, it is manifestly $\tau$-pollable.

By Lem.~\ref{lem:transfer} the TSS $\hat P^+$ is manifestly $\tau$-pollable.
So by definition the rules of $\hat P^+$ are {\pdr} as well as $\tau$-pollable
w.r.t.\ $\hat P^+$. By Lem.~\ref{lem:preservation positive}, each nxytt rule linearly provable
from $\hat P^+$, i.e.\ each linear $P$-ruloid, is {\pdr} w.r.t.\ $\hat P^+$.
By Lem.~\ref{lem:negative} all linear $P$-ruloids are {\ndr} w.r.t.\ $\hat P^+$.
Thus $P$ is delay resistant.
\qed
\end{trivlist}

\begin{trivlist}
\item [\hspace{\labelsep}{\it Proof of Thm.~\ref{thm:manifest-Lambda}.}]
Let $P$ be a TSS meeting the conditions of Thm.~\ref{thm:manifest-Lambda}.
By Prop.~\ref{prop:eliminating-Lambda}, the TSS $\hat P^\ddagger$ is manifestly delay resistant.
From here, the argument proceeds as above.
\qed
\end{trivlist}

\end{document}